\documentclass[12pt]{amsart}
\usepackage{amsmath, amsthm, amssymb,cite,enumitem}
\usepackage{fullpage}
\usepackage[english]{babel}
\usepackage{framed}
\usepackage[utf8]{inputenc}
\usepackage{fontenc}
\usepackage{amsmath}
\usepackage{amsfonts}
\usepackage{amssymb}
\usepackage{calligra}
\usepackage{calrsfs}
\usepackage{yfonts}
\usepackage[mathscr]{euscript}
\usepackage[toc,page]{appendix}
\usepackage{bm}

\usepackage{xcolor}
\usepackage[citebordercolor={green}]{hyperref}

\newcommand{\bma}{{\bm{a}}}
\newcommand{\bmb}{{\bm{b}}}

\newcommand{\bmi}{{\bm{i}}}
\newcommand{\bmj}{{\bm{j}}}
\newcommand{\bmk}{{\bm{k}}}
\newcommand{\bmalph}{{\bm{\alpha}}}
\newcommand{\bmbeta}{{\bm{\beta}}}
\newcommand{\bmmu}{{\bm{\mu}}}
\newcommand{\bmnu}{{\bm{\nu}}}

\usepackage{slashed}
\usepackage{esint} 

\usepackage{mathtools}

\DeclarePairedDelimiter\floor{\lfloor}{\rfloor}

\newtheorem{theorem}{Theorem}

\newtheorem{lemma}[theorem]{Lemma}
\newtheorem{proposition}[theorem]{Proposition}
\newtheorem{corollary}[theorem]{Corollary}

\numberwithin{theorem}{section}
\theoremstyle{remark}
\newtheorem{remark}[theorem]{Remark}

\numberwithin{equation}{section}

\usepackage{tikz}
\usepackage{pgfplots}
\pgfplotsset{compat=1.16}
\usetikzlibrary{intersections, pgfplots.fillbetween, patterns}

\newcommand{\magenta}{\color{magenta}}

\newcommand{\enint}{E^{\text{i}}}
\newcommand{\enext}{E^{\text{e}, \kappa}}

\newcommand{\dint}{\mathcal{D}^{\text{i}}}
\newcommand{\dext}{\mathcal{D}^{\text{e}}}

\newcommand{\Sext}{\Sigma^{\text{e}}}

\newcommand{\hin}{\mathcal{H}}

\newcommand{\op}{\overline{\partial}}
\newcommand{\gd}{\slashed {\partial}}

\newcommand{\hff}{h^{1,\flat}}
\newcommand{\hnn}{h^{1,\natural}}
\newcommand{\hab}{h^1_{\alpha\beta}}
\newcommand{\habf}{h^{1, \flat}_{\alpha\beta}}
\newcommand{\habn}{h^{1, \natural}_{\alpha\beta}}

\newcommand{\pb}{\underline{\partial}}

\newcommand{\Lb}{\underline{L}}

\renewcommand{\L}{\underline{L}}

\newcommand{\q}[1]{\mathcal{#1}}
\newcommand{\m}[1]{\mathbb{#1}}

\newcommand{\scal}{\mathcal{S}}
\newcommand{\hcal}{\mathcal{H}}
\newcommand{\ep}{\epsilon}

\newcommand{\R}{{\mathbb R}}
\renewcommand{\S}{{\mathbb S}}
\newcommand{\N}{{\mathbb N}}

\newcommand{\Wf}{\mathbf{W}}
\newcommand{\Ff}{\mathbf{F}}

\newcommand{\tr}{\text{tr}}

\newcommand{\ub}{{\underline{u}}}

\newcommand{\Circle}{{\mathbb{S}^1}}

\begin{document}
	
	\title{The global stability of the Kaluza--Klein spacetime}	
	
	\author{C\'ecile Huneau}
	\address{\'Ecole Polytechnique \& CNRS}
	\email{cecile.huneau@polytechnique.edu}
	
	\author{Annalaura Stingo}
	\address{\'Ecole Polytechnique}
	\email{annalaura.stingo@polytechnique.edu}
	
	\author{Zoe Wyatt}
	\address{King's College London}
	\email{zoe.wyatt@kcl.ac.uk}

	\begin{abstract} In this paper we show the classical global stability of the flat Kaluza--Klein spacetime, which corresponds to Minkowski spacetime in $\m R^{1+4}$ with one direction compactified on a circle.  We consider small perturbations which are allowed to vary in all directions including the compact direction. These perturbations lead to the creation of massless modes and Klein--Gordon modes. On the analytic side, this leads to a PDE system coupling wave equations to an infinite sequence of Klein--Gordon equations with different masses. The techniques we use are based purely in physical space using the vectorfield method.
	\end{abstract}
\maketitle

	\section{Introduction}
	
	The goal of the present article is to prove the global stability of the Kaluza--Klein spacetime for the Einstein vacuum equations
	\begin{equation}\label{EE}
		R_{\mu\nu}[g] = 0
	\end{equation}
	where $R_{\mu\nu}$ denotes the Ricci tensor of an unknown Lorentzian metric $g$. The Kaluza--Klein spacetime is a solution of \eqref{EE} on $\R^{1+3}\times \S^1$ and consists of a Lorentzian metric $\overline{g}$, given in the standard coordinates $(t,x)\in\R^{1+3}$, $y\in \S^1$ by 
	\[
	\overline{g} = -(dt)^2 + \sum_{\bmi=1}^3 (dx^\bmi)^2 + (dy)^2.
	\]
	The Einstein equations in this higher dimensional setting have, as in the standard $3+1$ setting, a well-posed initial value formulation. The data consist of a triplet $(\Sigma_0, g_0, K_0)$ where $\Sigma_0$ is a 4-dimensional manifold diffeomorphic to $\R^3\times\S^1$ equipped with a Riemannian metric $g_0$ and $K_0$ is a symmetric two-tensor. Solving \eqref{EE} with initial data $(\Sigma_0, g_0, K_0)$ means that one looks for a 5-dimensional manifold $\mathcal{M}$ with a Lorentzian metric $g$ satisfying \eqref{EE} and an embedding $\Sigma_0 \hookrightarrow \mathcal{M}$ such that $g_0$ is the pullback of $g$ to $\Sigma_0$ and $K_0$ is the second fundamental form of $\Sigma_0$. The initial value problem is overdetermined and the data must satisfy the \emph{constraint equations}\footnote{We use the Einstein summation convention over repeated indexes.  Greek indexes run from 0 to 4 while Latin indexes run from 1 to 4. \textbf{Bold} Greek and Latin indexes run up to 3.
		We use the notation $x^0=t$ and $x^4 = y$ so that $\partial_\mu= \partial/\partial x^\mu$ for $\mu =0,\dots, 4$ denotes any derivative along the coordinate axes.}
	\[
	R[g_0] - K_0^{ij}{K_0}_{ij} + {K_0}_i^i{K_0}^j_j=0, \quad \nabla^j {K_0}_{ij} - \nabla_i {K_0}^j_j = 0 
	\]
	where $R[g_0]$ is the scalar curvature of $g_0$ and $\nabla$ is the Levi-Civita connection of $g_0$. These equations simply come from the vanishing of the time components of the Einstein tensor 
	\[R_{0i}=0, \qquad R_{00} - \frac{1}{2}Rg_{00} = 0.\]
	In PDE terminology, the \emph{local well-posedness} of the Einstein equations was proved in the seminal works of Choquet-Bruhat \cite{CB52} and Choquet-Bruhat and Geroch \cite{CBG69}, who show the existence and uniqueness (up to diffeomorphisms) of a maximal globally hyperbolic spacetime arising from any set of smooth initial data satisfying the constraint equations. This is a local result in the sense that it does not guarantee that the spacetime solution $(\mathcal{M}, g)$ is causally geodesically complete. 
	We observe that their proofs, which are performed in a 4-dimensional setting, do not actually depend on the particular manifold $\mathcal{M}$ considered (nor on its dimension, or whether or not it is compact or a product with compact factors) and therefore apply to the Kaluza--Klein setting. We also mention the recent work of the first author with V\^alcu \cite{HV21} in which initial data for the Einstein equations on manifolds of the form $\R^{1+n}\times \mathbb{T}^m$ are constructed.
	
	The articles mentioned above constitute the starting point to investigate and prove the global stability of the flat metric $\overline{g}$. An informal statement of our main result is the following
	\begin{theorem}\label{Thm: informal}
		Let $(\Sigma_0, g_0, K_0)$ be an arbitrary set of smooth asymptotically flat initial data satisfying the constraint equations, with $\Sigma_0\cong \R^3\times \S^1$,
		\[
		\begin{gathered}
			{g_0} = \begin{pmatrix}
				(1+\chi(r)M/r)I_3 & 0 \\
				0 & 1
			\end{pmatrix} + {g_0^1}, \quad (I_3)_{\bmi\bmj}=\delta_{\bmi\bmj} \\
			\text{where } {g_0^1}_{ij} = O(r^{-1-\kappa}), \ {K_0}_{ij} = O(r^{-2-\kappa}) \ \text{as } r=|x|\rightarrow\infty, \ \kappa>0
		\end{gathered}
		\]
		and such that $g_0-\delta$ and $K_0$ satisfy global smallness assumptions. Then, there exists a causally geodesically complete spacetime asymptotically converging to the Kaluza--Klein spacetime. 
	\end{theorem}
	In the above theorem, $\chi$ is a cut-off function supported outside some ball centered at 0 and $M$ is a positive constant corresponding to the ADM mass. We refer to the work of Dai \cite{Dai04} on the positive mass theorem for manifolds including those of Kaluza--Klein type.

	\smallskip
	The global stability problem for the flat metric $\overline{g}$ can be cast into the form of a small data global existence problem for quasilinear wave equations. The Einstein equations can be written as a system of quasilinear wave equations for the unknown metric coefficients $g_{\alpha\beta}$ if one works with a standard gauge, called the \emph{harmonic} or \emph{wave coordinate} or \emph{De Donder} gauge, in which the (harmonic) coordinates $\{x^\alpha\}_{\alpha=0,\dots,4}$ are defined to be solutions of the geometric wave equation\footnote{$g^{\mu\nu}$ and $\overline{g}^{\mu\nu}$ denote respectively the coefficients of the inverse metric of $g$ and $\overline{g}$. Unless differently specified, we lower and raise indexes using the metric $\overline{g}$, i.e. for any tensor $\pi_{\alpha\beta}$ we define $\pi^{\alpha\beta}:=\overline{g}^{\alpha\mu}\overline{g}^{\beta
			\nu}\pi_{\mu\nu}$.} $\Box_g x^\alpha= g^{\mu\nu}\nabla_\mu \nabla_\nu x^\alpha =0$, where $\nabla$ denotes the Levi-Civita connection of $g$. Relative to these coordinates the metric $g$ satisfies the so-called \emph{wave condition}
	\begin{equation}\label{wave_condition}
		g^{\alpha\beta}g_{\nu\mu}\Gamma^\nu_{\alpha\beta} = g^{\alpha\beta}\partial_\beta g_{\alpha\mu} - \frac{1}{2} g^{\alpha\beta}\partial_\mu g_{\alpha\beta} = 0, \quad \mu=0,\dots,4
	\end{equation}
	under which the wave operator $\Box_g$ on functions coincides with the reduced wave operator $\tilde{\Box}_g = g^{\mu\nu}\partial_\mu\partial_\nu$.
	In this gauge the equations \eqref{EE} become
	\begin{equation}\label{g_equations}
		\tilde{\Box}_g g_{\alpha\beta} = \tilde{F}_{\alpha\beta}(g)(\partial g, \partial g) \quad \text{on } \R^{1+3}\times \S^1
	\end{equation}
	where $\tilde{F}_{\alpha\beta}(u)(v,v)$ depends quadratically on $v$. A straightforward computation shows that these source terms decompose into the sum of the following 
	\[
	\begin{aligned}
		\tilde{P}(\partial_\alpha g, \partial_\beta g) &:=\frac{1}{4}g^{\mu\nu}g^{\rho\sigma}\left(\partial_\alpha g_{\mu\nu}\partial_\beta g_{\rho\sigma} - 2 \partial_\alpha g_{\mu\rho}\partial_\beta g_{\nu\sigma}\right)
		\\
		\tilde{Q}_{\alpha\beta}(\partial g, \partial g)&:=
		g^{\mu\nu} g^{\rho\sigma}\partial_\mu g_{\rho\alpha}\partial_\nu g_{\sigma\beta}
		-g^{\mu\nu}g^{\rho\sigma}Q_{\mu\sigma}(\partial g_{\rho\alpha}, \partial g_{\nu\beta})
		+g^{\mu\nu}g^{\rho\sigma}Q_{\alpha\mu}(\partial g_{\nu\sigma}, \partial g_{\rho\beta})
		\\&\quad
		+g^{\mu\nu}g^{\rho\sigma}Q_{\beta\mu}(\partial g_{\nu\sigma}, \partial g_{\rho\alpha})
		+\frac12 g^{\mu\nu}g^{\rho\sigma}Q_{\sigma\alpha}(\partial g_{\mu\nu}, \partial g_{\rho\beta})
		+ \frac12 g^{\mu\nu}g^{\rho\sigma}Q_{\sigma\beta}(\partial g_{\mu\nu}, \partial g_{\rho\alpha})
	\end{aligned}
	\]
	where $Q_{\mu\nu}$ denotes the quadratic null form\footnote{The quadratic form $Q_0(\partial \phi, \partial\psi) = \overline{g}^{\mu\nu}\partial_\mu\phi\partial_\nu\psi$ is also a null form.}
	\[
	Q_{\mu\nu}(\partial \phi, \partial \psi) = \partial_\mu\phi \partial_\nu \psi  - \partial_\nu\phi \partial_\mu \psi.
	\]
	The initial conditions $(g_{\alpha\beta}|_{t=0}, \partial_t g_{\alpha\beta}|_{t=0})$ for \eqref{g_equations} are defined from $(\Sigma_0, g_0, K_0)$ as follows
	\begin{equation}\label{data_g}
		\begin{aligned}
			& g_{ij}|_{t=0} = {g_0}_{ij}, \qquad g_{00}|_{t=0} = -a^2,  \qquad g_{0i}|_{t=0} = g_{i0}|_{t=0} = 0,\\
			&(\partial_t g_{ij})|_{t=0} = -2a {K_0}_{ij}, \qquad  (\partial_t g_{00})|_{t=0} = 2a^3 g_0^{ij}{K_0}_{ij} ,\\
			&  (\partial_t g_{0i})_{t=0} = a^2 g_0^{kl}\partial_l {g_0}_{ki} - \frac{1}{2}a^2 g_0^{kl}\partial_i {g_0}_{kl} - a\partial_i a
		\end{aligned}
	\end{equation}
	where $a^2 := (1 - M\chi(r)r^{-1})$ denotes the lapse function, so that they are compatible with the constraint equations and satisfy the wave condition. In particular the constraint equations yield a decay for $g_{ij}$ of the form
	$$g_{\bmi\bmj} = \big(1 + M\chi(r)r^{-1}\big)\delta_{\bmi\bmj} + O(r^{-1-\kappa}), \qquad g_{44}=1+ O(r^{-1-\kappa})$$
	The initial data for $g_{00}$ and $g_{0i}$ are free and we set them as in \eqref{data_g}, following what was done by Lindblad and Rodnianski in their work \cite{LR10}, for compatibility with the wave coordinates for Schwarzschild.
	The condition $(\partial_t g_{ij})|_{t=0} = -2a {K_0}_{ij}$ is given so that $K_0$ is the second fundamental form of $\Sigma_0$, i.e. $K_0(X,Y) = -g|_{t=0}(\nabla_X \partial_t, Y)$ for any vector fields $X,Y$.
	
	Any solution to the Einstein equations \eqref{EE} with smooth data $(\Sigma_0, g_0, K_0)$ satisfies \eqref{g_equations}-\eqref{data_g} when written in harmonic coordinates. Conversely, any solution $g_{\alpha\beta}$ of \eqref{g_equations}-\eqref{data_g} with initial data compatible with the constraint equations and satisfying the wave condition \eqref{wave_condition} will satisfy \eqref{wave_condition} for all times and hence gives rise to a solution of \eqref{EE} with data $(\Sigma_0, g_0, K_0)$ defined from \eqref{data_g}. We refer to Ringstr\"om \cite{ringstrom} for more details on the subject. From now on, we will then entirely focus on the formulation \eqref{g_equations}-\eqref{data_g}.

	\subsection{State of the art}
	There is a vast literature in general relativity concerning the stability of physical solutions to the Einstein equations. In the 4-dimensional setting, the global stability of the simplest solution, the Minkowski metric, was proved in a monumental work by Chistodoulou and Klainerman \cite{CK93} and later revisited in the works of Lindblad and Rodnianski \cite{LR05, LR10} using the harmonic gauge. See also the results by Klainerman and Nicol\`o \cite{KN03}, Bieri and Zipser \cite{BZ09}, Hintz and Vasy \cite{HintzVasyMinkowski}, Choquet-Bruhat, Chru\'{s}ciel and Loizelet \cite{CBCL06} for Einstein-Maxwell systems and by Speck \cite{Speck14} for Einstein equations coupled to a family of nonlinear electromagnetic field equations. 
	
	Analogous global stability results have also been proved for other 4-dimensional coupled Einstein matter systems. Einsten-Klein--Gordon systems were investigated by LeFloch and Ma \cite{LeFloch-Ma16} in the case of restricted data coinciding with the Schwarzschild metric outside a compact set, and global stability was later proved by Ionescu and Pausader \cite{IP2} in the case of unrestricted data. We also cite the works by
	Fajman, Joudioux and Smulevici \cite{FJS21} and Lindblad and Taylor \cite{LT20} proving a global stability result for Einstein-Vlasov systems for a class of restricted data, and the result by Bigorgne, Fajman, Joudioux, Smulevici and Thaller \cite{BFJST21} about the asymptotic stability of Minkowski spacetime with non-compactly supported massless Vlasov matter.
	There is also a very rich literature concerning the stability of other explicit 4-dimensional solutions to the Einstein equations, for instance the Kerr solution or solutions to the Einstein equations with positive cosmological constants, but it is not our purpose to list such references here.
	
	\smallskip
	Higher dimensional solutions of the Einstein equations, in particular spacetimes with additional compact directions $\R^{1+3}\times \mathcal{K}$, have attracted substantial attention from the theoretical physics community throughout the past century. Theories of higher dimensional gravity are in fact of great interest in supergravity and string theory as possible models for quantum gravity and are possible candidates for providing a unified description of all the fundamental forces in nature (gravity, electromagnetism, weak force and strong force). 
	A guiding  philosophy of supergravity theories is that one should be able to recover 4-dimensional physics from higher-dimensional models, hence to perform some sort of dimensional reduction by assuming the extra directions to be compact.
	
	The classical mathematical approach to the unification of general relativity with electromagnetism goes back to the works of physicists Kaluza \cite{Kaluza21} and Klein \cite{Klein26}. In their original works, one extra dimension is considered and the  five spacetime dimensional gravity is compactified on a circle $\S^1_R$ of radius $R$ to obtain at low energies a $1+3$ dimensional Einstein-Maxwell-Scalar field system. We will briefly discuss the reduction from the 5-dimensional to the 4-dimensional model in the next subsection. 
	
	In a seminal work by Witten \cite{WITTEN} it was proved that the Kaluza--Klein spacetime $\overline{g}$ is unstable at the semiclassical level. However, classical global stability was conjectured to hold true and such a result was proved by the third author \cite{Wyatt18} for small perturbations that do not depend on the compact direction. The goal of this paper is  to extend the result of \cite{Wyatt18} and to prove the global stability of $\overline{g}$ for more general perturbations that can a-priori depend also on the compact direction. We mention that a result analogous to \cite{Wyatt18} for cosmological Kaluza--Klein spacetimes, where the Minkowski spacetime is replaced by the 4-dimensional Milne spacetime, has also recently been shown by Branding, Fajman and Kr\"{o}ncke \cite{BFK19}. Furthermore global existence, without a restriction to $\S^1$-independent data, was shown on a quasilinear system of wave equations by the first two authors in \cite{HS} and on a 
	semilinear wave equation on a cosmological Kaluza--Klein spacetime in \cite{JWang}.
	In the context of higher-dimensional gravity we also cite a result by Ettinger \cite{ettinger} on the global well-posedness of a 11-dimensional, semilinear, gauge-invariant wave equation, and a global stability result by Andersson, Blue, Yau and the third author \cite{ABWY} for spacetimes with a supersymmetric compactification: that is, spacetimes $(\mathcal{M}, \hat{g})$ with $\mathcal{M} = \R^{1+n}\times K$ and $\hat{g} = \eta_{1+n} + k$, where $\eta_{1+n}$ is the $(1+n)$-dimensional Minkowski metric and $(K,k)$ is a compact Riemannian manifold that admits a spin structure and a nonzero parallel spinor. Their proof uses the assumption $n\ge 9$ but the result is conjectured to hold true for $n\ge 3$.

	\subsection{The zero-mode truncation}\label{zero-mode truncation}
	
	The Einstein equations in harmonic coordinates reduce to \eqref{g_equations}, which is a system of wave equations on the product space $\R^{1+3}\times \S^1$.
	Assuming for a moment that the compactifying circle $\S^1$ is replaced by the circle $\S^1_R$ of radius $R>0$, and by Fourier expanding the solution $g$ of \eqref{g_equations} along the periodic coordinate
	\[
	g_{\alpha\beta}(t,x,y) =\sum_{k\in\mathbb{Z}}e^{iky}g^k_{\alpha\beta}(t,x),
	\] 
	it turns out that
	\[
	(-\partial^2_t + \Delta_x + \partial^2_y) g_{\alpha\beta} = \sum_{k\in\mathbb{Z}} e^{iky} (-\partial^2_t + \Delta_x - (|k|/R)^2) g^k_{\alpha\beta}
	\] 
	which shows that the zero-modes $g^0_{\alpha\beta}$ of the metric coefficients are massless waves while the non-zero modes $g^k_{\alpha\beta}$ are massive (Klein--Gordon) waves with mass $|k|/R$ for $k\ne 0$. Equations \eqref{g_equations} are hence equivalent to a system on $\R^{1+3}$ which couples wave equations to an infinite sequence of Klein--Gordon equations with mass $|k|/R$, $k\in \mathbb{Z}\setminus \{0\}$.
	
	The heuristic physics argument, as explained by Pope in \cite{Pope}, to deal with this phenomenon is to assume the radius $R$ to be very small (a choice that would justify why we ``don't see'' the additional dimensions) so that the masses $|k|/R$ are too large to be physically observable. The non-zero modes are then neglected and the solution is truncated to the massless mode, in other words one assumes that $g_{\alpha\beta}(t,x,y) = g_{\alpha\beta}(t,x)$ is independent of the $y$ coordinate. 
	
	Under the zero-mode truncation assumption, one can reduce the Kaluza--Klein model to a three-dimensional Einstein-Maxwell-scalar field system. As explained in\cite{Pope}, this is done using the following standard ansatz, in which the higher dimensional metric coefficients $g_{\alpha\beta}$ are expressed in terms of three-dimensional fields $\hat{g}_{\bmalph\bmbeta}, \phi, \mathcal{A}_{\bmalph}$ by
	\[
	g_{\bmalph\bmbeta} = e^{2\kappa \phi}\hat{g}_{\bmalph\bmbeta} + e^{2\rho \phi}\mathcal{A}_{\bmalph} \q A_{\bmbeta}, \quad g_{\bmalph 4} = e^{2\rho\phi}\q A_{\bmalph}, \quad g_{44} = e^{2\rho\phi}
	\]
	where $\kappa = \sqrt{12}/12$ and $\rho=-2/\sqrt{12}$. The Einstein vacuum equations \eqref{EE} reduce then to the following minimally coupled $(1+3)$-dimensional Einstein-Maxwell-Scalar field system 
	\[
	\begin{aligned}
		& R_{\bmalph\bmbeta} = \frac{1}{2}\partial_{\bmalph}\phi\, \partial_{\bmbeta}\phi + \frac{1}{2}e^{-6\kappa\phi}\big(\q F_{\bmalph\bmmu}{\q F_{\bmbeta}}^{\bmmu} - \frac{1}{4}\q F_{\bmmu\bmnu}\q F^{\bmmu\bmnu}\hat{g}_{\bmalph\bmbeta} \big)\\
		& \nabla^{\bmalph} \big(e^{-6\kappa \phi} \q F_{\bmalph\bmbeta}\big) = 0\\
		& \tilde{\Box}_{\hat{g}} \phi = -\frac{3}{2}\kappa e^{-6\kappa \phi} \q F_{\bmmu\bmnu}\q F^{\bmmu\bmnu} 
	\end{aligned}
	\]
	where $\q F_{\bmalph\bmbeta} = \partial_\bmalph \q A_\bmbeta - \partial_\bmbeta \q A_\bmalph$. The above reduction can be also performed in higher dimensional settings where $\mathcal{M} = \R^{1+3}\times\mathbb{T}^d$. In the Kaluza--Klein setting, this truncation to the zero mode is consistent in the sense that a solution to the above Einstein-Maxwell-Scalar field system will be a solution to the original vacuum Einstein equations in 5 dimensions.
	
	The full global stability of the Kaluza--Klein spacetime to general perturbations, that may a-priori depend on the compact direction, involves studying solutions to a significantly more complicated PDE system than the simpler dynamics of the above Einstein-Maxwell-Scalar field system studied in \cite{Wyatt18}. This is the goal of the present article. We point out that we do not want to focus here on the dependence of the solution on the radius $R$ and, since there is no canonical choice of the radius $R$, we set $R = 1$.

	\subsection{4D Wave-Klein--Gordon systems}
	
	The dependence of the metric coefficients $g_{\alpha\beta}$ on the periodic coordinate $y$ and their Fourier decomposition along this direction reveal that system \eqref{g_equations} is equivalent to a system coupling wave equations to an infinite sequence of Klein--Gordon equations with different masses. The new system is also quasilinear and the coupling between the wave and Klein--Gordon components of the solution is strong. 
	
	The study of systems coupling (a finite number of) wave and Klein--Gordon equations has attracted considerable interest from the mathematical community, especially in the past three decades.
	In terms of small data global well-posedness results in $1+3$ spacetime dimensions we cite the initial results by Georgiev \cite{Georgiev:system} and Katayama \cite{Katayama:coupled_systems}, followed by LeFloch and Ma  \cite{LeFloch_Ma}, Wang \cite{Q.Wang, Wang20-EKG} and Ionescu and Pausader \cite{IoPa} who study such systems as a model for the full Einstein-Klein--Gordon equations, see \cite{LeFloch-Ma16} and \cite{IP2}. In \cite{LeFloch_Ma} and \cite{Q.Wang} global well-posedness is proved for compactly supported initial data and quadratic quasilinear nonlinearities that satisfy some suitable conditions, including the \textit{null condition} of Klainerman \cite{K86} for self-interactions between the wave components of the solution. An idea used in these works is that of employing hyperbolic coordinates
	in the forward light cone; this was first introduced by Klainerman \cite{klainerman:global_existence} for Klein--Gordon equations and Tataru in the wave context \cite{Tataru2002}, and later reintroduced by LeFloch and Ma in \cite{LeFloch_Ma} under the name of \textit{hyperboloidal foliation method}.
	In \cite{IoPa} global regularity and scattering is proved in the case of small smooth initial data that decay at a suitable rate at infinity and nonlinearities that do not verify the null condition but present a particular resonant structure. We also cite the work by Dong and the third author \cite{DW20}, who prove global well-posedness for a quadratic semilinear interaction in which there are no derivatives on the massless wave component. Other related results are \cite{Bachelot88, OTT_KGZ, T96_KGZ, T03_DP, T03_MH, klainerman_wang_yang, DLFW}. See also \cite{ma:2D_semilinear, Ma2020, ma:2D_tools, ma:2D_quasilinear, stingo_WKG, IS2019, DW20_2d, ma:1D_semilinear} for results about wave-Klein--Gordon systems in lower dimensions, in particular a work by the second author \cite{stingo_WKG} and a subsequent result in collaboration with Ifrim \cite{IS2019}, which are the only ones where 2-dimensional strongly coupled quadratic wave-Klein--Gordon systems with small mildly decaying data are investigated. Advanced techniques, among which semiclassical microlocal analysis, para/pseudo-differential calculus, wave packets, modified quasilinear energies, are employed there to tackle a problem that is critical, quasilinear and very weakly dispersive.
	
	\medskip
	A now-standard tool used in most of the aforementioned works is the vector field method. Linear wave and Klein--Gordon equations on $\R^{1+n}$ are invariant under translations, Euclidean rotations and hyperbolic rotations (linear wave equations are also scale-invariant). These symmetries provide a family of admissible vector fields (in the common terminology they are also referred to as \emph{Killing} vector fields of Minkowski spacetime), 
	\[
	\partial_\mu, \qquad \Omega_{\bmi\bmj} = x_\bmi\partial_\bmj - x_\bmj \partial_\bmi, \qquad \Omega_{0\bmi} = t\partial_\bmi + x_\bmi\partial_t
	\]
	which  commute with the linear wave and Klein--Gordon operators
	and are used to define higher order energy functionals which control the Sobolev regularity of the solution as well as its decay (and that of its derivatives) in space at infinity.
	The rotations $\Omega_{\bmi\bmj}$ and $\Omega_{0\bmi}$ are also usually referred to as \emph{Klainerman vector fields}.
	In the absence of Klein--Gordon equations, that is in the case of wave equations only, one can also consider the scaling vector field $\mathcal{S} = t\partial_t + x^\bmi\partial_\bmi$ (a \emph{conformal Killing} vector field of Minkowski) and use the control on higher order energies to derive fixed-time pointwise decay bounds for the solution via the so-called Klainerman-Sobolev inequalities (see Klainerman \cite{klainerman85}) 
	\begin{equation}\label{fixedtime-KS}
		(1+|t|+|x|)^{n-1}(1+ ||t|-|x||)|u(t,x)|^2\le C \sum_{|I|\le (n+2)/2} \|Z^I u(t,\cdot)\|^2_{L^2(\R^n)}.
	\end{equation}
	In the above inequality $Z$ denotes any of the vector fields $\partial_\mu, \Omega_{\bmi\bmj}, \Omega_{0\bmi}, \mathcal{S}$ and $Z^I$ is any product of $|I|$ such vector fields. Suitable energy estimates and pointwise decay bounds are subsequently used to control the nonlinear terms in the energy inequality and are essential to close the continuity argument which is at the core of the proof of a long-time/global existence result for small data.
	
	The inequality \eqref{fixedtime-KS} is, however, useless when dealing with Klein--Gordon equations. The scaling vector field does not commute well with the linear operator and one cannot generally expect to have a good control of the $L^2$ norm of $\mathcal{S}u$ when $u$ is a Klein--Gordon solution. Instead, if $u$ is compactly supported inside the light cone\footnote{Any cone $t  =|x|+c$ with $c>0$ would do.} $t = |x| + 1$ one can define higher order energy functionals on hyperboloids $t^2 - |x|^2 = s^2$ and exploit Klainerman-Sobolev inequalities on hyperboloids (see for instance \cite{FWY21})
	\begin{equation} \label{KS_hyperboloids_intro}
		\sup_{\hcal_s} t^{n/2}|u(t,x)|\le C \sum_{|I|\le (n+2)/2} \|B^I u \|_{L^2(\hcal_s)}
	\end{equation}
	where now $B^I$ are products involving hyperbolic rotations only, to get a good pointwise control on the solution. This approach has been largely used in the case of compactly supported initial data thanks to the finite speed of propagation satisfied by both wave and Klein--Gordon equations, but it is not adapted to treat the case of initial data that only enjoys some decay at infinity. Other methods have been employed to handle such cases, based on Fourier analysis, normal forms  and/or microlocal analysis: see for instance the work by Ionescu and Pausader \cite{IoPa} in the 1+3 dimensional setting, by the second author \cite{stingo_WKG} and in collaboration with Ifrim \cite{IS2019} for the 1+2 dimensional case, and references therein. 
	See also a recent work by LeFloch and Ma \cite{LeFlochMa22-I} using a foliation that merges hyperboloids with constant time slices.
	
	\subsection{The 5D problem: main theorem and overview of the proof}
	
	According to the positive mass theorem, the solution $g_{\alpha\beta}$ of the Cauchy problem \eqref{g_equations}-\eqref{data_g} must have a non-trivial tail at spacelike infinity\footnote{We choose to write this tail so that $g_{\alpha\beta}$ corresponds to the Schwarzschild metric in wave coordinates at leading order.} which suggests to set $g_{\alpha\beta} = \overline{g}_{\alpha\beta} + h^0_{\alpha\beta}+h^1_{\alpha\beta}$ where \begin{equation}\label{h0_def}
		\begin{aligned}
			& h^0_{\bmalph\bmbeta} = \chi\Big(\frac{r}{t}\Big)\chi(r)\frac{M}{r}\delta_{\bmalph\bmbeta}, \qquad h^0_{44} =0,\\
			& \chi\in \mathcal{C}^\infty(\R) \text{ with } \chi(s)=0 \text{ for } s\le 1/2, \ \chi(s)=1 \text{ for } s\ge 3/4, \ r=|x|
		\end{aligned}
	\end{equation}
	and look for $h^1_{\alpha\beta}$ the solution to the following system of quasilinear wave equations
	\begin{equation}\label{h_equations}
		\tilde{\Box}_g h^1_{\alpha\beta} = F_{\alpha\beta}(h)(\partial h, \partial h) - \tilde{\Box}_g h^0_{\alpha\beta}, \qquad \text{on  } \R^{1+3}\times \S^1
	\end{equation}
	with data 
	$(h^1|_{\alpha\beta}, \partial_th^1_{\alpha\beta})|_{t=2}$ being small and sufficiently decaying in space. 
	The semilinear source term in the above right hand side decompose into the following sum
	\[
	F_{\alpha\beta}(h)(\partial h, \partial h) = P_{\alpha\beta}(\partial h, \partial h) + \mathbf{Q}_{\alpha\beta}(\partial h, \partial h) + G_{\alpha\beta}(h)(\partial h, \partial h)
	\]
	where
	\begin{itemize}
		\item[-] $P_{\alpha\beta}(\partial h, \partial h)$ are quadratic \emph{weak null} terms
		\[
		P_{\alpha\beta}(\partial h, \partial h) = \frac{1}{4}\bar{g}^{\mu\rho}\bar{g}^{\nu\sigma}\left(\partial_\alpha h_{\mu\rho}\partial_\beta h_{\nu\sigma} - 2 \partial_\alpha h_{\mu\nu}\partial_\beta h_{\rho\sigma}\right),
		\]
		\item[-] $\mathbf{Q}_{\alpha\beta}(\partial h, \partial h)$ is a linear combination of the classical quadratic null forms,
		\item[-] $ G_{\alpha\beta}(h)(\partial h, \partial h)$ are cubic terms. More precisely, they are quadratic in $\partial h$ with smooth coefficients depending on $h$ so that $ G_{\alpha\beta}(0)(\partial h, \partial h) =0$.
	\end{itemize}
	The reduced wave operator can be written as $\tilde{\Box}_g = \Box_{xy} + H^{\mu\nu}\partial_\mu\partial_\nu$, where $\Box_{xy} = -\partial^2_t + \Delta_x + \partial^2_y$ is the flat wave operator and $H^{\mu\nu} : = g^{\mu\nu} - \overline{g}^{\mu\nu}$ is the formal inverse of $h_{\mu\nu}$ for small $h$, i.e. 
	\begin{equation} \label{H_def}
		H^{\mu\nu} = -h^{\mu\nu} + \mathcal{O}^{\mu\nu}(h^2) = -\overline{g}^{\mu\rho}\overline{g}^{\nu\sigma} h_{\rho\sigma} + \mathcal{O}^{\mu\nu}(h^2).
	\end{equation}
	We can now give a more precise statement of our main result.
	{
		\begin{theorem}\label{thm:Main}
			Let $\kappa>0$. There exists $N\in \N$ sufficiently large and $\ep_0>0$ small such that, for any $0<\ep<\ep_0$ and initial data $g_0,K_0$ solving the constraint equations and satisfying
			\begin{equation}\label{data_maintheo}
				\begin{gathered}
					\sum_{m\leq N} \sum_{i+j=m}\|(1+r)^{\frac{1}{2}+i+\kappa}\partial_y^j\nabla_x^i(g_0-g^0)\|_{\dot{H}^1_{x,y}} +\|(1+r)^{\frac{1}{2}+i+\kappa}\partial_y^j\nabla_x^iK_0\|_{L^2_{x,y}} \leq \ep,\\
					\sum_{m\leq N-1} \sum_{i+j=m}\|(1+r)^{\frac{3}{2}+i+\kappa}\partial_y^j\nabla_x^i(g_0-g^0)\|_{\dot{H}^2_{x,y}} +\|(1+r)^{\frac{3}{2}+i+\kappa}\partial_y^j\nabla_x^iK_0\|_{\dot{H}^1_{x,y}} \leq \ep
				\end{gathered}
			\end{equation}
			together with the $L^2$ estimate
			$$\|(1+r)^{-\frac{1}{2}+\kappa}(g_0-g^0)\|_{L^2_{x,y}}\le \ep$$
			with $r=|x|$ and $g^0$ defined by
			\[
			g^0_{\bmi\bmj}=(1+M\chi(r)r^{-1})\delta_{\bmi\bmj}, \quad g^0_{44}= 1, \quad g^0_{4\bmi}=0,
			\]
			there exists a unique global solution $g_{\alpha\beta}$ to \eqref{g_equations} with initial data given by \eqref{data_g}. This solution obeys the Einstein equations and decomposes as $g_{\alpha\beta}=\bar{g}_{\alpha\beta}+h^0_{\alpha\beta}+h^1_{\alpha\beta}$, with $h^0_{\alpha\beta}$ defined by \eqref{h0_def} and $h^1_{\alpha\beta}$ satisfying the pointwise estimate
			$$|h^1_{\alpha\beta}|\leq \frac{C_0\ep}{(1+t+|x|)^{1-\gamma}}$$
			with $C_0$ a numerical constant and $\gamma>0$ arbitrarily small but fixed.
		\end{theorem}
	}
	
	The proof of the above result is based on a bootstrap argument, i.e. on the propagation of some suitable a-priori energy estimates and pointwise decay bounds on the solution, which is performed in two main steps:
	
	\emph{Step 1}: deduction of higher order energy inequalities and of sharp pointwise estimates from the a-priori energy assumptions;
	
	\emph{Step 2}: estimates of the trilinear and quartic terms appearing in the right hand side of the energy inequalities. In particular, deduction of suitable higher order $L^2$ estimates of the source terms from the a-priori energy assumptions and the pointwise decay bounds.
	
	\smallskip
	\noindent In order to run the above argument and in view of the issues discussed in the previous subsection, one needs to find a strategy to obtain (at least in the first instance) pointwise decay bounds on the solution from the a-priori assumptions, knowing that inequality \eqref{fixedtime-KS} cannot be used and \eqref{KS_hyperboloids_intro} is valid only in the interior of some light cone. 
	
	Similar to \cite{HS}, the approach we take in the present paper is to decompose the whole spacetime and study the problem separately in two regions, corresponding to the interior and exterior of a hyperboloid\footnote{In this curved background, the Minkowski cone $\{t=|x|+1\}$ is in fact only asymptotically spacelike.} asymptotically approaching the cone $\{t=|x|+1\}\times \S^1$. This decomposition is quite natural, in that the analysis in the exterior is totally independent of that in the interior and requires different tools. It also allows us to explain our arguments with more clarity.
	
	\subsubsection{Exterior region: the bootstrap assumptions.}
	
	The bootstrap assumptions in the exterior region are higher order weighted energy estimates on the solution
	\[
	\begin{aligned}
		& \enext(t, Z^{\le N}h^1)^{1/2}\le 2C_0\ep t^{\sigma}, \\
		& E^{e, 1+\kappa}(t, \partial Z^{\le N-1}h^1)^{1/2}\le 2C_0\ep t^\sigma
	\end{aligned}
	\]
	where the weighted energy functional is defined, for any $\lambda>0$, as
	\begin{multline*}
		E^{e, \lambda}(t, \hab) = \iint_{\{|x|\ge t-1\}\times \S^1} (2+|x|-t)^{1+2\lambda} |\nabla_{txy} \hab(t,x,y)|^2 dxdy\ \\
		+ \int_2^t\iint_{\{|x|\ge \tau-1\}\times \S^1} (2+|x|-t)^{2\lambda} |\overline{\nabla} \hab(\tau, x,y)|^2 dxdyd\tau.
	\end{multline*}
	In the above integrals, $\nabla_{txy}$ denotes the spacetime gradient while $\overline{\nabla} = (\op_0, \dots, \op_4) = (\partial_t+\partial_r, \gd_{\bmi}, \partial_y)$ denotes the tangent gradient to the cones $\{t = r+1\}\times \S^1$, with $\gd_\bmi = \partial_i - \frac{x_\bmi}{r}\partial_t$ being the angular derivatives. The parameter $\kappa$ in the above a-priori estimates is related to the asymptotic decay of the data, $N\in \N$ is assumed to be sufficiently large and $0<\sigma<\kappa$ sufficiently small.
	Weighted Sobolev and Hardy inequalities allow us to obtain fixed-time pointwise decay bounds on the solution from the assumptions on the weighted energies, as for any given smooth function $U$
	\[
	\begin{gathered}
		|\nabla_{txy}U(t,x,y)|\le C (1+|x|)^{-1}(2+|x| -t)^{-\frac12 -\lambda}\sum_{|I|\le 3}E^{e, \lambda}(t,Z^I U)^{1/2},\\
		|U(t,x,y)|\le C (1+|x|)^{-1}(2+|x|-t)^{-\lambda}\sum_{|I|\le 2}E^{e, \lambda}(t,Z^I U)^{1/2}.
	\end{gathered}
	\]
	They also allow us to uncover faster spacetime decay for the tangential derivatives, since their weighted $L^2$-spacetime norm is controlled by the energy, and to recover the well-known property of waves that higher order derivatives enjoy better decay in terms of the distance from the outgoing Minkowski cones, which follows from the second energy assumption above. We point out that, in the context of waves on $\R^{1+3}$ where the full range of vector fields $\Gamma\in \{\Omega_{\bmi\bmj}, \Omega_{0\bmi}, \scal\}$ is available, the latter two properties are easily derived from algebraic relations. In particular, one can use that
	\[
	|\op \psi|\lesssim \sum_{|I|\le 1}\frac{|\Gamma^I \psi|}{1+t+|t-|x||},\qquad |\partial^2\psi|\lesssim \sum_{|I|\le 1}\frac{|\partial \Gamma^I \psi|}{1+|t-|x||}.
	\]
	
	\subsubsection{Interior region: the bootstrap assumptions.}
	
	The bootstrap assumptions in the interior region are bounds on higher order energies defined on truncated hyperboloids
	\[
	\hin_s = \{(t,x): t^2-|x|^2 = s^2 \text{ and } t\ge 1+\sqrt{1+|x|^2} \}\times \S^1, \qquad s\ge 2
	\]
	which are the branches of hyperboloids contained in the interior region, and pointwise decay bounds on differentiated metric coefficients carrying only Klainerman vector field derivatives. We denote the zero-mode (respectively zero-average) component of coefficient $\hab$ by
	\[
	\habf = \fint_{\S^1}\hab dy, \quad \text{ and } \quad \habn = \hab-\habf.
	\]
	We assume that, for some large integers $1\ll N_1\ll N$ and some small\footnote{In practice, $\zeta,\gamma$ and $\delta$ are going to be replaced with a hierarchy of increasing $\zeta_k, \gamma_k$ and $ \delta_k$, where $k$ accounts for the number of Klainerman vector fields in the product $Z^I$, so that $\zeta_i\ll \gamma_j\ll \delta_k$ for any $i, j,k$ and the algebraic relation $\gamma_i + \delta_j <\delta_k$ whenever $j<k$.} $0<\zeta<\gamma\ll\delta$, the following bounds are satisfied
	\[
	\begin{aligned}
		&E^i(s, \partial^{\le 1} Z^{\le N}\hab) \le C\epsilon^2 s^{1+\zeta},\\
		& E^i(s, Z^{\le N}\habf) \le C\epsilon^2 s^{\zeta},\\
		& E^i(s, \partial^{\le N-N_1}Z^{\le N_1}\hab) \le C\epsilon^2 s^{\delta} 
	\end{aligned}
	\]
	where 
	\[
	\begin{aligned}
		E^i(s, \hab) &:=\iint_{\hin_s } \big|(s/t)\partial_t \hab \big|^2 + |\underline{\nabla}\hab|^2 dxdy \\
		&\, = \iint_{\hin_s} \big|(s/t)\nabla_x\hab\big|^2 + \big|(1/t)\scal \hab\big|^2 + \sum_{1\le i<j\le 3}\big|(1/t)\Omega_{ij}\hab\big|^2 + |\partial_y\hab|^2\, dxdy
	\end{aligned}
	\] 
	and with $\Gamma\in \{\Omega_{\bmi\bmj}, \Omega_{0\bmi}\}$
	\[
	|\Gamma^{\le N_1} \habf(t,x)|\le C\ep (1+t)^{-1+\gamma}(1+|t-|x||)^\gamma, \quad 
	\]
	\[
\|	t^{\frac32}\partial_y^{\le 1}(\partial^I\Gamma^J\habn)\|_{L^\infty_xL^2_y(\hin_s)} + \|	t^{\frac12}s\partial_{tx}(\partial^I\Gamma^J\habn)\|_{L^\infty_xL^2_y(\hin_s)}\le C\ep s^\gamma, \quad |I|+|J|\le N_1+1,\ |J|\le N_1.
	\]
	
	In the above energy functional, $\underline{\nabla} = (\underline{\partial_1}, \dots, \underline{\partial_4})$ denotes the tangent gradient (to the hyperboloids) with $\underline{\partial}_{\bmi} = \partial_{\bmi} + (x_{\bmi}/t)\partial_t$ for $\bmi=1,2,3$ and $\underline{\partial}_4 =\partial_4$.
	Klainerman-Sobolev inequalities on hyperboloids permit us to deduce pointwise decay bounds for the solution, as for any given smooth function $U$ one has
	\[
	\begin{gathered}
		\sup_{\S^1}|\nabla_{tx}U(t,x,y)|\le C (1+t)^{-1}(1+|t-|x||)^{-1/2}\sum_{|I|\le 3}E^i(s, Z^I U)^{1/2}, \\
		\sup_{\S^1} |\underline{\nabla}U(t,x,y)|\le C(1+t)^{-3/2}\sum_{|I|\le 3}E^i(s,Z^IU)^{1/2}.
	\end{gathered}
	\]
	Note that the latter inequality shows, again, that tangential derivatives
	enjoy better decay estimates than usual derivatives. We postpone the explanation of why we use the above hierarchy of energy assumptions to later in this section.
	
	\subsubsection{Estimates on inhomogeneities: null and weak-null terms}
	
	Once energy bounds and pointwise decay bounds are available, one has to estimate the trilinear and quartic terms appearing in the right hand side of the energy inequalities. These involve the source terms of the equation satisfied by the differentiated coefficients $Z^K\hab$
	\[
	\tilde{\Box}_g Z^Kh^1_{\alpha\beta} = F^K_{\alpha\beta} + F^{0,K}_{\alpha\beta}
	\]
	where
	\[
	F^K_{\alpha\beta} = Z^KF_{\alpha\beta}(h)(\partial h, \partial h) - [Z^K, H^{\mu\nu}\partial_\mu\partial_\nu]h^1_{\alpha\beta} , \qquad F^{0,K}_{\alpha\beta} = Z^K\tilde{\Box}_g h^0_{\alpha\beta}
	\]
	are semilinear quadratic interactions.
	The explicit inhomogeneous terms $F^{0, K}_{\alpha\beta}$ and the differentiated cubic terms $Z^KG_{\alpha\beta}(h)(\partial h, \partial h)$ are short range perturbations of the linear equations. We do not discuss them here as they cause no issue in the analysis. The differentiated null terms $Z^K\mathbf{Q}_{\alpha\beta}(\partial h, \partial h)$ are also easily controlled, thanks to the following well-known property
	\[
	\begin{aligned}
		|Q_0(\partial \psi, \partial \varphi)| + |Q_{\alpha\beta}(\partial \psi, \partial \varphi)|& \lesssim |\op \psi| |\partial \varphi| + |\partial\psi| |\op \varphi|\\
		& \lesssim |\pb\psi| |\partial \varphi| + |\partial\psi||\pb\varphi| + (s/t)^2 |\partial\psi| |\partial\varphi|
	\end{aligned}
	\]
	and the better behavior of tangential derivatives.
	
	The quadratic interactions that are more delicate to treat and require special attention are the differentiated null terms $Z^KP_{\alpha\beta}(\partial h, \partial h)$ and the commutator terms $[Z^K, H^{\mu\nu}\partial_\mu\partial_\nu]\hab$.
	The particular structure of such terms was first highlighted by Lindblad and Rodnianski \cite{LR03, LR05, LR10} in the 4-dimensional setting and shows all its  potential in the null frame $\mathcal{U} = \{L,\Lb, S^1, S^2\}\cup\{\partial_y\}$, where $L = \partial_t + \partial_r$, $\Lb = \partial_t-\partial_r$ and $S^1, S^2$ are smooth vector fields tangent to the spheres $\S^2 = \{u\in \R^3 : u\cdot x/|x|=0\}$. 
	
	As concerns the weak null terms, one sees that if the metric tensor is expressed with respect to $\q U$ then
	\[
	P_{\alpha\beta}(\partial h, \partial h)\sim \partial h_{ T U}^2 + \partial h_{LL}\partial h_{\Lb\Lb}, \qquad T\in \q T, U\in \q U \footnote{For any tensor $\pi_{\alpha\beta}$ and any two vector fields $X = X^\alpha\partial_\alpha$, $Y = Y^\alpha\partial_\alpha$, we define $\pi_{XY} = \pi_{\alpha\beta}X^\alpha Y^\beta$.}
	\]
	where $\mathcal{T} = \{L, S^1, S^2\}\cup \{\partial_y\}$ denotes the frame tangent to the flat outgoing cones. On the one hand, the choice of gauge (in particular the \emph{wave coordinate condition}) ensures that the derivatives of $h_{LT}$ coefficients are well behaved, as they satisfy
	\begin{equation}\label{partial_hLL}
		|\partial h_{LT}|\lesssim |\op h| + \mathcal{O}(h\cdot\partial h).
	\end{equation}
	On the other hand, the metric coefficients $h_{TU}$ solve quasilinear wave equations whose source terms are null or cubic. In the exterior region, we exploit this property to prove that the higher order weighted energies of such coefficients grow at a slower rate $t^{C\ep}$, where $\ep\ll \sigma$ is the size of the data. From this we infer an improved pointwise decay for $\partial Z^Kh_{TU}$ with $|K|\ll N$ and the following weighted $L^2$ bound for the differentiated weak null terms
	\[
	\sum_{i=0}^1\big\|(2+r-t)^{\frac{1}{2}+i+\kappa}\partial^i Z^{\le N-i}P_{\alpha\beta}\big\|_{L^2}\lesssim \ep^2 t^{-1+C\ep} + \q O(\ep^2 t^{-1-}).
	\]
	The above estimate shows that the weak null terms contribute to a slow growth of the exterior energies. In the interior region, the enhanced pointwise bounds satisfied by the derivatives of $Z^K h^1_{TU}$ for $|K|\ll N$ are instead obtained directly from the equations they satisfy, using integration along characteristics as done in \cite{LR10}. This approach is possible provided that we already have at our disposal suitable bounds on the solution in the exterior region.

	\subsubsection{Commutator terms in the exterior region}
	
	The commutator terms also display an important structure when expressed with respect to the null frame. The tensor $H^{\mu\nu}$ is decomposed as follows
	\begin{equation} \label{dec_H}
		H^{\mu\nu}:= H^{0,\mu\nu} + H^{1,\mu\nu}, \quad	H^{0,\bmmu\bmnu} :=-\chi\Big(\frac{r}{t}\Big)\chi(r) \frac{M}{r}\delta^{\bmmu\bmnu}, \quad H^{0,44}=0, 
	\end{equation}
	where $ H^{0, \mu\nu}$ is the ``Schwarzschild part'' of $H$. The estimates of 	$[Z^K, H^{0, \mu\nu}\partial_\mu\partial_\nu]\hab$ are straightforward and, similar to the weak null terms discussed above, responsible for a slow growth of the exterior energy. The estimates of the commutator involving coefficients $H^{1,\mu\nu}$ are instead obtained using the fact that, for any tensor $\pi^{\mu\nu}$ and function $\psi$,
	\[
	|\pi^{\mu\nu} \partial_{\mu}\partial_{\nu} \psi|\lesssim |\pi_{LL}|| \partial^2 \psi |+ |\pi| |\partial \op \psi| 
	\]
	so that either  the tensor coefficient is a ``good'' coefficient $\pi_{LL}$ or one of the two derivatives acting on $\psi$ is a tangential derivative. As highlighted above, in the exterior region the enhanced behaviour of second order derivatives $\partial \op$ as well as of $\partial^2$ is encoded in the energy assumptions. What is more, weighted Hardy type inequalities and weighted Sobolev-Hardy inequalities allow us to get a good control of the higher order weighted $L^2$ norms, as well as to recover good pointwise decay bounds, of the solution with no derivatives.
	Suitable higher order weighted $L^2$ estimates for these commutator terms in the exterior region follow then rather easily. 
	
	\subsubsection{Commutator terms in the interior region}
	
	A much more delicate analysis of the commutator terms is required in the interior region. On the one hand, the interior energy assumptions do not provide us with additional information on the second order derivatives and the interior energy functionals only give a $\dot{H}^1$ type control on the differentiated solution. The classical Hardy inequality written on hyperboloids is
	\[
	\| r^{-1}U\|_{L^2(\hin_s)}\lesssim \|\pb U\|_{L^2(\hin_s)} + \|\partial U\|_{L^2(\Sext_{t_s})}
	\]
	where $\Sext_{t_s}$ is the exterior constant time slice that intersects the interior hyperboloid $\hin_s$ on the boundary between the two regions. Such an inequality provides us with a control of the $L^2$ norm of the undifferentiated solution at the costly expense of a $r^{-1}$ factor. 
	On the other hand, no extra decay (in terms of the distance from the outgoing cones) is expected for the second order derivatives of the solution. In fact, the zero-average component of the solution $\habn$ is a Klein--Gordon type function, in that each of its Fourier mode along the $y$-direction is solution to a Klein--Gordon equation (see subsection \ref{zero-mode truncation}). As a consequence of this latter fact one only has $|\partial^2\habn| + |\partial \habn|\lesssim (1+t+r)^{-3/2}$, which coupled with the above Hardy inequality gives
	\[
	\big\|Z^K\hff \cdot\partial^2 \habn\big\|_{L^2(\hin_s)}\lesssim s^{-1/2}\big(\|\pb \hff\|_{L^2(\hin_s)} + \|\partial \hff\|_{L^2(\Sext_{t_s})}\big)\lesssim s^{-1/2+\delta}.
	\]
	The same inequality holds if $\hff$ is replaced by $(H^{1,\mu\nu})^\flat$.		
	These ``wave-Klein--Gordon'' contributions to the commutator are the ones responsible for the $s^{1+}$ growth  of the higher order energies on $\hin_s$. They are, however, absent in the equations satisfied by the zero-modes $Z^K\habf$, as for any two functions $f,g$ one has
	\[
	(f\cdot g)^\flat = f^\flat \cdot g^\flat + \big(f^\natural\cdot g^\natural \big)^\flat,
	\]
	therefore a much slower growth is expected for the higher order energies of $\habf$. 
	
	The above observation motivates the use of a hierarchy in the interior energy assumptions and to separately propagate the higher order energy estimates for the zero-modes. 
	To propagate the different interior energy assumptions, we then need to estimate the commutators $[Z^K, \pi^{1,\mu\nu}\partial_\mu\partial_\nu]\phi$ separately for $\pi=H^{1,\flat}, H^{1, \natural}$ and $\phi=\habf, \habn$. The analysis is reasonably straightforward  when $\pi=H^{1, \natural}$ as we can rely on the Poincar\'e inequality. When $\pi=H^{1, \flat}$ the analysis is finer, as we express the metric coefficients $H^{1, \mu\nu}$ relative to the null framework and all derivatives in terms of $\partial_t, \partial_y$ and of the tangential derivatives $\pb_\bma$ to hyperboloids. Doing this, we see that
	\begin{multline*}
		|[Z^K, (H^{1,\mu\nu})^\flat\partial_\mu\partial_\nu]\phi|\\
		\lesssim |Z^K H^{1, \flat}_{LL}||\partial^2_t \phi| + |Z^K H^{1, \flat}_{4L}||\partial_t\partial_y \phi| + \frac{|t^2-r^2|}{t^2}|Z^KH^{1, \flat}||\partial^2_t \phi| + \frac{ |Z^KH^{1, \flat}||\partial Z \phi|}{1+t+r} +\dots
	\end{multline*}
	The remarkable property of the above right hand side is that each quadratic term either contains coefficients $H^{1, \flat}_{LL}$ and $H^{1, \flat}_{4L}$ - which are ``good'' as a consequence of the wave condition - or have an extra decaying factor $(|t^2-r^2|/t)^2$ and $(1+t+r)^{-1}$. 
	Then suitable estimates on the $L^2(\hin_s)$ norms of the above terms are obtained by using a Hardy inequality \emph{\`a la} Lindblad and Rodnianski with weights in $t-r$, which allow us to better exploit the pointwise decay of our solution. We point the reader to subsection \ref{sub:commutators_int} for further details. 
	
	\subsubsection{The null framework}
	We emphasise the importance of choosing a framework which correctly highlights the structure of the weak null terms. In fact, the absence of ``bad'' interactions, such as $(\partial h_{\Lb\Lb})^2$ and $\partial h_{LL}\cdot\partial h_{TU}$ in the expression of the weak null terms with respect to the null framework, is crucially related to the fact that the transversal field $\Lb$ is orthogonal to $\S^2\times \S^1$, i.e. that $\bar{g}(\Lb, A)=0$ for $A\in \{S^1, S^2, \partial_y\}$. On the contrary, the framework arising naturally from hyperboloids
	\[
	\mathcal{F}=\Big\{\partial_t, \quad \pb_\bma = \partial_a + \frac{x^\bma}{t}\partial_t\Big\}\cup \{\partial_y\},
	\]
	in which the ``transversal field'' ($\partial_t$ in the above example) is not orthogonal to $\S^2\times \S^1$, causes the analogue of the bad interaction $\partial h_{\Lb\Lb}\cdot\partial h_{TU}$ to appear and critically fails to give a useful expression for the weak null terms.  This consideration leads us to adopt the null frame decomposition both in the exterior and the interior region  and to combine it with the foliation by hyperboloids in the latter region. Indeed in this region, and when required, the metric coefficients are expressed with respect to the null frame $\q U$ (in order to use the enhanced behavior of $h_{LT}$ and $h_{TU}$ coefficients) while derivatives are written in terms of those in $\q F$ (in order to distinguish between the ``good'' tangential derivatives $\pb_\bma, \partial_y$ and the ``bad'' direction $\partial_t$).  Note our approach is  different from what was done in previous works on Einstein-Klein--Gordon systems. 
	We finally mention that a different framework than the null one is used by Ionescu and Pausader \cite{IP2}, which is reminiscent of the div-curl decomposition of vector-fields in fluid models and is more compatible with the Fourier transform approach employed there.

	\subsubsection{The Einstein-Klein--Gordon equations}
	We conclude by pointing out that our proof can be used, \textit{mutatis mutandis}, to provide a new proof of the stability of the Minkowski solution to the Einstein-Klein--Gordon equations. 		
	To briefly illustrate this point, we recall that in a harmonic gauge the Einstein-Klein--Gordon equations read
	\begin{equation}\label{EKG}
		\tilde{\Box}_g h_{\bmalph\bmbeta} = \tilde{F}_{\bmalph\bmbeta}(h)(\partial h, \partial h)-2\Big(\partial_\bmalph \phi \partial_\bmbeta \phi + \frac{m^2}{2} g_{\bmalph\bmbeta}\Big), \qquad
		\tilde{\Box}_g\phi = m^2 \phi .
	\end{equation}
	These equations are posed on $\R^{1+3}$, $m>0$ is a constant parameter and $h$ is a perturbation away from the Minkowski spacetime $\mathbf{m}$ defined via $g_{\bmalph\bmbeta} = \mathbf{m}_{\bmalph\bmbeta}-h_{\bmalph\bmbeta}$. The system \eqref{EKG} is much simpler to treat than \eqref{g_equations}. For example, without the $\S^1$, the metric tensor $h$ remains entirely wave-like and so all problematic wave--Klein--Gordon commutators no longer occur. The only Klein--Gordon field  is $\phi$ and it couples into the equation for the metric only via semilinear nonlinearities. This coupling is weak in the sense that the bootstrap assumptions for $h_{\bmalph\bmbeta}$ and $\phi$ can be propagated separately. 
	To conclude, due to our choice of null framework, combined with the separate analysis used in the interior and exterior regions, our proof provides an alternative perspective from what was done in previous works on Einstein-Klein--Gordon systems in \cite{LeFloch-Ma16, IP2}. 
	
	\subsection{Notation}\label{sub:notations}
	
	Below is a list of notation, some of which have already been introduced in the introduction, that we will use throughout the paper.
	
	\medskip
	\noindent \textsc{Coordinates}:
	\begin{itemize}
		\item[•] $\{x^\alpha\}_{\alpha=0,\dots, 4}$ with $x^0=t\in\R$, $x = (x^1, x^2, x^3)\in \R^3$, $x^4= y\in \S^1$ are the harmonic coordinates. They satisfy the geometric wave equation $g^{\mu\nu}\nabla_\mu\nabla_\nu x^\alpha = 0$. We will always denote by $r = |x|$ the radial component of $x$;
		\item[•] $u=t+r$ and $\ub=t-r$ are the null coordinates. They are used in the exterior region.
	\end{itemize}
	
	\medskip
	\noindent \textsc{Derivatives}:
	\begin{itemize}
		\item[•] $\nabla_{txy} = (\partial_0, \dots, \partial_4)$ denotes the spacetime gradient, with $\partial_\mu = \partial/\partial x^\mu$. $\nabla_{xy}$ denotes the full spatial gradient in $\R^3\times \S^1$ while $\nabla_x$ is the spatial gradient in $\R^3$. $\nabla_{tx}$ is the 4D spacetime gradient;
		\item[•] $\partial_{xy}$ denotes any of the derivatives $\partial_i$ with $i=1,\dots, 4$, while $\partial_x$ denotes any of the derivatives $\partial_{\bmi}$ with $\bmi=1,2,3$. The definition of $\partial_{tx}$ and $\partial_{txy}$ are similar. We will use $\partial$ and $\partial_{txy}$ interchangeably;
		\item[•] $\Box_{xy} = -\partial^2_t + \Delta_x+\partial^2_y$ and $\Box_x = -\partial^2_t + \Delta_x$;
		\item[•] $\partial_r = (x^{\bmi}/r)\partial_{\bmi}$ denotes the radial derivative in $\R^3$;
		\item[•] $\gd$ denotes any of the angular components $\gd_{\bmi} = \partial_\bmi - (x_\bmi/r)\partial_r$ of $\partial_{\bmi}$ for $\bmi=1,2,3$;
		\item[•] $\partial_u = (1/2)(\partial_t + \partial_r)$ and $\partial_\ub = (1/2)(\partial_t - \partial_r)$ denote the null derivatives;
		\item[•] $\overline{\nabla} = (\op_0, \dots, \op_4) = (\partial_t+\partial_r, \gd_{\bmi}, \partial_y)$ denotes the tangent gradient to the cones $\{t = r+1\}\times \S^1$. Moreover $\overline{\nabla}_x= (\op_0, \dots, \op_3) = (\partial_t + \partial_r, \gd_{\bmi})$;
		\item[•] $\op$ denotes any of the tangent derivatives $\op_\alpha$ in $\R^{1+3}\times \S^1$, $\op_x$ denotes any of the tangent derivatives $\op_{\bmalph}$ in $\R^{1+3}$;
		\item[•] $\underline{\nabla} = (\pb_1, \dots, \pb_4)$ denotes the tangent gradient to the hyperboloids in $\R^{1+3}\times \S^1$, with $\pb_{\bmi} = \partial_{\bmi} + (x^{\bmi}/t)\partial_t$ and $\pb_4 = \partial_y$. Moreover $\underline{\nabla}_x = (\pb_1, \pb_2, \pb_3)$;
		\item[•] $\pb$ denotes any of the tangent derivatives $\pb_\bma, \pb_4$ in $\R^{1+3}\times \S^1$, $\pb_x$ denotes any of tangent derivatives $\pb_{\bma}$ in $\R^{1+3}$. Sometimes we will use $\pb_0 = \partial_t$.
	\end{itemize}
	
	\medskip
	\noindent \textsc{Products}:
	\begin{itemize}
		\item Given a multi-index $\alpha = (\alpha_0, \alpha_1, \dots, \alpha_4) \in \mathbb{N}^5$, its length is computed classically as $|\alpha| = \sum_{i=0}^4 \alpha_i$. We set $\partial^\alpha := \partial_0^{\alpha_0}\partial_1^{\alpha_1}\partial_2^{\alpha_2}\partial_3^{\alpha_3} \partial_4^{\alpha_4}$ and $\partial^\alpha_{x}:=\partial^{\alpha_1}_1\partial^{\alpha_2}_2\partial^{\alpha_3}_3$. The definition of $\partial^\alpha_{xy}$ and $\partial^\alpha_{tx}$ are analogous;
		\item More generally, given a family of vector fields $\{X_1, \dots, X_n\}$ and a multi-index $\alpha=(\alpha_1, \dots, \alpha_n) \in\mathbb{N}^n$, $X^\alpha =X^{\alpha_1}_1\dots X^{\alpha_n}_n$.	With an abuse of notation we will sometimes write $X^k$ (resp. $X^{\le k}$) instead of $\sum_{\alpha: |\alpha|=k}X^\alpha$ (resp. $\sum_{\alpha: |\alpha|\le k}X^\alpha$).
	\end{itemize}
	
	\medskip
	\noindent \textsc{Metrics}:
	\begin{itemize}
		\item[•] $\overline{g} = -(dt)^2 + \sum_\bmi (dx^\bmi)^2 + (dy)^2$ denotes the Kaluza--Klein metric on $\R^{1+3}
		\times \S^1$;
		\item[•] $g$ denotes a solution of the Einstein equations \eqref{EE} on $\R^{1+3}\times\S^1$;
		\item[•] $\overline{g}^{\alpha\beta}$ and $g^{\alpha\beta}$ denote the inverse of the metrics $\overline{g}_{\alpha\beta}$ and $g_{\alpha\beta}$ respectively. For any other arbitrary $n$-tensor tensor $\pi_{\alpha_1\dots\alpha_n}$, indices are raised and lowered using $\overline{g}$, e.g ${\pi^{\alpha_1}}_{\alpha_2\dots\alpha_n} = \overline{g}^{\alpha_1\mu}\pi_{\mu\alpha_2\dots\alpha_n}$;
		\item[•] $H^{\alpha\beta}= g^{\alpha\beta} - \overline{g}^{\alpha\beta}$ corresponds to the formal inverse of $h_{\alpha\beta} = g_{\alpha\beta} - \overline{g}_{\alpha\beta}$. When $h$ is sufficiently small we have $H^{\alpha\beta} = -h^{\alpha\beta}+\mathcal{O}^{\alpha\beta}(h^2)$.
	\end{itemize}

	\medskip
	\noindent \textsc{Null Frame and Decomposition}:
	\begin{itemize}
		\item[•] $L = \partial_t + \partial_r$ denotes the vector field tangent to the outgoing null cones in $\R^{1+3}\times \S^1$. In components, $L^0 = 1$, $L^{\bmi} = x^{\bmi}/|x|$ and $L^4=0$;
		\item[•] $\Lb = \partial_t - \partial_r$ denotes the vector field tangent to the incoming null cones in $\R^{1+3}\times \S^1$. In components, $\Lb^0 = 1$, $\Lb^{\bmi} = -x^{\bmi}/|x|$ and $\Lb^4=0$;
		\item[•] $S^1$ and $S^2$ denote orthogonal vector fields spanning the tangent space of the spheres $t=const$, $r=const$, $y\in\S^1$;
		\item[•] $\q U = \{L, \Lb, S^1, S^2, \partial_y\}$ denotes the full null frame in $\R^{1+3}\times \S^1$;
		\item[•] $\q T = \{L, S^1, S^2, \partial_y\}$ denotes the tangent frame in $\R^{1+3}\times \S^1$;
		\item[•] $\q L =\{L\}$;
		\item[•] For any vector field $X$ and frame vector $U$, $X_U= X_\alpha U^\alpha$ where $X_\alpha = \overline{g}_{\alpha\beta}X^\beta$;
		\item[•] For any arbitrary vector field $X= X^\alpha \partial_\alpha= X^L L+X^{\Lb} \Lb + X^{S^1}S^1 + X^{S^2}S^2 + X^{\partial_y}\partial_y$ where $X^L = -(1/2) X_{\Lb}$, $X^{\Lb} = -(1/2) X_L$, $X^A = X_A$ for $A = S^1, S^2, \partial_y$;
		\item[•] For any $(0,2)$ tensor $\pi$ and two vector fields $X,Y$
		\[
		\pi_{XY} = \pi_{\alpha\beta}X^\alpha Y^\beta.
		\]
		For any two families $\q V, \q W$ of vector fields,
		$|\pi|_{\q V \q W}:= \sum_{V\in \q V, W\in \q W} |\pi_{VW}|$;
		\item[•] The metric $\overline{g}$ has the following form relative to the null frame, note $A, B \in \{S^1, S^2, \partial_y\}$,
		\[
		\overline{g}_{LL} = \overline{g}_{\Lb\Lb} = \overline{g}_{LA} = \overline{g}_{\Lb A} = 0, \quad \overline{g}_{L\Lb}=\overline{g}_{\Lb L}=-2, \quad \overline{g}_{AB} = \delta_{AB}.
		\]
		As concerns the inverse metric, we have
		\[
		\overline{g}^{LL} = \overline{g}^{\Lb\Lb} = \overline{g}^{LA} = \overline{g}^{\Lb A} = 0, \quad \overline{g}^{L\Lb}=\overline{g}^{\Lb L}=-1/2, \quad \overline{g}^{AB} = \delta^{AB}.
		\] 
	\end{itemize}
	
	\medskip
	\noindent \textsc{Admissible Vector Fields:}
	\begin{itemize}
		\item[•] $\{\Gamma\} = \{\Omega_{\bmi\bmj}, \Omega_{0\bmi}\}$ is the family of Klainerman vector fields, where $\Omega_{\bmi\bmj} = x_\bmi\partial_\bmj - x_\bmj \partial_\bmi, $ and $ \Omega_{0\bmi} = t\partial_\bmi + x_\bmi\partial_t$;
		\item[•] $\{Z\}= \{\partial_\mu, \Omega_{\bmi\bmj}, \Omega_{0\bmj}, \partial_y\}$ is the family of admissible vector fields;
		\item[•] For any multi-index $K= (I, J)$, we set $Z^K = \partial^I\Gamma^J$. If $|I|+|J|=n$ and $|J|=k$, we say that $K$ is a multi-index of type $(n,k)$.
	\end{itemize}

	\medskip
	\noindent
	\textsc{Commutators with the null frame:}
	\begin{itemize}
		\item[•] $ [\Omega_{0\bmj}, \partial_t+\partial_r ] = -\gd_\bmj - \frac{x_\bmj}{r}(\partial_t+\partial_r)$, \hspace{5pt} $ [\Omega_{0\bmj}, \gd_\bmk] = \big(-\delta_{\bmj\bmk} + \frac{x_\bmj x_\bmk}{r^2}\big)\big[(\partial_t+\partial_r) + \frac{1}{r}\Omega_{0r}\big] $\\[-5pt]
		\item[•] $[\Omega_{\bmi\bmj}, \partial_t+\partial_r]= 0$ ,\hspace{5pt} $ [\Omega_{\bmi\bmj}, \gd_\bmk] = -\delta_{\bmi\bmk}\gd_\bmj+ \delta_{\bmj\bmk} \gd_\bmi$ \\[-5pt]
		\item[•] $[\partial_k , \partial_t+\partial_r] = \frac{\delta_{jk}}{r}\partial_j - \frac{x_jx_k}{r^3}\partial_j$ \\[-5pt]
		\item[•] $[\Omega_{0\bmj}, \partial_y] =[\Omega_{\bmi\bmj}, \partial_y] = [\partial_\alpha, \partial_y] = 0$
	\end{itemize}
	\medskip
	\noindent \textsc{Commutators with the hyperbolic derivatives:}
	\begin{itemize}
		\item[•] $[\Omega_{0\bmj}, \partial_t] = - \partial_\bmj$, $[\Omega_{0\bmj}, \pb_{\bma}] = -\frac{x_{\bma}}{t}\pb_\bmj$, $[\Omega_{0\bmj}, \pb_4] = 0$\\[-5pt]
		\item[•] $[\Omega_{\bmi\bmj}, \partial_t] = [\Omega_{\bmi\bmj}, \pb_4] =0$,  $[\Omega_{\bmi\bmj}, \pb_{\bma}] = \delta_{\bma \bmj}\pb_\bmi - \delta_{\bmi\bma}\pb_\bmj$
	\end{itemize}
	
	\medskip
	\noindent \textsc{Exterior Region:}
	
	\begin{itemize}
		\item[•] $\widetilde{\hcal} = \{(t,x) : (t-1)^2 - r^2 = 1\}\times \S^1$ denotes the hyperboloid that separates the interior and exterior region. It asymptotically approaches the cone $\{t=r+1\}\times \S^1$;
		\item[•] $ \dext :=\{(t,x) : 2\le t \le 1+\sqrt{1+r^2} \}\times\S^1$ denotes the exterior region;
		\item[•] $\dext_T$ denotes the portion of exterior region in the time slab $[2, T)$;
		\item[•] $
		\Sext_t := \{x\in \m R^3 : |x|\ge \sqrt{(t-1)^2-1}\} \times \m S^1
		$ denotes a constant time slice in the exterior region;
	\end{itemize}
	
	\medskip
	\noindent \textsc{Interior Region:}
	
	\begin{itemize}
		\item[•] $\dint :=\{(t,x) : t\ge 1 + \sqrt{1 + r^2} \} \times \S^1$ denotes the interior region;
		\item[•]  $\hin_s := \{t^2 - r^2 = s^2 \text{ and } t\ge r+1\}\times \S^1$ denotes a truncated hyperboloid in $\R^{1+3}\times \S^1$;
		\item[•] $S_{s,r} := \hin_s \cap \{|x| = r\}$ is the two-sphere of radius $r$ on the hyperboloid $\hin_s$;
		\item[•] $\hin_{[s_0,s]} :=\{(t,x, y) \in \dint: s_0^2\le t^2-|x|^2\le s^2\}$ denotes the hyperbolic slab in the interior region between $\hin_{s_0}$ and $\hin_s$ when $s>2$;
		\item[•] $\hin_{[s_0,\infty)} :=\{(t,x, y) \in \dint: 2\le t^2-|x|^2\}$ is the unbounded portion of interior region above some hyperboloid $\hin_{s_0}$.
	\end{itemize}
	
	\subsection{From the null frame to hyperbolic derivatives}
	
	Below are some useful formulas relating the null framework $\q U$ to the hyperbolic derivatives $\pb_\bma$. We recall that $s = \sqrt{t^2-r^2}$. We have that
	\begin{equation}\label{null_hyp1}
		L = \Big(1-\frac{r}{t}\Big)\partial_t + \frac{x^\bmj}{r}\pb_\bmj, \quad \Lb =  \Big(1+\frac{r}{t}\Big)\partial_t - \frac{x^\bmj}{r}\pb_\bmj, \quad \gd_\bmj = \pb_\bmj -\frac{x_\bmj x^\bmi}{r^2}\pb_\bmi
	\end{equation}
	and
	\begin{equation} \label{null-hyp2}
		U V =c^{00}_{UV}\partial^2_t+ c^{a0}_{UV}\pb_a\partial_t + c^{0b}_{UV}\partial_t\pb_b + c^{ab}_{UV}\pb_a\pb_b + d^0_{UV}\partial_t + d^c_{UV}\pb_c, \qquad U, V\in \q U
	\end{equation}
	where 
	\begin{equation} \label{c_coefficients}
		\begin{gathered}
			c^{00}_{LL} = (1-r/t)^2,\hspace{5pt} c^{00}_{L\Lb}=(1-r^2/t^2), \hspace{5pt} c^{00}_{\Lb\Lb}=(1+r/t)^2, \hspace{5pt}
			c^{00}_{AU}=0 \text{ for } A = \{S^1, S^2, \partial_y\}, \\
			c^{04}_{L\partial_y} = 1-r/t, \hspace{5pt} c^{04}_{\Lb \partial_y}=(1+r/t), \hspace{5pt} c^{04}_{UV}=0 \text{ otherwise} \\
			c^{44}_{\partial_y\partial_y}=1, \hspace{5pt} c^{44}_{UV}=0 \text{ otherwise}
		\end{gathered}
	\end{equation}
	and 
	\begin{equation}\label{null-hyp-coeff}
		\begin{gathered}
			|\partial^I \Gamma^J c^{\alpha\beta}_{\Lb\Lb}|\lesssim_{IJ} (1+t+r)^{-|I|}, \quad |\partial^I \Gamma^J d^\gamma_{UV}|\lesssim_{IJ} (1+t+r)^{-1-|I|}, \qquad |I|\ge 0 \\
			|c^{\alpha\beta}_{LU}|\lesssim \frac{t^2-r^2}{t^2}, \quad |\partial^I \Gamma^J c^{\alpha\beta}_{TU}|\lesssim_{IJ} (1+t+r)^{-|I|},\qquad  |I|\ge 1.
		\end{gathered}
	\end{equation}
	For any tensor $\pi$ we have the following relations		
	\begin{equation}\label{H1-UV-cUV-exp}
		4(t/s)^2\pi^{UV} c^{00}_{UV} = \pi_{\Lb\Lb} \frac{s^2}{(t+r)^2} + \pi_{LL} \frac{(t+r)^2}{s^2} + \pi_{L\Lb} 
	\end{equation}		
	and
	\begin{equation} \label{eq:curved_part_semihyp}
		\begin{aligned}	
			\pi^{\mu\nu}\partial_\mu\partial_\nu  &= \pi^{UV} c^{\mu\nu}_{UV}\pb_\mu\pb_\nu  + \pi^{UV} d_{UV}^{\mu}\pb_\mu \\
			& = \pi^{UV}\big[ c_{UV}^{00}\partial_t^2 +  c_{UV}^{\bma \beta}\pb_\bma\pb_\beta+   c_{UV}^{\alpha \bmb}\pb_\alpha\pb_\bmb + c_{UV}^{4\beta} \partial_y \pb_\beta + d_{UV}^{\mu}\pb_\mu\big] \\	\pi^{\bmmu\bmnu}\partial_\bmmu\partial_\bmnu & = \pi^{UV}\big[ c_{UV}^{00}\partial_t^2 +  c_{UV}^{\bma \bmbeta}\pb_\bma\pb_\bmbeta+   c_{UV}^{\bmalph \bmb}\pb_\bmalph\pb_\bmb + d_{UV}^{\bmmu}\pb_\bmmu\big] .
		\end{aligned}
	\end{equation}
	For any smooth function $u = u(t,x)$, we have the following inequalities
	\begin{equation}\label{der_null_semihyp}
		\begin{gathered}
			|\gd u|\lesssim |\pb_x u|, \qquad |\op u|\lesssim \Big(\frac{s}{t}\Big)^2|\partial u | + |\pb_x u| \\
			|\op \partial u|\lesssim \Big(\frac{s}{t}\Big)^2|\partial^2 u | + |\pb_x \partial u|\lesssim \Big(\frac{s}{t}\Big)^2|\partial^2 u | + \frac{1}{t}|\partial Z^{\le 1} u| \\
			|\gd \gd u|\lesssim \frac{1}{t}|\pb_x Z^{\le 1} u|, \qquad |\op \op u|\lesssim \Big(\frac{s}{t}\Big)^4|\partial^2 u | + \Big(\frac{s}{t}\Big)^2\frac{1}{t}|\partial Z^{\le 1}u| + \frac{1}{t}|\pb_x Z^{\le 1}u|.
		\end{gathered}
	\end{equation}

	\subsection{Outline of the paper}
	The rest of the paper is organized in three main sections. Section \ref{sec:wave_condition} introduces some properties the metric coefficients inherit from the wave condition and which will be used throughout. Section \ref{sec:exterior} is devoted to perform the bootstrap argument in the exterior region and hence to prove the global existence of the solution to \eqref{h_equations} there. In section \ref{sec:interior} we perform the bootstrap argument in the interior region and conclude the proof of the main theorem. Two appendix sections follow: in section \ref{sec:energy_appendix} we state and prove the exterior and interior energy inequalities, while section \ref{sec:weighted_Sobolev_Hardy} contains a list of weighted Sobolev and Hardy inequalities.

	\section{The wave condition}\label{sec:wave_condition}
	
	The metric solution $g$ to \eqref{EE} satisfies, when written in harmonic coordinates $\{x^\mu\}_\mu$, the \emph{wave coordinate condition}
	\[
	g^{\mu\nu}{\Gamma^\lambda_{\mu\nu}} =0, \quad \lambda = \overline{0,4}
	\]
	where $\Gamma^\lambda_{\mu\nu}$ are the Christoffel symbols of $g$ in the coordinates $\{x^\mu\}_\mu$. The above equations are equivalent to each of the following
	\begin{equation} \label{wave_condition1}
		\partial_\mu\big(g^{\mu\nu}\sqrt{|\det g|}\big)=0, \quad g^{\alpha\beta}\partial_\alpha g_{\beta\nu} = \frac12 g^{\alpha\beta}\partial_\nu g_{\alpha\beta}, \quad \partial_\alpha g^{\alpha\nu} = \frac12 g_{\alpha\beta} g^{\nu\mu}\partial_\mu g^{\alpha\beta}, \quad \nu=\overline{0,4}.
	\end{equation}
	These relations are particularly useful when written with respect to the null framework, as they allow us to recover additional information on metric coefficients $H_{LT}$ (and hence on $h_{LT}$) for any $T\in \q T$, and to show that their derivatives have a special behavior compared to those of general coefficients $H^{\mu\nu}$. This is the content of the following Lemmas, which are presented in a slightly different form than the ones in \cite{LR10}.
	
	\begin{lemma}\label{lem:wave_cond_nullframe}
		Let $g$ be a Lorentzian metric satisfying the wave coordinate condition relative to a coordinate system $\{x^\mu\}_{\mu=0}^4$. Let $K = (I,J)$ be any multi-index with positive length and assume that the perturbation tensor $H^{\mu\nu} = g^{\mu\nu} - \bar{g}^{\mu\nu}$ satisfies the following
		\[
		| Z^{K'}H|\lesssim C, \quad \forall\,  |K'|\le \floor{|K|/2}.
		\]
		Then
		\begin{equation}
			|\partial H|_{\q L \q T}\lesssim |\op H| + |H||\partial H| \lesssim \Big(\frac{s}{t}\Big)^2 |\partial H| + |\pb_{xy} H| + |H||\partial H| \label{ineq_Hlt} 
		\end{equation}
		and in any region where $r\gtrsim t\gtrsim 1$
		\begin{equation} \label{ineq_HLT_higher}
			\begin{aligned}
				&| Z^K\partial H|_{\q L \q T} + |\partial Z^KH|_{\q L \q T}   \lesssim \sum_{|K'|\le |K|} (|\op  Z^{K'}H| + r^{-1}| Z^{K'}H|)  + \hspace{-10pt} \sum_{|K_1|+|K_2|\le |K|}\hspace{-10pt} | Z^{K_1}H| |\partial Z^{K_2}H| \\
				& \lesssim \sum_{|K'|\le |K|} \Big(\frac{s}{t}\Big)^2 |\partial  Z^{K'}H| + |\pb Z^{K'} H| + r^{-1}| Z^{K'}H|   + \sum_{|K_1|+|K_2|\le |K|}\hspace{-10pt} | Z^{K_1}H| |\partial Z^{K_2}H|
			\end{aligned} 
		\end{equation}
		Similar estimates hold for the metric tensor $h_{\mu\nu} = g_{\mu\nu} - \bar{g}_{\mu\nu}$.
	\end{lemma}
	\begin{proof}
		We write $g^{\mu\nu}$ in terms of the perturbation metric $H^{\mu\nu}$. From the following equality
		\[
		g^{\mu\nu}\sqrt{|\text{det}g|} = (\bar{g}^{\mu\nu} + H^{\mu\nu})\Big(1-\frac{1}{2}\tr H + \mathcal{O}(H^2)\Big)
		\]
		and the wave condition \eqref{wave_condition1} we obtain that
		\begin{equation}\label{divergence_wave_cond}
			\partial_\mu\Big(H^{\mu\nu} - \frac{1}{2}\bar{g}^{\mu\nu}\tr H +\mathcal{O}^{\mu\nu}(H^2)\Big)=0, \quad \text{where } \mathcal{O}^{\mu\nu}(H^2)=\mathcal{O}(|H|^2).
		\end{equation}
		The divergence of a vector field can be expressed relative to the null frame as follows
		\begin{equation} \label{divergence_nullframe}
			\partial_\mu F^\mu = 
			L_\mu  \partial_\ub F^\mu - \L_\mu \partial_u F^\mu + A_\mu \partial_A F^\mu, \quad A\in \{S_1, S_2, \partial_y\}
		\end{equation}
		so setting $\tilde{H}^{\mu\nu}:=H^{\mu\nu} -\frac{1}{2}\bar{g}^{\mu\nu}\tr H$ and contracting \eqref{divergence_wave_cond} with any $T\in \q T$ we deduce that
		\begin{equation}\label{eq:partial_HLL}
			\partial_\ub H_{LT} =  \partial_\ub \tilde{H}_{LT} =  \partial_u\tilde{H}_{\L T}-\partial_A \tilde{H}_{AT}+ {\q O}(H\cdot\partial H).	
		\end{equation}
		The first of the above equalities follows from the fact that $\bar{g}^{LL}=g^{LA}= 0$. Relation \eqref{eq:partial_HLL} and the first two inequalities in \eqref{der_null_semihyp} imply immediately \eqref{ineq_Hlt}.
		
		We now recall the commutators between any admissible vector field $Z$ and the null frame, which can be summarized in the following formula
		\[
		[ Z, \op_\alpha] = \sum_{\bmbeta=0}^3 c_{ Z \alpha}^{\bmbeta}\op_{\bmbeta}+ \sum_{\bmi=1}^3d_{ Z \alpha}^{\bmi} \frac{\partial_{\bmi}}{r} +  e_{ Z \alpha} \frac{\Omega_{0r}}{r}, \qquad c_{ Z \alpha}^{\bmbeta}, d_{ Z \alpha}^{\bmi}, e_{ Z \alpha} = \q O \big(\frac{x}{r}\big)
		\]
		where $c_{ Z \alpha}^{\bmbeta}, d_{ Z \alpha}^{\bmi}, e_{ Z \alpha}$ are smooth homogeneous functions of $x$ such that $c_{ \partial \alpha}^{\bmbeta}=e_{\partial \alpha}=0$, $d^{\bmi}_{\Gamma \alpha} = 0$ and
		\[
		|\partial^I_x c_{ Z \alpha}^{\bmbeta}| + |\partial^I d_{ Z \alpha}^{\bmi}| + |\partial^I e_{Z\alpha}|\lesssim r^{-|I|}, \qquad |I|\ge 0.	
		\]  
		Using an induction argument on $|K|$, one can show that for any sufficiently smooth function $w$ the following inequality holds true whenever $r\gtrsim t$
		\begin{equation}\label{commut_tangent_ext}
			|[ Z^K, \op] w|\lesssim \sum_{|K'|< |K|}( |\op  Z^{K'}w|+ r^{-1}|\partial Z^{K'}w|) + \sum_{|K'|\le |K|}r^{-1} | Z^{K'}w|.
		\end{equation}
		As concerns the commutators with the transverse vector field, we simply have
		\begin{equation}\label{commut_transvers_ext}
			|[ Z^K, \partial_t-\partial_r]w|\lesssim \sum_{|K'|<|K|}|\partial Z^{K'}w|\lesssim \sum_{|K'|<|K|}| \partial_\ub Z^{K'}w| + |\op  Z^{K'}w|.
		\end{equation}
		In order to obtain \eqref{ineq_HLT_higher}, we apply $ Z^K$ vector fields to both sides of equality \eqref{eq:partial_HLL}. Using \eqref{commut_tangent_ext} we find  that
		\[
		\begin{aligned}
			| Z^K\partial_\ub H|_{\q L \q T} & \lesssim \sum_{|K'|\le |K|} (|\op  Z^{K'}H| + r^{-1}| Z^{K'}H|) + \sum_{|K_1|+|K_2|\le |K|} | Z^{K_1}H| |\partial Z^{K_2}H|
		\end{aligned}
		\]
		which, together with \eqref{commut_transvers_ext}, yields
		\[
		\begin{aligned}
			|\partial_\ub  Z^K H|_{\q L \q T}& \lesssim  \sum_{|K'|\le |K|} (|\op  Z^{K'}H| + r^{-1}| Z^{K'}H|)+ \sum_{|K_1|+|K_2|\le |K|} | Z^{K_1}H| |\partial Z^{K_2}H|\\
			&   + \sum_{|K'|<|K|}| \partial_\ub Z^{K'}H|_{\q L \q T} .
		\end{aligned}
		\]
		The conclusion of the proof of the first inequality in \eqref{ineq_HLT_higher} then follows by induction on $|K|$. The latter follows using also \eqref{der_null_semihyp}. Finally, inequalities \eqref{ineq_Hlt} and \eqref{ineq_HLT_higher} for $h$ simply follow from the equality $H^{\mu\nu} = -h^{\mu\nu} + \mathcal{O}(h^2)$.
	\end{proof}
	
	Inequalities \eqref{ineq_Hlt} and \eqref{ineq_HLT_higher} hold true also for the tensor $H^{1, \mu\nu}$ introduced in \eqref{dec_H}.
	
	\begin{lemma} \label{lem:wave_cond_h1}
		Under the same assumptions of the previous lemma, we have that
		\begin{equation}\label{ineq_H1lt_higher}
			\begin{aligned}
				|Z^K\partial H^1|_{\q L \q T} + |\partial Z^K H^1|_{\q L \q T}  \lesssim \sum_{|K'|\le |K|} (|\op  Z^{K'}H^1| + r^{-1}| Z^{K'}H^1|)  \\
				+  \sum_{|K_1|+|K_2|\le |K|}\hspace{-10pt} | Z^{K_1}H^1| |\partial Z^{K_2}H^1| + \frac{M\chi_0\big(t/2\le r\le 3t/4\big)}{(1+t+r)^{2}} \\
			\end{aligned}
		\end{equation}
		and
		\begin{equation}\label{wave_cond_interior}
			\begin{aligned}
				|Z^K\partial H^1|_{\q L \q T} + |\partial Z^K H^1|_{\q L \q T}  \lesssim \sum_{|K'|\le |K|}\Big(\frac{s}{t}\Big)^2 |\partial Z^{K'}H^1| +  |\pb  Z^{K'}H^1| + r^{-1}| Z^{K'}H^1| \\
				+  \sum_{|K_1|+|K_2|\le |K|}\hspace{-10pt} | Z^{K_1}H^1| |\partial Z^{K_2}H^1| + \frac{M\chi_0\big(t/2\le r\le 3t/4\big)}{(1+t+r)^{2}} \\
			\end{aligned}
		\end{equation}
		where $\chi_0\big(t/2\le r\le 3t/4\big)$ is a cut-off function supported for $t/2\le r\le 3t/4$. Similar estimates hold true for $h^1$.
	\end{lemma}
	\begin{proof}
		We set $\tilde{H}^{0,\mu\nu}:= H^{0,\mu\nu}-\frac{1}{2}\bar{g}^{\mu\nu}\text{tr}(H^0)$ and derive from the definition of $H^0$ that
		\begin{equation}\label{div_H0}
			\partial_\mu \tilde{H}^{0,\mu\nu} = 2\chi'\Big(\frac{r}{t}\Big)\chi(r)\frac{M}{t^2}\delta^{\nu 0}. 
		\end{equation}
		We inject the above formula into \eqref{divergence_wave_cond} and obtain that
		\[
		Z^K \partial_\mu (H^{1,\mu\nu}) = -Z^K \partial_\mu \mathcal{O}^{\mu\nu}(H^2) -  Z^K \Big(2\chi'\Big(\frac{r}{t}\Big)\chi(r)\frac{M}{t^2}\delta^{\nu 0}\Big).
		\] 
		Then the result of the statement follows using the same argument as in previous lemma's proof. Furthermore, from 
		\eqref{dec_H} a similar inequality can be proved for $h^1_{\mu\nu}$.
	\end{proof}
	
	\section{The Exterior Region} \label{sec:exterior}
	
	The goal of this section is to prove the existence in the exterior region $\dext$ of the solution $h^1_{\alpha\beta}$ of \eqref{h_equations} with data satisfying the hypothesis of theorem \ref{thm:Main}. The proof is based on a bootstrap argument in which the a-priori assumptions are bounds on the higher order weighted energies of $h^1_{\alpha\beta}$, introduced below.
	
	For any fixed $\kappa>0$, we define the exterior weighted energy functional of $\hab$ as 
	\begin{multline*}
		\enext(t, \hab) = \iint_{\{|x|\ge t-1\}\times \S^1} (2+|x|-t)^{1+2\kappa} |\nabla_{txy} \hab(t,x,y)|^2 dxdy\ \\
		+ \int_2^t\iint_{\{|x|\ge \tau-1\}\times \S^1} (2+|x|-t)^{2\kappa} |\overline{\nabla} \hab(\tau, x,y)|^2 dxdyd\tau
	\end{multline*}
	and denote $\enext(t, h^1) = \sum_{\alpha,\beta}\enext(t, \hab)$. 
	We fix $N\in \N$ with $N\ge 7$ and assume the existence of a positive constant $C_0$ and of some small parameters $0<\sigma<\kappa/3\ll 1$ such that the solution $h^1$ of \eqref{h_equations} exists in $\dext_{T_0}$ and for all $t\in [2, T_0)$ it satisfies
	\begin{align}
		& \enext(t, Z^{\le N}h^1)^{1/2}\le 2C_0\ep t^{\sigma} \label{Boot1_ext}\\
		& E^{e, 1+\kappa}(t, \partial Z^{\le N}h^1)^{1/2}\le 2C_0\ep t^\sigma. \label{Boot2_ext}
	\end{align}
	The result we want to prove here affirms the following
	
	\begin{proposition} \label{prop:bootstrap_ext}
		Let $N\in \N$ with $N\ge 6$ be fixed. There exists a constant $C_0$ sufficiently large, $0<\epsilon_0\ll 1$ sufficiently small and a universal positive constant $C$ such that, for every $0<\ep<\ep_0$ if $h^1$ is a solution of \eqref{h_equations} in the time interval $[2, T_0)$ and satisfies the bounds \eqref{Boot1_ext}-\eqref{Boot2_ext} for all $t\in [2, T_0)$, then in the same interval it actually satisfies
		\begin{align}
			& \enext(t, Z^{\le N}h^1)^{1/2}\le C_0\ep t^{\frac{\sigma}{2} + CC_0\epsilon} \label{enhanced_Boot1_ext}\\
			& E^{e, 1+\kappa}(t, \partial Z^{\le N}h^1)^{1/2}\le C_0\ep t^{\frac{\sigma}{2} + CC_0\epsilon}. \label{enhanced_Boot2_ext}
		\end{align}
	\end{proposition} 
	The time $T_0$ in the statement of the above proposition is arbitrary and one can hence infer that the solution exists globally in $\dext$.
	We also observe that, as a consequence of the energy assumptions \eqref{Boot1_ext}-\eqref{Boot2_ext}, there exists an integrable function $l\in L^1([2, T_0))$ such that
	\begin{align}
		\label{boot1_ex}
		&\big\| (2+r-t)^{\frac{1}{2}+\kappa}\ \partial Z^{\le N} h^1 \big\|_{L^2 (\Sext_t)} \leq 2C_0 \ep t^\sigma \\  
		\label{boot2_ex}
		&		\big\| (2+r-t)^\kappa\ \op Z^{\le N} h^1 \big\|_{L^2(\Sext_t)} \le 2C_0\ep \sqrt{l(t)}t^{\sigma} \\
		\label{boot3_ex}
		&\big\| (2+r-t)^{\frac{3}{2}+\kappa}\ \partial^2 Z^{\le N} h^1 \big\|_{L^2 (\Sext_t)} \leq 2C_0 \ep t^\sigma \\  
		\label{boot4_ex}
		&		\big\| (2+r-t)^{1+\kappa}\ \op \partial Z^{\le N} h^1 \big\|_{L^2(\Sext_t)} \le 2C_0\ep \sqrt{l(t)}t^{\sigma} .
	\end{align}

	The first step to recover the enhanced bounds \eqref{enhanced_Boot1_ext}-\eqref{enhanced_Boot2_ext} is to compare the equation satisfied by the differentiated unknown $Z^K h^1_{\alpha\beta}$ for any $K=(I,J)$ with $|K|\le N+1$, with the linear inhomogeneous equation \eqref{linearized_wave}.
	The commutation of $Z^K$ with equation \eqref{h_equations} shows that $Z^Kh^1_{\alpha\beta}$ solves
	\begin{equation} \label{h1_eqt_higher}
		\tilde{\Box}_g Z^Kh^1_{\alpha\beta} = F^K_{\alpha\beta} + F^{0,K}_{\alpha\beta}, \qquad F^{0,K}_{\alpha\beta} = Z^K\tilde{\Box}_g h^0_{\alpha\beta}
	\end{equation}
	with source term $F^K_{\alpha\beta}$ given by
	\begin{equation}\label{source_ext_high}
		F^K_{\alpha\beta} = Z^KF_{\alpha\beta}(h)(\partial h, \partial h) - [Z^K, H^{\mu\nu}\partial_\mu\partial_\nu]h^1_{\alpha\beta} .
	\end{equation}
	The second step consists in recovering suitable pointwise decay estimates and $L^2$ estimates for tensors $h_{\alpha\beta}$ and $H^{\alpha\beta}$ and their derivatives.  Our aim is in fact to apply energy inequality \eqref{energy_ineq_ext} with $\Wf = Z^Kh^1_{\alpha\beta}$, $\Ff = F^K_{\alpha\beta} + F^{0,K}_{\alpha\beta}$ and $w(q) = (2+r-t)^{\frac{1}{2}+i+\kappa}$ with $i=0,1$ depending on $K$. Such estimates allow us, on the one hand, to justify the use of \eqref{energy_ineq_ext} and, on the other, to suitably estimate the different contributions to the right hand side of such energies, and hence to propagate \eqref{Boot1_ext}-\eqref{Boot2_ext}. The derivation of these bounds is the content of the following subsections.

	\subsection{Pointwise bounds} \label{subsec:pointwise_exterior}
	
	A first set of pointwise decay bounds for the metric perturbation $h^1_{\alpha\beta}$, as well as for tensor $H^{1,\alpha\beta}$, are obtained from the a-priori energy assumptions \eqref{Boot1_ext}-\eqref{Boot2_ext} via the weighted Sobolev and Hardy embeddings stated in appendix \ref{sec:weighted_Sobolev_Hardy}. 
	As concerns the mass term, a straightforward computation using directly the expression of $h^0_{\alpha\beta}$ in \eqref{h0_def} shows that for all $(t,x,y)\in \R^{1+3}\times \S^1$
	\begin{equation}\label{h0_estimate}
		\sup_{\S^1} |\partial^I Z^J h^0|\lesssim \ep (1+r)^{-1-|I|}.
	\end{equation}

	\begin{proposition}\label{prop:pointwise_ext}
		Let us define the weighted pointwise norm
		\[
		|u(t)|_\lambda : = \sup_{(x,y)\in \Sext_t} (1+t+r)(2+r-t)^{1 + \lambda} |u(t,x,y)|.
		\]
		Assume that the solution $h^1_{\alpha\beta}$ of \eqref{h_equations} exists in the time interval $[2, T_0)$ and satisfies \eqref{Boot1_ext}-\eqref{Boot2_ext} for all $t\in [2, T_0)$. Then the following estimates hold true in $\dext_{T_0}$
		\begin{equation} \label{KS1_ext}
			|\partial Z^{\le N-3}h^1|_{\kappa} +|\partial^2 Z^{\le N-3}h^1|_{1/2+\kappa} \lesssim C_0\ep t^\sigma  
		\end{equation}
		\begin{equation} \label{KS2_ext}
			|\op Z^{\le N-3} h^1|_{\kappa-1/2} + |\op \partial Z^{\le N-3}h^1|_{\kappa}  \lesssim C_0\ep t^\sigma \sqrt{l(t)}
		\end{equation}
		\begin{equation} \label{KS3_ext}
			|Z^{\le N-2} h^1|_{\kappa -1}\lesssim C_0\ep t^{\sigma}
		\end{equation}
		\begin{equation}\label{KS2.1_ext}
			|Z^{\le N-3}  h^{1, \natural}|_{\kappa -1/2}\lesssim C_0\ep t^{\sigma}.
		\end{equation}
		Estimates \eqref{KS1_ext}-\eqref{KS3_ext} hold true also for tensor $H^{1,\alpha\beta}$. 
	\end{proposition}
	
	\begin{proof}
		Estimate \eqref{KS1_ext} (resp. \eqref{KS2_ext}) for the second order derivatives $\partial^2$ (resp. $\op \partial $) of  $Z^{\le N-3}h^1$ follows from the energy bound \eqref{boot3_ex} (resp. \eqref{boot4_ex}) and from inequality \eqref{sobolev2_ext} applied with $\beta = 1+\lambda$ and $\lambda = 1/2+\kappa$ (resp. $\lambda =\kappa$).
		
		Estimate \eqref{KS1_ext} (resp. \eqref{KS2_ext}) for the first order derivatives $\partial$ (resp. $\op$) of $Z^{\le N-3}h^1$ follows from the energy bound \eqref{boot3_ex} (resp. \eqref{boot4_ex}) and inequality \eqref{sob+hardy1_ext} applied with $\beta = 1+\kappa$  (resp. $\beta= 1/2 + \kappa$), and estimate \eqref{KS2.1_ext} follows using in addition the Poincar\'e inequality.

		Estimate \eqref{KS3_ext} follows from the energy bound \eqref{boot1_ex} and inequality \eqref{sob+hardy1_ext} with $\beta = \kappa$.

		Finally, one can show that estimates \eqref{KS1_ext}-\eqref{KS3_ext} hold true for $H^{1,\alpha\beta}$ using \eqref{H_def}.
	\end{proof}
	
	As a result of inequality \eqref{ineq_H1lt_higher} and the pointwise bounds we just obtained, we can show that the metric coefficients $h^1_{LT}$ satisfy enhanced pointwise decay estimates compared to those in Proposition \ref{prop:pointwise_ext}.
	
	\begin{proposition}
		Under the assumptions of Proposition \ref{prop:pointwise_ext}, we have
		\begin{gather} 
			|\partial Z^{\le N-3} h^1_{LT}(t,x,y)|\lesssim C_0\ep \big[(1+t+r)^{-1+\sigma}\sqrt{l(t)}(2+r-t)^{-\frac12-\kappa} +  (1+t+r)^{-2+2\sigma}(2+r-t)^{-\kappa}\big]\label{der_hLT_ext} \\
			|Z^{\le N-4}h^1_{LT}(t,x,y)| \lesssim C_0\ep (1+t+r)^{-\frac32+2\sigma}(2+r-t)^{\frac12-\kappa}\label{est_hLT2}
		\end{gather}
		and
		\begin{equation} \label{est_hLT1}
		\|Z^{\le N} h^1_{LT}(t,r)\|_{L^2(\m S^2\times \m S^1)}\lesssim \frac{C_0 \ep t^\sigma}{t^{\kappa-\mu}	(2+r-t)^\mu } \Big(\frac{\sqrt{l(t)}}{r^{1/2}} + \frac{1}{r} \Big) .
		\end{equation}
		The same bounds are satisfied by $H^1_{LT}$.
	\end{proposition}
	\begin{proof}
		From relation \eqref{H_def} and pointwise bounds \eqref{h0_estimate}, \eqref{KS1_ext}, \eqref{KS3_ext}, it is clear that the estimates \eqref{der_hLT_ext} and \eqref{est_hLT2} for $h^1_{LT}$ are also satisfied by $H^1_{LT}$.
		
		Bound \eqref{der_hLT_ext} follows immediately from inequality \eqref{ineq_H1lt_higher} coupled with \eqref{KS1_ext}-\eqref{KS3_ext}.
	
		The proof of \eqref{est_hLT2} requires more work because a naive integration of \eqref{der_hLT_ext} along the integral curves of $\partial_t-\partial_r$ does not produce the required result due to the factor $\sqrt{l(t)}$ (it would if this was replaced by the explicit decay $t^{-1/2}$).
		Of course, the estimate is satisfied in the region where $r\ge 2t$ simply after \eqref{KS3_ext}. We then restrict our attention to the portion of exterior region for which $r<2t$ and proceed as follows:
		\begin{itemize}
			\item[-] first, we recover a better bound for $\partial Z^{\le N-4}h^1_{LT}$ than the one in \eqref{der_hLT_ext}, in which $\sqrt{l(t)}$ is replaced by a decay $t^{-\frac12 +}$. This is obtained by the integration of $\pb_r\partial Z^{\le N-4}h^1_{LT}$ along hyperboloids in some dyadic time slab, where $\pb_r = \frac{x^\bmj}{tr}\Omega_{0\bmj}$;
			\item[-] then, we deduce the desired estimate on $Z^{\le N-4}h^1_{LT}$ by integration of the bounds obtained in step 1 along the integral curves of $\partial_{\ub}$.
		\end{itemize}
		\emph{Step 1.} From the relation $\pb_r = \frac{x^\bmj}{tr}\Omega_{0\bmj}$ and inequality \eqref{der_hLT_ext}, we see that in fact 
		\begin{multline*}
			|\pb_r\partial Z^{\le N-4}h^1_{LT}| \lesssim (1+t+r)^{-1}|\partial Z^{\le N-3}h^1_{LT}| \\
			\lesssim C_0\ep (1+t+r)^{-2+\sigma}\sqrt{l(t)}(2+r-t)^{-\frac12 -\kappa} + C_0\ep (1+t+r)^{-3+2\sigma}(2+r-t)^{-\kappa}.
		\end{multline*}
		Moreover, thanks to the pointwise bound \eqref{KS1_ext} we have that on the cone $r=2t$
		\[
		|\partial Z^{\le N-4}h^1_{LT}(t, x, y)|\lesssim C_0 \ep (1+t+r)^{-2-\kappa+\sigma}.
		\]
		We dyadically decompose the time interval $[2, T_0) = \cup_{k=1}^{k_0} [2^k, 2^{k+1})\cap [2, T_0)$ where $k_0\sim \ln_2 T_0$ and denote $\q C^e_k$ the portion of the exterior region in the time slab $[2^k, 2^{k+1})$, so that $\q C^e = \cup_{k=1}^{k_0} \q C^e_k$. Inequality \eqref{der_hLT_ext} and the fact that $l\in L^1([2, T_0))$ imply the existence, for every fixed $k$, of a time $\tau_k\in [2^k, 2^{k+1})\cap [2, T_0)$ such that
		\[
		|\partial Z^{\le N-4}h^1_{LT}(\tau_k, x, y)|\lesssim C_0\ep \tau_k^{-\frac32+\sigma}(2+r-\tau_k)^{-\frac12-\kappa} + C_0\ep \tau_k^{-2+2\sigma}(2+r-\tau_k)^{-\kappa}.
		\]
		For every fixed $(t,x,y)\in \q C^e_k$, we then integrate $\pb_r\partial Z^{\le N-4}h^1_{LT}$ along the integral curve $\tau \mapsto \gamma(\tau)$ of $\pb_r$ passing through $(t,x,y)$\footnote{These are the hyperboloids $\{\tau^2 - |w|^2 = t^2-r^2\}$. If $t=r$ they degenerate into the cone $\{\tau  - |w| = t-r\}$.} until its first intersection with $\{\tau = \tau_k \}\cup \{|w| = 2\tau\}$. We denote $(\tau^*_k, x^*_k, y)$ the point at which such an intersection occurs first and observe that $\tau^*_k\sim 2^k\sim t$. We deduce that
		\begin{equation}\label{der_ZhLT_improved}
			\begin{aligned}
				|\partial Z^{\le N-4}h^1_{LT}(t,x,y)|  \le |\partial Z^{N-4}h^1_{LT}(\tau^*_k, x^*_k, y)| + \int_{\tau^*_k}^t |\pb_r\partial Z^{\le N-4}h^1_{LT}(\gamma(\tau))| d\tau  \\
				\lesssim C_0\ep (1+t+r)^{-\frac32+\sigma}(2+r-t)^{-\frac12-\kappa} + C_0\ep (1+t+r)^{-2+2\sigma}(2+r-t)^{-\kappa}.
			\end{aligned}
		\end{equation}

		\emph{Step 2.}
		We now integrate \eqref{der_ZhLT_improved} along the integral lines of $\partial_{\ub}$, up to $\mathcal{B}=\{r=2t\}\cup \{t=2\}$. After \eqref{KS3_ext}, we have that 
		\[
		|Z^{\le N-2}h^1|_\mathcal{B}\lesssim C_0\ep (1+t+r)^{-2}(2+r-t).
		\]
		Therefore, from \eqref{der_ZhLT_improved} we get that
		\[
		\begin{aligned}
			|Z^{\le N-4}h^1_{LT}(t,x,y)| \le |Z^{\le N-4}h^1_{LT}(\lambda^*_k, x^*_k, y)| + \int_{\lambda^*_k}^t | \partial_{\ub} Z^{\le N-4}h^1_{LT}(\gamma(\tau))|d\tau \\
			\lesssim C_0 \ep (1+t+r)^{-2}(2+r-t) + (1+ t + r)^{-\frac32+\sigma}\int_{\lambda^*_k}^t C_0\ep (2+t+r-2\tau)^{-\frac12-\kappa} d\tau \\
			+ (1+t+r)^{-2+2\sigma}\int_{\lambda^*_k}^t C_0\ep (2+t+r-2\tau)^{-\kappa} d\tau
		\end{aligned}
		\]
		and
		\[
		|Z^{\le N-4}h^1_{LT}(t,x,y)| \lesssim C_0\ep (1+t+r)^{-\frac32+2\sigma}(2+r-t)^{\frac12-\kappa}.
		\]			
	
	\medskip	
As concerns the proof of \eqref{est_hLT1}, 
	we begin by applying inequality \eqref{sob+hardy2_ext} with $\beta = \mu$ to $Z^{\le N}h^1_{LT}$. We get that
			\[
			(2+r-t)^{2\mu} r^2	 \| Z^{\le N} h^1_{LT}\|_{L^2(\m S^2\times \m S^1)}^2 
			\lesssim \iint_{\Sext_t} (2+r-t)^{1+2\mu}(\partial  Z^{\le N} h^1_{LT})^2 dx dy.
			\]
			We decompose the above right hand side into
			$$\iint_{\Sext_t\setminus \Sext_{2t}}(2+r-t)^{1+2\mu}(\partial  Z^{\le N} h^1_{LT})^2 dx dy+\iint_{\Sext_{2t}}(2+r-t)^{1+2\mu}(\partial  Z^{\le N} h^1_{LT})^2 dx dy.$$
			The integral over $\Sext_{2t}$ is simply estimated using the energy bound \eqref{boot1_ex} as follows
			\[
			\iint_{\Sext_{2t}}(2+r-t)^{1+2\mu}(\partial  Z^{\le N}h^1_{LT})^2 dx dy
			\lesssim t^{2(\mu-\kappa)}\iint_{\Sext_{t}}(2+r-t)^{1+2\kappa}|\partial  Z^{\le N} h^1|^2 dx dy\lesssim \ep^2 t^{2(\mu-\kappa+\sigma)}.
			\]		
			The integral over $\Sext_t\setminus \Sext_{2t}$ is estimated using \eqref{ineq_H1lt_higher} for $h^1$ 
			\begin{align*}	&\iint_{\Sext_t\setminus \Sext_{2t}}(2+r-t)^{1+2\mu}(\partial Z^{\le N} h^1_{LT})^2 dx dy \\
				& \lesssim 	\iint_{\Sext_t\setminus \Sext_{2t}}  (2+r-t)^{1+2\mu} M^2 \chi_0^2(t/2\le r\le 3t/4)r^{-4}dxdy\\
				&+ \iint_{\Sext_t\setminus \Sext_{2t}}(2+r-t)^{1+2\mu}\Big(|\op  Z^{\le N}h^1|^2 + r^{-2}|  Z^{\le N}h^1|^2  +\hspace{-15pt} \sum_{|K_1|+|K_2|\le N}\hspace{-15pt} | Z^{K_1}h^1|^2 |\partial Z^{K_2}h^1|^2\Big)  dx dy,
			\end{align*}
			where $\chi_0(t/2\le r\le 3t/4)$ is a smooth cut-off function supported in $t/2\le r\le 3t/4$. We observe that the portion of such support contained in the exterior region is bounded.
			Therefore we get the following:
			
			\noindent - from the smallness assumption on $M$ and the above observation
			\[
			\iint_{\Sext_t\setminus \Sext_{2t}} (2+r-t)^{1+2\mu} M^2 \chi_0^2(t/2\le r\le 3t/4) r^{-4}dxdy\lesssim \ep^2 t^{-4};
			\] 
			- from the energy bound \eqref{boot2_ex}
			\begin{multline*}
				\iint_{\Sext_t\setminus \Sext_{2t}}(2+r-t)^{1+2\mu} |\op  Z^{\le N}h^1|^2 dxdy \\
				\lesssim  t^{1+2(\mu-\kappa)} \iint_{\Sext_t} (2+r-t)^{2\kappa} 	|\op  Z^{\le N}h^1|^2 dxdy \lesssim \ep^2  t^{1+2(\mu-\kappa+\sigma)} l(t);
			\end{multline*}
			- from inequality \eqref{hardy_ext} with $\beta=2\kappa-1$ and the energy bound \eqref{boot1_ex}
			\[
			\begin{aligned}
				& \iint_{\Sext_t\setminus \Sext_{2t}}(2+r-t)^{1+2\mu}r^{-2} |  Z^{\le N}h^1|^2 dxdy  \lesssim t^{2(\mu-\kappa)}\iint_{\Sext_t\setminus \Sext_{2t}}(2+r-t)^{2\kappa-1}| Z^{\le N}h^1|^2 dxdy \\
				&\lesssim t^{2(\mu-\kappa)}\iint_{\Sext_t}(2+r-t)^{1+2\kappa}|\partial  Z^{\le N}h^1|^2dxdy\lesssim \ep^2 t^{2(\mu-\kappa+\sigma)};
			\end{aligned}
			\]
			- from the energy bound \eqref{boot1_ex} and the pointwise bound \eqref{KS3_ext} that
			\[
			\begin{aligned}
				\sum_{\substack{|K_1|+|K_2|\le N\\ |K_1|\le \floor{N/2}}}&\iint_{\Sext_t\setminus \Sext_{2t}} (2+r-t)^{1+2\mu} | Z^{K_1}h^1|^2 |\partial Z^{K_2}h^1|^2 dxdy \\
				& \lesssim \ep^2 t^{2\sigma} \iint_{\Sext_t\setminus \Sext_{2t}} (2+r-t)^{1+2(\mu-\kappa)}r^{-2} |\partial  Z^{\le N}h^1|^2 dxdy \lesssim \ep^4 t^{-2+4\sigma};
			\end{aligned}
			\]
			- finally, from the decay bound \eqref{KS1_ext}, inequality \eqref{hardy_ext} and the energy bound \eqref{boot1_ex}
			\[
			\begin{aligned}
				\sum_{\substack{|K_1|+|K_2|\le N\\ |K_2|\le \floor{N/2}}} & \iint_{\Sext_t\setminus \Sext_{2t}}(2+r-t)^{1+2\mu} | Z^{K_1}h^1|^2 |\partial Z^{K_2}h^1|^2 dxdy \\
				&\lesssim \ep^2 t^{2\sigma}\iint_{\Sext_t\setminus \Sext_{2t}}(2+r-t)^{-1+2(\mu-\kappa)}r^{-2} |Z^{\leq N}h^1|^2  dxdy \\
				& \lesssim \ep^2 t^{-2+2\sigma}\iint_{\Sext_t}(2+r-t)^{1+2\kappa}|\partial Z^{\leq N}h^1|^2dxdy\lesssim  \ep^4 t^{-2+4\sigma}.
			\end{aligned}
			\]
Summing up,
\[
\iint_{\Sext_t} (2+r-t)^{1+2\mu}(\partial  Z^{\le N} h^1_{LT})^2 dx dy\lesssim C_0^2\ep^2 t^{2(\mu-\kappa+\sigma)}(1+ t\, l(t))
\]			
which concludes the proof of \eqref{est_hLT1}.
	\end{proof}

	\subsection{The null and cubic terms}
	The combination of the energy assumptions and the decay bounds obtained in proposition \ref{prop:pointwise_ext} yield easily the following weighted $L^2(\Sext_t)$ estimates of the differentiated null and cubic terms.
	
	\begin{proposition} \label{prop:null+cubic}
		Fix $i=0,1$. Under the a-priori energy assumptions \eqref{Boot1_ext}-\eqref{Boot2_ext} we have
		\begin{equation} \label{null_ext}
			\|(2+r-t)^{\frac{1}{2}+i+\kappa}\partial^i Z^{\le N}{\bf Q}_{\alpha\beta}(\partial h, \partial h)\|_{L^2(\Sext_t)}
			\lesssim C_0^2\ep^2 t^{-1+ 2\sigma}\sqrt{l(t)} + C_0^2\ep^2 t^{-2+2\sigma}
		\end{equation}
		and 
		\begin{equation} \label{cub_ext}
			\|(2+r-t)^{\frac{1}{2}+i+\kappa}\partial^i Z^{\le N}G_{\alpha\beta}(h)(\partial h, \partial h)\|_{L^2(\Sext_t)}  \lesssim C_0^3
			\ep^3 t^{-2+3\sigma}.
		\end{equation}
	\end{proposition}
	\begin{proof}
		We write $h = h^1 + h^0$ and inject this decomposition into ${\bf Q}_{\alpha\beta}$ and $G_{\alpha\beta}$. 
		We prove estimates \eqref{null_ext} and \eqref{cub_ext} for null and cubic interactions involving only $h^1$-factors. The remaining interactions, i.e. those involving at least one $h^0$-factor, can be easily treated thanks to \eqref{h0_estimate} so we leave the details to the reader.
		
		It is well-known that the admissible vector fields $Z$ preserve the null structure, in the sense that for any null form $Q$
		\[
		Z Q(\partial \phi, \partial \psi) = Q( \partial Z \phi, \partial \psi)+Q(\partial \phi, \partial Z \psi) + \tilde{Q}(\partial \phi, \partial \psi)
		\]
		where $\tilde{Q}$ is also a null form. 
		Together with the fundamental property
		\[
		|Q(\partial \phi, \partial \psi)|\lesssim |\op \phi| |\partial \psi| + |\partial \phi| |\op \psi|,
		\]
		it implies that for $i=0,1$
		\[
		|\partial^i Z^{\le N} {\bf Q}_{\alpha\beta}(\partial h^1, \partial h^1 )|\lesssim \sum_{\substack{|K_1|+|K_2|\le N \\ |I_1|+|I_2|=i}} |\op \partial^{I_1} Z^{K_1} h^1| |\partial \partial^{I_2} Z^{K_2}h^1| .
		\]
		We observe that at least one of the two indexes in the above summation has length smaller than $\floor{ N/2}$. Therefore, if $N$ is sufficiently large (e.g. $N\ge 6$) so that $\floor{N/2}\le N-3$ we deduce the following:
		
		- from \eqref{KS2_ext} and \eqref{boot1_ex} 
		\begin{multline*}
			\sum_{\substack{|K_1|+|K_2|\le N \\ |K_1|\le \floor{N/2}\\ |I_1|+|I_2|=i }} \|(2+r-t)^{\frac{1}{2}+i+\kappa} \op \partial^{I_1} Z^{K_1}h^1\ \partial \partial^{I_2} Z^{K_2}h^1\|_{L^2(\Sext_t)} \\
			\lesssim C_0\ep t^{-1+\sigma} \sqrt{l(t)} \sum_{j\le i} \|(2+r-t)^\frac{(i+j)}{2}\partial \partial^j Z^{\le N}h^1\|_{L^2(\Sext_t)}\lesssim C_0^2\ep^2 t^{-1+2\sigma}\sqrt{l(t)};
		\end{multline*}
		
		- from \eqref{KS1_ext} and \eqref{boot2_ex}
		\begin{multline*}
			\sum_{\substack{|K_1|+|K_2|\le N\\ |K_2|\le \floor{N/2}}} \|(2+r-t)^{\frac{1}{2}+i+\kappa} \op Z^{K_1}h^1\ \partial \partial^i Z^{K_2}h^1\|_{L^2(\Sext_t)} \\
			\lesssim C_0\ep t^{-1+\sigma} \| \op Z^{\le N}h^1\|_{L^2(\Sext_t)}\lesssim C_0^2\ep^2 t^{-1+2\sigma}\sqrt{l(t)};
		\end{multline*}
		
		- from \eqref{KS1_ext} and \eqref{boot4_ex}
		\begin{multline*}
			\sum_{\substack{|K_1|+|K_2|\le N-1 \\ |K_2|\le \floor{(N-1)/2}}} \|(2+r-t)^{\frac{3}{2}+\kappa} \op \partial Z^{K_1}h^1\ \partial Z^{K_2}h^1\|_{L^2(\Sext_t)} \\
			\lesssim C_0\ep t^{-1+\sigma} \|(2+r-t)^{\frac{1}{2}} \op \partial Z^{\le N-1}h^1\|_{L^2(\Sext_t)}\lesssim C_0^2\ep^2 t^{-1+2\sigma}\sqrt{l(t)}.
		\end{multline*}
		
		As concerns the cubic terms, we have that
		\[
		|\partial^i Z^{\le N}G_{\alpha\beta}(h^1)(\partial h^1, \partial h^1)|\lesssim \sum_{\substack{|K_1|+|K_2|+|K_3|\le N\\ |I_1|+|I_2|+|I_3| = i }} |\partial^{I_1}Z^{K_1}h^1| |\partial \partial^{I_2} Z^{K_2}h^1| |\partial \partial^{I_3} Z^{K_3}h^1| .
		\]
		From \eqref{KS1_ext}, inequality \eqref{hardy_ext} with $\beta = 2\kappa -1$ and \eqref{boot1_ex}, we deduce
		\begin{multline*}
			\sum_{\substack{|K_1|+|K_2|+|K_3|\le N\\ |K_2| + |K_3|\le \floor{N/2}\\ |I_2|+|I_3|=i } }   \|(2+r-t)^{\frac{1}{2}+i+\kappa} Z^{K_1}h^1\, \partial\partial^{I_2}Z^{K_2}h^1\, \partial \partial^{I_3}Z^{K_3}h^1\|_{L^2(\Sext_t)} \\
			\lesssim C_0^2\ep^2 t^{-2+2\sigma} \|(2+r-t)^{-\frac{3-i}{2}-\kappa} Z^{\le N}h^1\|_{L^2(\Sext_t)}\\ \lesssim C_0^2\ep^2 t^{-2+2\sigma} \|(2+r-t)^{\frac{1}{2}+\kappa}\partial Z^{\le N}h^1\|_{L^2(\Sext_t)} 
			\lesssim C_0^3\ep^3 t^{-2+3\sigma}.
		\end{multline*}
		The other interactions are easier to treat and their estimates are obtained similarly to what has been done above for the quadratic terms. We leave the details to the reader.		
	\end{proof}

	An immediate consequence of the pointwise bounds \eqref{KS1_ext}-\eqref{KS3_ext} is the following:
	
	\begin{proposition}
		Under the assumptions \eqref{Boot1_ext}-\eqref{Boot2_ext} we have that
		\begin{gather}
			|Z^{\le N-3}\mathbf{Q}_{\alpha\beta}(\partial h^1, \partial h^1)(t)|_{\frac12+2\kappa}\lesssim C_0^2\ep^2 t^{-1+2\sigma}\sqrt{l(t)}, \label{null_ext_point}\\
			|Z^{\le N-3}G_{\alpha\beta}(h^1)(\partial h^1, \partial h^1)(t)|_{1+3\kappa}\lesssim C_0^3\ep^3 t^{-2+3\sigma}.\label{cub_ext_point}
		\end{gather}
	\end{proposition}

	\subsection{The commutator terms}
	The goal of this section is to get suitable weighted $L^2(\Sext_t)$ estimates of the commutator terms $[Z^K, H^{\mu\nu}\partial_\mu\partial_\nu]h^1_{\alpha\beta}$ for $|K|\le N$. Such terms have a remarkable property when written in the null frame, which was first highlighted in \cite{LR10}. We present below a slightly different version of this, which involves expanding first in the null frame before evaluating the commutators.
	
	\begin{lemma}
		Let $K$ be any fixed multi-index and assume that $\pi^{\mu\nu}$ is a tensor satisfying
		\[
		| Z^{K'}\pi|\le C, \quad \forall\ |K'|\le \floor{|K|/2}.
		\]
		Then for any smooth function $\phi$
		\begin{equation}\label{comm_ext}
			\begin{aligned}
				|[Z^K, \pi^{\mu\nu}\partial_\mu\partial_\nu]\phi|  \lesssim \sum_{\substack{|K_1|+|K_2|\le |K| \\ |K_2|<|K|}} | Z^{K_1} \pi|_{\q L \q L}|\partial^2 Z^{K_2} \phi| + | Z^{K_1} \pi||\op \partial  Z^{K_2} \phi|  \\
				+  \sum_{|K_1|+|K_2|\le |K|} r^{-1}| Z^{K_1} \pi| |\partial  Z^{K_2} \phi| 
			\end{aligned}
		\end{equation}
		and
		\begin{equation}\label{comm_int}
			\begin{aligned}
				\Big|[Z^K, \pi^{\mu\nu}\partial_\mu\partial_\nu]\phi - \sum_{\substack{|K_1|+|K_2|\le |K|\\ |K_2|<|K|}} \big(Z^{K_1}\pi_{LL}\cdot \partial^2_t Z^{K_2}\phi + Z^{K_1}\pi_{4L}\cdot\partial_y\partial_t Z^{K_2}\phi + Z^{K_1}\pi_{44}\cdot\partial^2_y Z^{K_2}\phi\big)\Big|\\
				\lesssim \sum_{\substack{|K_1|+|K_2|\le |K|\\ |K_2|<|K|}} \frac{|t^2-r^2|}{t^2}|Z^{K_1}\pi| |\partial^2 Z^{K_2}\phi| +  |Z^{K_1}\pi| |\partial\pb_x Z^{K_2}\phi| +   \sum_{|K_1|+|K_2|\le |K|}  \frac{|Z^{K_1}\pi| |\partial Z^{K_2}\phi|}{1+t+r}.
			\end{aligned}
		\end{equation}
	\end{lemma}
	\begin{proof}
		Let $U, V$ denote any vector field in $\q U$.
		Inequality \eqref{comm_ext} follows from the following decomposition
		\begin{equation}\label{comm_null1}
			\begin{aligned}
				& [ Z^K, \pi^{\mu\nu}\partial_\mu\partial_\nu]\phi = [ Z^K, \pi^{UV}UV]\phi  =\hspace{-10pt} \sum_{\substack{|K_1|+|K_2|\le |K|\\ |K_2|<|K|}}\hspace{-10pt}( Z^{K_1}\pi^{UV}) UV  Z^{K_2}\phi +  Z^{K_1}\pi^{UV}\ [ Z^{K_2}, UV] \phi
			\end{aligned}
		\end{equation}
		and the fact that
		\begin{equation}\label{comm_null2}
			|[ Z^J, TU]\phi|\lesssim \sum_{|J'|<|J|} |\op \partial  Z^{J'}\phi| + \sum_{|J''|\le |J|} r^{-1}|\partial  Z^{J''} \phi|.
		\end{equation}
		Using instead \eqref{eq:curved_part_semihyp} we find that\footnote{ We recall that $\pb_0=\partial_t$}
		\[
		\begin{aligned}
			& [Z^K, \pi^{\mu\nu}\partial_\mu\partial_\nu]\phi  = [Z^K, \pi^{\bmmu\bmnu}\partial_\bmmu\partial_\bmnu]\phi + [Z^K, \pi^{4\nu}\partial_4\partial_\nu]\phi +  [Z^K, \pi^{\nu 4}\partial_\nu \partial_4]\phi\\
			& = [Z^K, \pi^{UV}c^{00}_{UV}\partial^2_t]\phi +[Z^K, \pi^{UV}( c_{UV}^{\bma \bmbeta}\pb_\bma\pb_\bmbeta+   c_{UV}^{\bmalph \bmb}\pb_\bmalph\pb_\bmb + d_{UV}^{\mu}\pb_\mu)]\phi \\
			& \quad + [Z^K, \pi^{4U}\partial_y U]\phi +  [Z^K, \pi^{U 4}U \partial_4]\phi
		\end{aligned}
		\]
		where, thanks to the fact that $|\partial^I\Gamma^J c^{00}_{LL}|\lesssim_{I, J} (t^2-r^2)/t^2$ and $\partial_\bmi =\pb_\bmi - \frac{x_\bmi}{t}\partial_t$, and up to homogeneous zero-order coefficients, we schematically have the following equalities
		\[
		[Z^K, \pi^{UV}c^{00}_{UV}\partial^2_t]\phi = \hspace{-5pt}\sum_{\substack{|K_1|+|K_2|\le |K|\\|K_2|<|K|}}\hspace{-5pt} Z^{K_1}\pi_{LL}\cdot\partial^2_t Z^{K_2}\phi + \frac{|t^2-r^2|}{t^2}Z^{K_1}\pi\cdot\partial^2 Z^{K_2}\phi + Z^{K_1}\pi\cdot \partial\pb_x Z^{K_2}\phi,
		\]
		\[
		[Z^K, \pi^{UV}( c_{UV}^{\bma \bmbeta}\pb_\bma\pb_\bmbeta+   c_{UV}^{\bmalph \bmb}\pb_\bmalph\pb_\bmb + d_{UV}^{\mu}\pb_\mu)]\phi = \hspace{-5pt}\sum_{\substack{|K_1|+|K_2|\le |K|\\ |K_2|<|K|}}\hspace{-5pt} Z^{K_1}\pi\cdot \partial\pb_x Z^{K_2}\phi + \frac{Z^{K_1}\pi\cdot\partial Z^{K_2} \phi}{1+t+r},
		\]
		\[
		[Z^K, \pi^{4U}\partial_y U]\phi = \sum_{\substack{|K_1|+|K_2|\le |K|\\|K_2|<|K|}} Z^{K_1}\pi_{4L}\cdot\partial_y\partial_t Z^{K_2}\phi + Z^{K_1}\pi_{44}\cdot \partial^2_yZ^{K_2}\phi + Z^{K_1}\pi\cdot\partial_y \pb_x Z^{K_2}\phi.
		\]
		
	\end{proof}

	\begin{proposition}\label{lem:comm_ext}
		Under the energy assumptions \eqref{Boot1_ext}-\eqref{Boot2_ext} we have for $i=0,1$
		\begin{equation}\label{ext_commutator}
			\begin{aligned}
				\left\|(2+r-t)^{\frac{1}{2}+i +\kappa} [\partial^i Z^{\le N}, H^{\mu\nu}\partial_\mu\partial_\nu]h^1_{\alpha\beta} \right\|_{L^2(\Sext_t)}  \lesssim  \ep t^{-1}E^{\text{e}, i+\kappa}(t, \partial^i Z^{\le N-i}h^1_{\alpha\beta})^{1/2} \\
				+  \ep^2t^{-(\kappa-\rho)+2\sigma} \big( t^{-1/2}\sqrt{l(t)}  +  t^{-1}\big).
			\end{aligned}
		\end{equation}
	\end{proposition}
	\begin{proof}
		We set $\phi = h^1_{\alpha\beta}$ in \eqref{comm_ext} and begin by observing that, for every $K$ with $|K|\le N$, the terms in the last line of the right hand side have already been estimated. In fact, the cubic terms satisfy \eqref{cub_ext} and the following bound was obtained in the proof of proposition \ref{prop:weak_ext}
		\[
		\sum_{\substack{i=0,1\\ |K_1|+|K_2|\le N\\ |I_1|+|I_2| = i}} \left\|(2+r-t)^{\frac{1}{2}+i+\kappa}\, r^{-1} \partial^{I_1} Z^{K_1}h \cdot \partial \partial^{I_2} Z^{K_2} h\right\|_{L^2(\Sext_t)} \lesssim \ep^2 t^{-\frac{3}{2}+2\sigma}.
		\]
		Using \eqref{h0_estimate} and \eqref{boot3_ex}, \eqref{boot4_ex}, it is straightforward to prove that for $i=0,1$
		\[
		\begin{aligned}
			\sum_{\substack{|K_1|+|K_2|\le N\\ |I_1|+|I_2| = i}} \big\|(2+r-t)^{\frac{1}{2}+i+\kappa} \partial^{I_1} Z^{K_1}h^0\cdot  \partial^2 \partial^{I_2}Z^{K_2}h^1 \big\|_{L^2(\Sext_t)}\lesssim \ep t^{-1} E^{\text{e}, i+\kappa}(\partial^i Z^{\le N}h^1)^{1/2},
		\end{aligned}
		\]
		hence we only focus on estimating the terms of the first line in the right hand side of \eqref{comm_ext} with $h$ replaced by $h^1$. We choose exponents $(p_1,p_2)$ such that
		\[
		(p_1,p_2) =\begin{cases}
			(2, \infty), \quad & \text{if } |K_1|=N \\
			(\infty, 2), \quad & \text{if } |K_2| = N-1 \\
			(4,4), \quad & \text{otherwise}.
		\end{cases}
		\]
		From the Sobolev's injections $H^2(\S^2\times \S^1)\subset L^\infty(\S^2\times \S^1)$ and $H^1(\S^2\times \S^1)\subset L^4(\S^2\times \S^1)$
		\begin{multline*}
			\sum_{\substack{ |K_1|+|K_2|\le N\\ |I_2| =i}} \big\|(2+r-t)^{\frac{1}{2}+i+\kappa} Z^{K_1}h^1\cdot \op\partial\partial^{I_2} Z^{K_2}h^1 \big\|_{L^2(\Sext_t)}^2 \\[-5pt]
			\le \sum_{\substack{ |K_1|+|K_2|\le N\\ |I_2| =i}}  \int_{r\ge t-1} (2+r-t)^{1+2i+2\kappa}\| Z^{K_1}h^1\|^2_{L^{p_1}(\S^2\times \S^1)} \|\op\partial \partial^{I_2} Z^{K_2}h^1\|^2_{L^{p_2}(\S^2\times \S^1)} r^2 dr\\
			\lesssim \int_{r\ge t-1}(2+r-t)^{1+2i+2\kappa} \| Z^{\le N}h^1\|^2_{L^2(\S^2\times \S^1)} \|\op \partial  Z^{\le N-1}h^1\|^2_{L^2(\S^2\times \S^1)} r^2 dr 
		\end{multline*}
		so using the inequality \eqref{sob+hardy2_ext} with $\beta=\kappa$ and the energy assumptions \eqref{boot1_ex}, \eqref{boot2_ex} we get
		\[
		\begin{aligned}
			&\lesssim \left\|(2+r-t)^{\frac{1}{2}+\kappa} \partial Z^{\le N}h^1\right\|^2_{L^2(\Sext_t)} \int_{r\ge t-1} (2+r-t)^{1+2i}r^{-2}\|\op \partial Z^{\le N-1}h^1\|^2_{L^2(\S^2\times\S^1)} r^2 dr \\
			&\lesssim t^{-1-2\kappa}\left\|(2+r-t)^{\frac{1}{2}+\kappa} \partial Z^{\le N}h^1\right\|^2_{L^2(\Sext_t)}\left\|(2+r-t)^{1+\kappa} \op\partial Z^{\le N-1}h^1\right\|^2_{L^2(\Sext_t)}\lesssim \ep^4 t^{-1-2\kappa+4\sigma} l(t).
		\end{aligned}
		\]
		The same Sobolev's embeddings, coupled with the decay bound \eqref{est_hLT1} and the energy assumption \eqref{boot1_ex}, also yield for $i=0,1$
		\[
		\begin{aligned}
			& \sum_{\substack{ |K_1|+|K_2|\le N\\ |I_2| =i}} \big\|(2+r-t)^{\frac{1}{2}+i+\kappa} Z^{K_1}h^1_{LL}\cdot \partial^2 \partial^{I_2}  Z^{K_2}h^1 \big\|_{L^2(\Sext_t)}^2 \\
			& \lesssim  \int_{r\ge t-1}(2+r-t)^{1 + 2i+2\kappa} \| Z^{\le N}h^1_{LL}\|^2_{L^2(\S^2\times \S^1)} \|\partial \partial Z^{\le N-1}h^1\|^2_{L^2(\S^2\times \S^1)} r^2 dr \\
			& \lesssim \ep^2t^{-2(\kappa-\rho) + 2\sigma} (t^{-1}l(t)  + t^{-2}) \left\|(2+r-t)^{\frac{1}{2}+i+\kappa} \partial \partial  Z^{\le N-1}h^1 \right\|_{L^2(\Sext_t)}^2 \\
			& \lesssim \ep^4 t^{-2(\kappa-\rho) + 4\sigma} (t^{-1}l(t)  + t^{-2}).
		\end{aligned}
		\]
		Finally, when $i=1$ the remaining terms to discuss are of the form $\partial Z^{K_1}h^1_{LL}\cdot \partial^2 Z^{K_2}h^1$ and $\partial Z^{K_1}h^1\cdot \op\partial Z^{K_2}h^1$ but those behave like null terms (the former thanks to \eqref{ineq_H1lt_higher} for $h^1$) and hence satisfy \eqref{null_ext}. The details are left to the reader.
		
	\end{proof}

	We also have the following pointwise estimate of the commutator terms involving a smaller number of vector fields. It will be useful in the proof of Lemma \ref{prp:good_coeffII}.
	
	\begin{lemma}
		Under the energy assumptions \eqref{Boot1_ext}-\eqref{Boot2_ext}, there exists $\delta'>0$ such that
		\begin{equation} \label{comm_ext_point}
			\big|[Z^{\le N-4}, H^{1, \mu\nu}\partial_\mu\partial_\nu]\hab\big|_\kappa \lesssim C_0^2\ep^2 t^{-2+2\sigma}\sqrt{l(t)} + \ep^2 t^{-2-\delta'}(2+r-t)^{-\frac12}.
		\end{equation}
	\end{lemma}
	\begin{proof}
		We use \eqref{comm_ext} with $\pi=H^1$ and $\phi=\hab$.
		Pointwise bounds \eqref{KS1_ext} and \eqref{est_hLT2} yield
		\[
		\sum_{|K_1|+|K_2|\le N-4}	\big| Z^{K_1}H^1_{LL} \cdot\partial^2 Z^{K_2}\hab \big|\lesssim C_0^2\ep^2 t^{-2-\delta+\sigma} (2+r-t)^{-\frac32-\kappa}
		\]
		for some $\delta>\sigma$, bounds \eqref{KS2_ext} and \eqref{KS3_ext} give
		\[
		\sum_{|K_1|+|K_2|\le N-4}\big| Z^{K_1}H^1\cdot \partial \op Z^{K_2}\hab\big|\lesssim C_0^2\ep^2 t^{-2+2\sigma}\sqrt{l(t)}(2+r-t)^{-1-2\kappa}
		\]
		and finally \eqref{KS1_ext} and \eqref{KS3_ext} imply
		\[
		\sum_{|K_1|+|K_2|\le N-4} r^{-1}\big| Z^{K_1}H^1\cdot \partial Z^{K_2}\hab\big|\lesssim C_0^2\ep^2 t^{-3+2\sigma}(2+r-t)^{-1-2\kappa}.
		\]			
		The result of the statement follows by setting $\delta'=\delta-\sigma$.
	\end{proof}
	
	\subsection{The $h^1_{TU}$ coefficients}
	\label{sub:hTU_coeff}
	
	In this subsection we show that, for any $T\in \q T$ and $U \in \q U$, the coefficients $h^1_{TU}$ satisfy better energy bounds than \eqref{Boot1_ext}, more precisely that for any fixed $0<\rho<\kappa$ there exists some positive constant $C$ such that
	\begin{equation} \label{energy_hTU}
		E^{\text{e}, \kappa-\rho}(t, Z^{\le N}h^1_{TU})^{1/2}\lesssim C_0\ep t^{C\ep}, \qquad t\in [2, T_0).
	\end{equation}
	This estimate essentially follows from the fact that no weak null terms appear among the source terms in the equation satisfied by $h^1_{TU}$. This can be simply seen by applying $T^\alpha U^\beta$ to \eqref{h_equations} and then commuting with $Z^K$, which shows that $Z^Kh^1_{TU}$ is solution to
	
	\begin{equation}\label{h1TU_higher}
		\tilde{\Box}_g  Z^K h^1_{TU} = F^K_{TU} + F^{0, K}_{TU}, \qquad F^{0,K}_{TU} = F^{0,K}_{\alpha\beta} T^\alpha U^\beta
	\end{equation}
	with source term $F^K_{TU}$ given by
	\begin{equation}\label{FKTU}
		\begin{aligned}
			F^K_{TU} & = - [ Z^K, H^{\mu\nu}\partial_\mu\partial_\nu]h^1_{TU} +  Z^KF_{TU}(h)(\partial h, \partial h)\\
			& + \sum_{|K'|\le |K|}C^{\bmi\alpha\beta}_{TU,K'}\gd_\bmi  Z^{K'} h^1_{\alpha\beta} + D^{\alpha\beta}_{TU,K'}  Z^{K'}h^1_{\alpha\beta} \\
			& +\sum_{|K_1|+|K_2|\le |K|}  E^{\bmi\alpha\beta}_{TU\mu\nu, K_1K_2} Z^{K_1} H^{\mu\nu}\cdot \gd_\bmi  Z^{K_2}h^1_{\alpha\beta} + F^{\alpha\beta}_{TU\mu\nu, K_1K_2}  Z^{K_1}H^{\mu\nu} \cdot Z^{K_2}h^1_{\alpha\beta}
		\end{aligned}
	\end{equation}
	and smooth coefficients $C^{i\alpha\beta}_{TU,K'},  E^{i\alpha\beta}_{TU\mu\nu, K_1K_2} = O(r^{-1})$, $D^{\alpha\beta}_{TU,K'},  F^{\alpha\beta}_{TU\mu\nu, K_1K_2} =O(r^{-2})$.  
	Besides the additional terms arising from the commutation of vector fields $T$ and $U$ with the reduced wave operator, the main difference between the source terms $F^K_{\alpha\beta}$ and $F^K_{TU}$ lies in the fact that the latter is a linear combination of quadratic null terms and cubic terms only, as $P_{TU} = P_{\alpha\beta}T^\alpha U^\beta$ and
	\begin{equation}\label{PTU}
		|P_{TU}(\phi, \psi)|\lesssim |\op \phi| |\partial \psi| + |\partial\psi| |\op \psi|.
	\end{equation}
	We compare \eqref{h1TU_higher} with equation \eqref{linearized_wave}. Thanks to the smallness of $H$ provided by \eqref{h0_estimate} and \eqref{est_hLT2}, we can apply the result of proposition \ref{prop:exterior_energy} with $\Wf = Z^Kh^1_{TU}$, $\Ff = F^K_{TU} + F^{0,K}_{TU}$, $w(q) = (2+r-t)^{1+2(\kappa-\rho)}$ and $t_1=2, t_2 = t$. 
	We obtain that
	\begin{equation}\label{energy_ineq_hTU}
		\begin{aligned}
			& E^{\text{e}, \kappa-\rho}(t, Z^K h^1_{TU}) \\
			&+ \hspace{-2pt} \int_{\tilde{\hcal}_{2,t}}  \hspace{-10pt} (2+r-t)^{1+2(\kappa-\rho)}\Big[\Big(\frac{1}{2( 1+r^2)} +\chi\left(\frac{r}{t}\right)\chi(r)\frac{M}{2r} \Big) |\partial_t Z^K h^1_{TU}|^2 +|\underline{\nabla}Z^K h^1_{TU}|^2 \Big]dxdy \\
			& + \iint_{\dext_{[2,t]}} (2+r-\tau)^{2(\kappa-\rho)}\big(|LZ^K h^1_{TU}|^2  + |\slashed \nabla Z^K h^1_{TU}|^2\big) d\tau dxdy\lesssim E^{\text{e}, \kappa-\rho}(2, Z^K h^1_{TU}) \\
			& + \iint_{\dext_{[2,t]}}(2+r-\tau)^{1+2(\kappa-\rho)}|(F^K_{TU} + F^{0,K}_{TU} + {\partial^\alpha H_\alpha}^\sigma \, \partial_\sigma Z^K h^1_{TU})\partial_t Z^K h^1_{TU}|d\tau dxdy\\
			& +  \iint_{\dext_{[2,t]}}(2+r-\tau)^{1+2(\kappa-\rho)}|{\partial_t H_\alpha}^\sigma \, \partial_\sigma Z^K h^1_{TU} \, \partial^\alpha Z^K h^1_{TU} |\, d\tau dxdy \\
			& + \iint_{\dext_{[2,t]}} (2+r-\tau)^{2(\kappa-\rho)} | H^{\rho\sigma}\partial_\rho Z^K h^1_{TU} \partial_\sigma Z^K h^1_{TU}|\ d\tau dxdy\\\
			&  + \iint_{\dext_{[2,t]}}  (2+r-\tau)^{2(\kappa-\rho)} | (-H^{0\sigma} + \omega_\bmj H^{\bmj\sigma}) \partial_\sigma Z^K h^1_{TU}\, \partial_t Z^K h^1_{TU} | \ d\tau dxdy.
		\end{aligned}
	\end{equation} 
	
	We start by estimating the contributions coming from the source term $F^{0, K}_{TU}$.
	\begin{lemma}\label{lem:box_h0}
		There exists $\delta>0$ such that, under the energy assumptions \eqref{Boot1_ext}-\eqref{Boot2_ext} and for $i=0,1$, we have 
		\begin{equation}\label{source0_higher}
			\| (2+r-t)^{\frac{1}{2}+i+\kappa} \partial^i Z^{\le N}\tilde{\Box}_g h^0_{\alpha\beta}\|_{L^2(\Sext_t)} \lesssim \ep t^{-3-i} + \ep^2 t^{-2+\delta}.
		\end{equation}
		Consequently
		\[
		\|(2+r-t)^{\frac{1}{2}+\kappa-\rho} Z^{\le N}F^{0}_{TU}\|_{L^2(\Sext_t)}\lesssim  \ep t^{-3} + \ep^2 t^{-2+\delta}.
		\]
	\end{lemma}
	\begin{proof}
		We recall that definition \eqref{h0_def} of $h^0$ and that $\tilde{\Box}_g h^0_{\alpha\beta} = F^{00} + F^{01} + F^{02}$ with 
		\begin{equation}\label{dec_F0}
			F^{00} = \Box h^0_{\alpha\beta}, \qquad F^{01} = -h^{0, \mu\nu}\partial_\mu\partial_\nu h^0_{\alpha\beta}, \qquad F^{02}=(-h^{1,\mu\nu} + \mathcal{O}^{\mu\nu}(h^2))\partial_\mu\partial_\nu h^0_{\alpha\beta}.
		\end{equation}
		The fundamental remark is that $\Box (M/r) = 0$ away from $r =0$, which implies that $F^{00}$ is supported for $t/2\le r\le 3t/4$. Consequently, $\text{supp} F^{00} \cap \dext_{T_0}$ is bounded and from \eqref{h0_estimate}
		\[
		\begin{gathered}
			|\partial^i Z^{\le N}F^{00}|\lesssim \frac{\epsilon}{(1+r)^{3+i}},\  \qquad \
			\| (2+r-t)^{\frac{1}{2}+i+\kappa} \partial^i Z^{\le N}F^{00}\|_{L^2(\Sext_t)}\lesssim \ep t^{-3-i}.
		\end{gathered}
		\]
		From \eqref{h0_estimate} we also see that
		\[
		\begin{gathered}
			|\partial^i Z^{\le N}F^{01}| \lesssim \frac{\ep^2}{(1+r)^{4+i}}\\
			\| (2+r-t)^{\frac{1}{2}+i+\kappa} \partial^i Z^{\le N}F^{01}\|_{L^2(\Sext_t)}\lesssim \ep^2 \| r^{-\frac{7}{2}+\kappa}\|_{L^2(\Sext_t)}\lesssim \ep^2 t^{-2+\kappa}.
		\end{gathered}
		\]
		From \eqref{h0_estimate}, the Hardy inequality \eqref{hardy_ext} with $\beta =2\kappa-1$ and estimate \eqref{boot1_ex}, we have that
		\[
		\begin{aligned}
			&\sum_{\substack{|K_1|+|K_2| \le N\\ |I_2|=i}} \|(2+r-t)^{\frac{1}{2}+i+\kappa} Z^{K_1}h^{1, \mu\nu} \partial_\mu\partial_\nu \partial^{I_2}
			Z^{K_2} h^0_{\alpha\beta}\|_{L^2(\Sext_t)}\\ 
			& \lesssim \ep \|(2+r-t)^{\frac{1}{2}+ i+\kappa}r^{-3-i} Z^{K_1}h^{1, \mu\nu}\|_{L^2(\Sext_t)}  \lesssim \ep t^{-2}\|(2+r-t)^{\frac{1}{2}+\kappa}\partial Z^{K_1}h^{1, \mu\nu}\|_{L^2(\Sext_t)}\lesssim \ep^2 t^{-2+\sigma}.
		\end{aligned}
		\]
		The cubic terms $\mathcal{O}^{\mu\nu}(h^2)\partial_\mu\partial_\nu h^0_{\alpha\beta}$ verify similar estimates, the details are left to the reader.
	\end{proof}
	
	In order to estimate the contributions due to the curved background, we first highlight the following relations.
	\begin{lemma} For any sufficiently smooth function $\phi$ we have
		\begin{gather}
			{ \partial^\alpha H_\alpha}^\sigma\, \partial_\sigma\phi = -\frac{1}{2}(\partial_{\ub} H_{LL} - \partial_u H_{\Lb L} + \partial_A H_{AL})\Lb\phi + ({\partial^\alpha H_\alpha}^T) T\phi \label{eq_null1}\\
			{\partial_t H_\alpha}^\sigma \, \partial_\sigma \phi \, \partial^\alpha \phi = \frac{1}{4}\partial_t H_{LL} (\Lb \phi)^2 + \partial_t H^{T\alpha}\, (\partial_\alpha \phi)(T \phi) \label{eq_null2}\\
			H^{\rho\sigma}\partial_\rho\phi \, \partial_\sigma \phi = \frac{1}{4}H_{LL} (\Lb \phi)^2 + H^{T\alpha}(T\phi)(\partial_\alpha\phi) \label{eq_null3} \\
			(-H^{0\sigma} + \omega_\bmj H^{\bmj\sigma})\partial_\sigma \phi = -\frac{1}{2}H_{LL}\, \Lb\phi + {H_L}^T (T\phi) \label{eq_null4}
		\end{gather}
	\end{lemma}
	
	\begin{proof}
		The proof of the above equalities follows after expressing all vector fields relative to the null frame $\q U=\{L, \Lb, S^1, S^2, \partial_y\}$ and observing that
		\[
		-H^{0\sigma} + \omega_\bmj H^{\bmj\sigma} = \bar{g}_{\mu\nu}L^\mu H^{\nu\sigma} = {H_L}^\sigma.
		\]
	\end{proof}
	
	\begin{lemma}
		Under the a-priori assumptions \eqref{Boot1_ext}-\eqref{Boot2_ext} we have that for $i=0,1$ and any multi-index $K$ with $|K|\le N$
		\begin{equation}\label{dH_contribution_ext}
			\begin{aligned}
				& \iint_{\dext_{[2,t]}} \hspace{-10pt} (2+r-\tau)^{1+2(i+\kappa)}|\partial^\mu {H_\mu}^\sigma \cdot \partial_\sigma \partial^i Z^Kh^1_{\alpha\beta}\cdot\partial_t \partial^i Z^Kh^1_{\alpha\beta}|\, d\tau dxdy \\
				& \iint_{\dext_{[2,t]}} \hspace{-10pt} (2+r-\tau)^{1+2(i+\kappa)}|\partial_t {H_\mu}^\sigma \cdot \partial_\sigma \partial^i Z^Kh^1_{\alpha\beta} \cdot \partial^\mu \partial^i Z^Kh^1_{\alpha\beta} |\, d\tau dxdy \lesssim C_0^3\ep^3
			\end{aligned}
		\end{equation}
		and
		\begin{equation} \label{H_contribution_ext}
			\begin{aligned}
				& \iint_{\dext_{[2,t]}}\hspace{-10pt} (2+r-\tau)^{2(i+\kappa)} | H^{\rho\sigma}\partial_\rho \cdot \partial^i Z^Kh^1_{\alpha\beta}\cdot \partial_\sigma \partial^i Z^Kh^1_{\alpha\beta}|  \, d\tau dxdy\\
				& + \iint_{\dext_{[2,t]}}\hspace{-10pt} (2+r-\tau)^{2(i+\kappa)} |(-H^{0\sigma} + \omega_\bmj H^{\bmj\sigma})\cdot \partial_\sigma \partial^i Z^Kh^1_{\alpha\beta}\cdot \partial_t \partial^i Z^Kh^1 _{\alpha\beta}
				|  d\tau dxdy\\
				& \lesssim \int_2^ t \frac{\epsilon E^{\text{e},i+\kappa}(\tau, \partial^i Z^Kh^1_{\alpha\beta})}{\tau}\,  d\tau + C_0^3\ep^3.
			\end{aligned}
		\end{equation}
		For any $0 \leq \rho < \kappa$, the same inequalities hold with $(\kappa,
		h^1_{\alpha\beta})$ replaced by $(\kappa-\rho, h^1_{TU})$.
	\end{lemma}
	\begin{proof}
		We start by remarking that from inequality \eqref{ineq_Hlt} and bounds \eqref{h0_estimate} to \eqref{KS3_ext}
		\[
		|\partial H_{LL}| + |\op H|\lesssim |\op h^1| + |\partial h^0| + |h||\partial h|\lesssim C_0\ep \Big(\frac{t^\sigma\sqrt{l(t)}}{r(2+r-t)^\kappa} + \frac{t^{2\sigma}}{r^{2}}\Big),
		\]
		which together with the energy bounds \eqref{boot1_ex} and \eqref{boot3_ex} implies
		\[
		\|  (2+r-t)^{\frac{1}{2}+i+\kappa} (|\partial H_{LL}| + |\op H|) \partial \partial^iZ^{\le N}h^1\|_{L^2(\Sext_t)}\lesssim C_0^2\ep^2 t^{-1+2\sigma}\sqrt{l(t)} + C_0^2\ep^2 t^{-2+3\sigma}.
		\]
		From the pointwise estimates \eqref{h0_estimate}, \eqref{KS1_ext} and energy bounds \eqref{boot2_ex}, \eqref{boot4_ex} we also have
		\[
		\|(2+r-t)^{\frac{1}{2}+i+\kappa} \partial H\cdot \op \partial^i Z^{\le N}h^1\|_{L^2(\Sext_t)}\lesssim C_0^2\ep^2  t^{-1+2\sigma}\sqrt{l(t)} .
		\]
		The Cauchy-Schwartz inequality, relation \eqref{eq_null1} and the above estimates yield that for $t\in [2, T_0)$
		\[
		\begin{aligned}
			&\iint_{\dext_{[2,t]}} \hspace{-10pt} (2+r-\tau)^{1+2(i+\kappa)}\partial^\mu {H_\mu}^\sigma \cdot \partial_\sigma \partial^i Z^Kh^1_{\alpha\beta}\cdot \partial_t \partial^i Z^Kh^1_{\alpha\beta}|\, d\tau dxdy\\
			&\lesssim \int_2^t \|(2+r-\tau)^{\frac{1}{2}+i+\kappa} \partial^\mu {H_\mu}^\sigma\cdot \partial_\sigma \partial^i Z^Kh^1_{\alpha\beta}\|_{L^2(\Sext_t)} E^{\text{e}, i+\kappa}(\tau, \partial^i Z^Kh^1_{\alpha\beta})^{1/2}d\tau \lesssim C_0^3 \ep^3.
		\end{aligned}
		\]
		Similarly, using relation \eqref{eq_null2} we get
		\[
		\begin{aligned}
			&\iint_{\dext_{[2,t]}} \hspace{-10pt} (2+r-\tau)^{1+2(i+\kappa)}|\partial_t {H_\mu}^\sigma \cdot \partial_\sigma \partial^i Z^Kh^1_{\alpha\beta} \cdot \partial^\mu \partial^i Z^Kh^1_{\alpha\beta} |\, d\tau dxdy  \\
			& \lesssim \iint_{\dext_{[2, t]}} \| (2+r-\tau)^{\frac{1}{2}+i +\kappa} \partial H_{LL}\cdot \partial \partial^i Z^K h^1_{\alpha\beta}\|_{L^2(\Sext_t)}E^{\text{e}, i+\kappa}(\tau, \partial^i Z^Kh^1_{\alpha\beta})^{1/2} d\tau dxdy\\
			& +\iint_{\dext_{[2, t]}} \| (2+r-\tau)^{\frac{1}{2}+i+\kappa} \partial H\cdot \op \partial^i Z^K h^1_{\alpha\beta}\|_{L^2(\Sext_t)}E^{\text{e}, i+\kappa}(\tau, \partial^i Z^Kh^1_{\alpha\beta})^{1/2} d\tau dxdy \lesssim C_0^3 \ep^3.
		\end{aligned}
		\]
		Finally, from formulas \eqref{eq_null3} and \eqref{eq_null4} and pointwise bounds \eqref{h0_estimate}, \eqref{KS3_ext} and \eqref{est_hLT1}
		\[
		\begin{aligned}
			& \iint_{\dext_{[2,t]}}\hspace{-10pt} (2+r-\tau)^{2(i+\kappa)} | H^{\rho\sigma}\cdot \partial_\rho  \partial^i Z^Kh^1_{\alpha\beta} \cdot\partial_\sigma \partial^i Z^Kh^1_{\alpha\beta}|  \, d\tau dxdy\\
			& + \iint_{\dext_{[2,t]}}\hspace{-10pt} (2+r-\tau)^{2(i+\kappa)} |(-H^{0\sigma} + \omega_\bmj H^{\bmj\sigma})\cdot \partial_\sigma \partial^i Z^Kh^1_{\alpha\beta}\cdot \partial_t \partial^i Z^Kh^1 _{\alpha\beta}
			|  d\tau dxdy \\
			& \lesssim \iint_{\dext_{[2,t]}}\hspace{-10pt} (2+r-t)^{2(i+\kappa)}\big[ |H_{LL}| |\partial \partial^i Z^Kh^1_{\alpha\beta}|^2 + |H| |\op \partial^i Z^Kh^1_{\alpha\beta}| |\partial \partial^i Z^K h^1_{\alpha\beta}|\big]\, d\tau dxdy \\
			& \lesssim \int_2^ t \frac{\epsilon E^{\text{e}, i+\kappa}(\tau, \partial^i Z^Kh^1_{\alpha\beta})}{\tau}\,  d\tau + C_0^3\ep^3.
		\end{aligned}
		\]
	\end{proof}

	\begin{proposition}
		Let $0<\rho<\kappa$ be fixed. Under the energy assumptions \eqref{Boot1_ext}-\eqref{Boot2_ext} there exists a constant $C>0$ such that \eqref{energy_hTU} holds for all $t\in [2, T_0)$.
	\end{proposition}
	\begin{proof}
		The result follows from inequality \eqref{energy_ineq_hTU} and the estimates we have obtained so far. The quadratic semilinear terms in $P_{TU}$ are now null, as observed in \eqref{PTU},
		therefore from \eqref{null_ext} and \eqref{cub_ext} and the smallness of $\epsilon$ it follows that
		\[
		\|(2+r-t)^{\frac{1}{2}+\kappa-\rho} Z^{\le N}F_{TU}\|_{L^2(\Sext_t)}\lesssim C_0^2 \ep^2 t^{-1+2\sigma}\sqrt{l(t)} + C_0^2\ep^2 t^{-2+3\sigma}.
		\]
		The only terms that still need to be addressed are the contributions to \eqref{FKTU} arising from the commutation of the null frame with the reduced wave operator. Using the energy bounds \eqref{boot2_ex} we see that
		\[
		\begin{aligned}
			& \|(2+r-t)^{\frac{1}{2}+\kappa-\rho}C^{\bmi\alpha\beta}_{TU,K'}\cdot \gd_\bmi  Z^{\le N} h^1_{\alpha\beta}\|_{L^2(\Sext_t)} \lesssim \|(2+r-t)^{\frac{1}{2}+\kappa-\rho}r^{-1}\overline{\nabla} Z^{\le N} h^1_{\alpha\beta}\|_{L^2(\Sext_t)}\\
			& \lesssim t^{-\frac{1}{2}-\rho}\|(2+r-t)^{\kappa}\overline{\nabla}  Z^{\le N} h^1_{\alpha\beta}\|_{L^2(\Sext_t)}\lesssim C_0\ep t^{-\frac{1}{2}-\rho+\sigma}\sqrt{l( t)}
		\end{aligned}
		\]
		while from \eqref{boot1_ex} and the weighted Hardy inequality \eqref{hardy_ext} with $\beta =2\kappa-1$ we get
		\[
		\begin{aligned}
			\|(2+r-t)^{\frac{1}{2}+\kappa-\rho}D^{\alpha\beta}_{TU,K'}\cdot  Z^{\le N} h^1_{\alpha\beta}\|_{L^2(\Sext_t)} \lesssim \|(2+r-t)^{\frac{1}{2}+\kappa-\rho}r^{-2} Z^{\le N} h^1_{\alpha\beta}\|_{L^2(\Sext_ t)} \\
			\lesssim t^{-1-\rho}\| (2+r-t)^{\kappa-\frac{1}{2}} Z^{\le N}h^1_{\alpha\beta}\|_{L^2(\Sext_t)}\lesssim  t^{-1-\rho}\| (2+r-t)^{\frac{1}{2}+\kappa} \partial Z^{\le N}h^1_{\alpha\beta}\|_{L^2(\Sext_ t)}\\  \lesssim C_0 \ep t^{-1-\rho+\sigma}.
		\end{aligned}
		\]
		We then recall the decomposition \eqref{dec_H} of the tensor $H$, with $H^{0,\mu\nu}$ satisfying \eqref{h0_estimate} and $H^{1,\mu\nu}$ verifying the bounds \eqref{KS1_ext}-\eqref{KS3_ext}.
		We similarly get
		\[
		\begin{aligned}
			\sum_{\substack{|K_1|+|K_2|\le N\\ |K_1|\le \floor{N/2}}} \|(2+r-t)^{\frac{1}{2}+\kappa-\rho} E^{\bmi\alpha\beta}_{TU\mu\nu, K_1K_2}\cdot Z^{K_1} H^{1,\mu\nu} \cdot \gd_\bmi  Z^{K_2}h^1_{\alpha\beta} \|_{L^2(\Sext_t)} \\
			+ \sum_{|K_1|+|K_2|\le |K|}  \|(2+r-t)^{\frac{1}{2}+\kappa-\rho} E^{\bmi\alpha\beta}_{TU\mu\nu, K_1K_2}\cdot Z^{K_1} H^{0,\mu\nu}\cdot \gd_\bmi  Z^{K_2}h^1_{\alpha\beta} \|_{L^2(\Sext_t)}
			\\
			\lesssim C_0\ep t^\sigma \|(2+r-t)^{\frac{1}{2}+\kappa-\rho}r^{-2}\overline{\nabla} Z^{\le N}h^1_{\alpha\beta} \|_{L^2(\Sext_t)} \lesssim C_0^2 \ep^2  t^{-\frac{3}{2}-\rho+2\sigma} \sqrt{l( t)}
		\end{aligned}
		\]
		and
		\[
		\begin{aligned}
			& \sum_{\substack{|K_1|+|K_2|\le N\\ |K_2|\le \floor{N/2}}} \|(2+r- t)^{\frac{1}{2}+\kappa-\rho} E^{\bmi\alpha\beta}_{TU\mu\nu, K_1K_2}\cdot Z^{K_1} H^{1,\mu\nu}\cdot\gd_\bmi  Z^{K_2}h^1_{\alpha\beta} \|_{L^2(\Sext_ t)} \\
			& \lesssim C_0\ep   t^\sigma\sqrt{l( t)} \|(2+r- t)^{\frac{1}{2}-\rho}r^{-2} Z^{\le N} H^{1,\mu\nu}\|_{L^2(\Sext_ t)} \\
			& \lesssim C_0\ep  t^{-1-\rho-\kappa+\sigma}\sqrt{l( t)} \|(2+r- t)^{\frac{1}{2}+\kappa} \partial  Z^{\le N} H^{1,\mu\nu}\|_{L^2(\Sext_ t)} \lesssim C_0^2 \ep^2  t^{-1-\rho-\kappa+2\sigma}\sqrt{l( t)}.
		\end{aligned}
		\]
		Finally,
		\[
		\begin{aligned}
			& \sum_{|K_1| + |K_2|\le N}  \|(2+r- t)^{\frac{1}{2}+\kappa-\rho} F^{\alpha\beta}_{TU\mu\nu, K_1K_2}\cdot Z^{K_1} H^{\mu\nu} \cdot Z^{K_2}h^1_{\alpha\beta} \|_{L^2(\Sext_ t)} \\
			&\lesssim  C_0\ep   t^\sigma \|(2+r-t)^{\frac{1}{2}+\kappa-\rho}r^{-3}  Z^{\le N}h^1 \|_{L^2(\Sext_t)} \lesssim C_0^2 \ep^2  t^{-2-\rho+2\sigma}.
		\end{aligned}
		\]
		By substituting the above estimates together with \eqref{ext_commutator}, \eqref{source0_higher}, \eqref{dH_contribution_ext} and \eqref{H_contribution_ext} into \eqref{energy_ineq_hTU} and choosing $\ep_0\ll 1$ sufficiently small so that $C_0\ep<1$ we finally find the existence of a universal constant $C$ such that
		\[
		\begin{aligned}
			E^{\text{e}, \kappa-\rho}(t, Z^Kh^1_{TU}) &\le C E^{\text{e}, \kappa-\rho}(2, Z^Kh^1_{TU}) + C C_0^2\ep^2 + \int_2^t \frac{C\epsilon E^{\text{e}, \kappa-\rho}(\tau, Z^Kh^1_{TU})}{\tau}\,  d\tau.
		\end{aligned}
		\]
		Observe that $E^{\text{e}, \kappa-\rho}(2, Z^Kh^1_{TU})\lesssim \enext(t, h^1)$. Gr\"onwall's inequality and the energy assumption \eqref{Boot1_ext} allow us to obtain
		\[
		E^{\text{e}, \kappa-\rho}(t, Z^Kh^1_{TU}) \le C(\enext(2, Z^Kh^1) + C_0^2\ep^2)t^{C\epsilon}\le 2C C_0^2\ep^2 t^{C\epsilon}
		\]
		and hence conclude the proof.
	\end{proof}

	An immediate consequence of \eqref{energy_hTU} are the following weighted $L^2$ bounds
	\begin{align}
		\label{hTU_en1}
		&\big\| (2+r-t)^{\frac{1}{2}+\kappa-\rho}\ \partial Z^{\le N} h^1_{TU} \big\|_{L^2 (\Sext_t)} \lesssim C_0 \ep t^{C\epsilon} \\  
		\label{hTU_en2}
		&		\big\| (2+r-t)^{\kappa-\rho}\ \op Z^{\le N} h^1_{TU} \big\|_{L^2(\Sext_t)} \lesssim C_0\ep t^{C\ep}\sqrt{l(t)}
	\end{align}
	for all $t\in [2, T_0)$, where $l\in L^1([2, T_0))$. The weighted Sobolev injection \eqref{sobolev2_ext} with $\beta = 1+2(\kappa-\rho)$ also yields the following pointwise bound 
	\begin{equation}\label{KS4_ext}
		|\partial Z^{\le N-3}h^1_{TU}|_{\kappa-\rho-\frac12}\lesssim C_0\ep t^{C\epsilon}.
	\end{equation}
	
	\subsection{The weak null terms}
	
	The goal of this subsection is to recover suitable higher order weighted $L^2(\Sext_t)$ estimates for the quadratic weak null terms $P_{\alpha\beta}(\partial h, \partial h)$ defined as
	\[
	P_{\alpha\beta}(\partial h, \partial h) = \frac{1}{4}\bar{g}^{\mu\rho}\bar{g}^{\nu\sigma}\left(\partial_\alpha h_{\mu\rho}\partial_\beta h_{\nu\sigma} - 2 \partial_\alpha h_{\mu\nu}\partial_\beta h_{\rho\sigma}\right).
	\]
	These estimates are based on the following remarkable property, highlighted in the works of Lindblad and Rodnianski \cite{LR03, LR05, LR10}, on Lemma \ref{lem:wave_cond_h1} and on the bounds \eqref{hTU_en1}, \eqref{KS4_ext} satisfied by $h^1_{TU}$.

	\begin{lemma} \label{lem:weak_null_frame}
		Let $\pi, \theta$ be arbitrary 2-tensors and $P$ be the quadratic form defined by
		\[
		P(\pi, \theta) = \frac{1}{4}\bar{g}^{\mu\rho}\bar{g}^{\nu\sigma}\left(\pi_{\mu\rho}\theta_{\nu\sigma} - 2 \pi_{\mu\nu}\theta_{\rho\sigma}\right).
		\]
		Then
		\[
		|P(\pi, \theta)|\lesssim |\pi|_{\q T \q U} |\theta|_{\q T \q U} + |\pi|_{\q L \q L}|\theta| + |\pi| |\theta|_{\q L \q L}.
		\]
	\end{lemma}
	
	\begin{proposition}\label{prop:weak_ext}
		Fix $i=0,1$. There exists some constant $C>0$ such that, under the a-priori energy assumptions \eqref{Boot1_ext}-\eqref{Boot2_ext}, we have
		\begin{equation}\label{weak_ext}
			\left\|(2+r-t)^{\frac{1}{2}+i+\kappa}\partial^i Z^{\le N} P_{\alpha\beta} \right\|_{L^2(\Sext_t)}  \lesssim C_0^2\ep^2 \big[t^{-1+C\ep} + t^{-1+2\sigma}\sqrt{l(t)} + \ep t^{-\frac{3}{2}+2\sigma}\big] .
		\end{equation}
	\end{proposition}
	\begin{proof}
		We write $h = h^1 + h^0$ and plug this decomposition into $P_{\alpha\beta}(\partial h, \partial h)$. Using \eqref{h0_estimate} and the energy bounds \eqref{boot1_ex}, \eqref{boot3_ex} it is straightforward to prove that there exists some small $\delta>0$ such that
		\[
		\sum_{\substack{i,j=0,1\\|K_1|+|K_2|\le N\\ |I_1|+|I_2|=i}} \hspace{-5pt} \| (2+r-t)^{\frac{1}{2}+i+\kappa}\, \partial^{I_1} Z^{K_1}h^0\cdot \partial \partial^{I_2} Z^{K_2}h^j \|_{L^2(\Sext_t)} 
		\lesssim C_0^2\ep^2 t^{-2+\delta}.
		\]
		Hence we focus on proving that estimate \eqref{weak_ext} holds true for $P_{\alpha\beta}(\partial h^1, \partial h^1)$.
		
		We start by noticing that for any multi-index $K$, $ Z^K P_{\alpha\beta}(\partial h^1, \partial h^1)$ is a linear combination of terms of the form $P_{\mu\nu}(\partial  Z^{K_1}h^1, \partial  Z^{K_2}h^1)$ for some multi-indexes $K_1, K_2$ such that $|K_1|+|K_2|\le |K|$ and $\mu,\nu =0,\dots, 4$.
		Applying Lemma \ref{lem:weak_null_frame} and Lemma \ref{lem:wave_cond_h1} we see that for $i=0,1$
		\begin{equation}\label{Z^Kweak}
			\begin{aligned}
				& |\partial^i Z^{\le N}P_{\alpha\beta}(\partial h^1, \partial h^1)|\lesssim \sum_{\substack{|K_1|+|K_2|\le N \\ |I_1|+|I_2| = i}} |\partial\partial^{I_1}Z^{K_1}h^1|_{\q T \q U} |\partial\partial^{I_2}Z^{K_2}h^1|_{\q T \q U} \\
				& + \sum_{\substack{|K_1|+|K_2|\le N \\ |I_1|+|I_2| = i}} |\op  \partial^{I_1} Z^{K_1}h^1| |\partial \partial^{I_2} Z^{K_2}h^1| + r^{-1}| \partial^{I_1} Z^{K_1}h^1| |\partial\partial^{I_2} Z^{K_2}h^1| \\
				& + \sum_{\substack{|K_1|+|K_2|+|K_3|\le N \\ |I_1|+|I_2|+|I_3| = i}} |\partial^{I_1} Z^{K_1}h^1|| \partial \partial^{I_2} Z^{K_2}h^1||\partial \partial^{I_3} Z^{K_3}h^1|\\
				& + \frac{M\chi_0(t/2\le r\le 3t/4)}{(1+t+r)^2}\sum_{j\le i} |\partial \partial^j Z^{\le N}h^1|,
			\end{aligned}
		\end{equation}
		where $\chi_0(t/2\le r\le 3t/4)$ is supported for $t/2\le r\le 3t/4$. Since the intersection of this support with the exterior region in bounded, it is immediate to see that the weighted $L^2(\Sext_t)$ norm of the last term in the above right hand side is bounded by $C_0 \ep^2 t^{-2}$.
		
		The cubic terms and the quadratic terms involving a tangential derivative have been estimated in proposition \ref{prop:null+cubic} and satisfy \eqref{cub_ext} and \eqref{null_ext} respectively. The weighted $L^2$ norm of the quadratic term with the extra $r^{-1}$ factor is bounded by $C_0^2\ep^2t^{-3/2+2\sigma}$, we leave the details to the reader. Finally, from \eqref{hTU_en1} and \eqref{KS4_ext} with $\rho>0$ such that $k>2\rho$
		\begin{multline*}
			\sum_{\substack{i=0,1\\ |K_1|+|K_2|\le N\\ |I_1|+|I_2|=i}} \left\|(2+r-t)^{\frac{1}{2}+ i +\kappa} |\partial \partial^{I_1} Z^{K_1}h^1|_{\q T \q U} |\partial \partial^{I_2} Z^{K_2}h^1|_{\q T \q U} \right\|_{L^2(\Sext_t)}\\
			\lesssim C_0\ep t^{-1+C\ep}\sum_{i=0,1} \left\|(2+r-t)^{i-\frac{1}{2} + \rho} |\partial  Z^{\le N}h^1|_{\q T \q U} \right\|_{L^2(\Sext_t)}\lesssim C_0^2\ep^2 t^{-1+2C\ep}.
		\end{multline*}

	\end{proof}
	
	From Lemma \ref{lem:weak_null_frame} and bounds \eqref{h0_estimate},  \eqref{KS1_ext}, \eqref{der_hLT_ext}, \eqref{KS4_ext} we also get the following pointwise estimate for the differentiated weak null terms.
	\begin{proposition}
		There exists a constant $C>0$ such that,
		under the a-priori assumptions \eqref{boot1_ex}-\eqref{boot4_ex}, we have that
		\begin{equation}\label{weak_point_ext}
			\begin{aligned}
				\big|Z^{\le N-3}P_{\alpha\beta}(\partial h, \partial h)(t)\big|_{\frac12}& \lesssim C_0^2 \ep^2  \big( t^{-1+2C\ep} + t^{-1+2\sigma}\sqrt{l(t)}\big).
			\end{aligned}
		\end{equation}
	\end{proposition}

	\subsection{Propagation of the energy estimates}
	
	We now proceed to the proof of proposition \ref{prop:bootstrap_ext}. We recall that for any multi-index $K$, the differentiated coefficients $Z^K h^1_{\alpha\beta}$ solve \eqref{h1_eqt_higher} with source term \eqref{source_ext_high}. We set $i=1$ if $Z^K = \partial Z^{K'}$ and $|K'|\le N$, $i=0$ if simply $|K|\le N$.
	Thanks to the smallness of $H$ provided by \eqref{h0_estimate} and \eqref{est_hLT2}, we apply \eqref{energy_ineq_ext} with $\Wf = Z^Kh^1_{\alpha\beta}$, $\Ff = F^K_{\alpha\beta} + F^{0,K}_{\alpha\beta}$, $w(q) = (2+r-t)^{1+2(i +\kappa)}$ and $\omega = x/|x|$. For every $t\in [2, T_0)$ we get the following energy inequality
	\begin{equation} \label{energy_ineq_high_ext}
		\begin{aligned}
			& E^{e, i+\kappa}(t, Z^Kh^1_{\alpha\beta})\\
			& + \int_{\tilde{\hcal}_{2,t}} (2+r-t)^{1+2(i+\kappa)}\Big[\Big(\frac{1}{2( 1+r^2)} +\chi\left(\frac{r}{t}\right)\chi(r)\frac{M}{2r} \Big)|\partial_t Z^K h^1_{\alpha\beta}|^2 + |\underline{\nabla}Z^Kh^1_{\alpha\beta}|^2 \Big] dxdy
			\\
			& \lesssim  E^{e, i+\kappa}(2, Z^Kh^1_{\alpha\beta}) + \iint_{\dext_{[2,t]}} \hspace{-10pt} (2+r-\tau)^{1+2(i+\kappa)}|(F^K_{\alpha\beta} + \partial^\mu {H_\mu}^\sigma \cdot \partial_\sigma Z^Kh^1_{\alpha\beta})\partial_t Z^Kh^1_{\alpha\beta}| d\tau dxdy\\
			& + \iint_{\dext_{[2,t]}} \hspace{-10pt} (2+r-\tau)^{1+2(i+\kappa)}\Big[|F^{0,K}_{\alpha\beta} \, \partial_t Z^Kh^1_{\alpha\beta}| +|\partial_t {H_\mu}^\sigma \cdot \partial_\sigma Z^Kh^1 \cdot\ \partial^\mu Z^Kh^1_{\alpha\beta} |\Big]\, d\tau dxdy \\
			& + \iint_{\dext_{[2,t]}}\hspace{-10pt} (2+r-\tau)^{2(i+\kappa)} | H^{\rho\sigma}\cdot\partial_\rho  Z^Kh^1_{\alpha\beta}\cdot \partial_\sigma Z^Kh^1_{\alpha\beta}|  \, d\tau dxdy\\
			& + \iint_{\dext_{[2,t]}}\hspace{-10pt} (2+r-\tau)^{2(i+\kappa)} |(-H^{0\sigma} + \omega_\bmj H^{\bmj\sigma}) \partial_\sigma 
			Z^Kh^1_{\alpha\beta}\cdot \partial_t Z^Kh^1 _{\alpha\beta}
			|  d\tau dxdy.
		\end{aligned}
	\end{equation}
	The above inequality is satisfied for all $t\in [2, T_0)$ and the implicit constant is a universal constant.

	\begin{proof}[Proof of proposition \ref{prop:bootstrap_ext}]
		We recall the definition of the source term $F^K_{\alpha\beta} = Z^KF_{\alpha\beta}$ where
		\[
		F_{\alpha\beta}(h)(\partial h, \partial h) = P_{\alpha\beta}(\partial h, \partial h) + \mathbf{Q}_{\alpha\beta}(\partial h, \partial h) + G_{\alpha\beta}(h)(\partial h, \partial h).
		\]
		The combination of estimates \eqref{null_ext}, \eqref{cub_ext} and \eqref{weak_ext} yields that for $i=0,1$
		\[
		\| (2+r-t)^{\frac{1}{2}+i+\kappa} F^K_{\alpha\beta}\|_{L^2(\Sext_t)}\lesssim  C_0^2\ep^2 \big(t^{-1+C\ep} + t^{-1+2\sigma}\sqrt{l(t)}\big)
		\]
		which together with the Cauchy-Schwarz inequality and energy assumptions \eqref{Boot1_ext}, \eqref{Boot2_ext} gives
		\begin{multline*}
			\iint_{\dext_{[2,t]}} \hspace{-10pt} (2+r-\tau)^{1+2(i+\kappa)}|F^K_{\alpha\beta}\cdot \partial_t Z^Kh^1_{\alpha\beta}| d\tau dxdy \\
			\lesssim \int_2^t  \| (2+r-\tau)^{\frac{1}{2}+i+\kappa} F^K_{\alpha\beta}\|_{L^2(\Sext_t)} E^{\text{e}, i+\kappa}(\tau, Z^K h^1_{\alpha\beta})^\frac{1}{2} d\tau \\
			\lesssim \int_2^t C_0^3 \ep^3( \tau^{-1+C\ep+\sigma} + \tau^{-1+2\sigma}\sqrt{l(\tau)}) d\tau \lesssim C_0^3 \ep^3 t^{C\ep +\sigma}.
		\end{multline*}
		From estimate \eqref{source0_higher} we also get that
		\[
		\iint_{\dext_{[2,t]}} \hspace{-10pt} (2+r-\tau)^{1+2(i+\kappa)}|F^{0,K}_{\alpha\beta} \cdot \partial_t Z^Kh^1_{\alpha\beta}|\, d\tau dxdy \lesssim C_0\ep^2.
		\]
		By injecting the above bounds, together with \eqref{dH_contribution_ext} and \eqref{H_contribution_ext}, into \eqref{energy_ineq_high_ext} we deduce the existence of a constant $\tilde{C}>0$ such that for all $t\in [2, T_0)$
		\[
		\begin{aligned}
			&	E^{\text{e}, i+\kappa}(t, Z^Kh^1_{\alpha\beta}) \\
			&	+   \int_{\tilde{\hcal}_{2,t}} (2+r-t)^{1+2(i+\kappa)}\Big[\Big(\frac{1}{2( 1+r^2)} +\chi\left(\frac{r}{t}\right)\chi(r)\frac{M}{2r} \Big) |\partial_t Z^K h^1_{\alpha\beta}|^2 + |\underline{\nabla}Z^Kh^1_{\alpha\beta}|^2 \Big] dxdy\\ 
			&	\le \tilde{C} E^{\text{e}, i+\kappa}(2, Z^Kh^1_{\alpha\beta}) + \tilde{C}C_0\ep^2 + \tilde{C}C_0^3\ep^3 t^{C\ep +\sigma} +  \tilde{C}\ep \int_2^t \frac{ E^{\text{e}, i+\kappa}(\tau, Z^Kh^1_{\alpha\beta})}{\tau}\,  d\tau.
		\end{aligned}
		\]
		By Gr\"onwall's inequality we then deduce that 
		\[
		E^{\text{e}, i+\kappa}(t, Z^Kh^1_{\alpha\beta})\le \tilde{C}\big(E^{\text{e}, i+\kappa}(2, Z^Kh^1_{\alpha\beta})+C_0\ep^2+ C_0^3\ep^3 t^{C\ep+\sigma} \big)t^{\tilde{C}\epsilon}.
		\]
		We denote the sum $C+\tilde{C}$ simply by $C$. Finally, we choose $C_0\gg 1$ sufficiently large so that $3\tilde{C} E^{\text{e}, i+\kappa}(2, Z^Kh^1_{\alpha\beta})\le (C_0\ep)^2$ and $3\tilde{C}<C_0$, then $\ep_0>0$ sufficiently small so that $3\tilde{C}C_0^3\ep_0<1$ and $2C\ep_0<\sigma$ to infer that
		\[
		E^{\text{e}, i+\kappa}(t, Z^Kh^1_{\alpha\beta})\le C_0^2\ep^2 t^{\sigma + C\epsilon}.
		\]
		As a byproduct, we also deduce that
		\begin{equation}\label{est_boundary_term}
			\int_{\tilde{\hcal}_{2,t}}\hspace{-10pt} (2+r-t)^{1+2(i+\kappa)}\Big[\Big(\frac{1}{2( 1+r^2)} +\chi\left(\frac{r}{t}\right)\chi(r)\frac{M}{2r} \Big) |\partial_t Z^K h^1_{\alpha\beta}|^2 + |\underline{\nabla}Z^Kh^1_{\alpha\beta}|^2 \Big] dxdy \lesssim C_0^2\ep^2 t^{\sigma + C\ep}.
		\end{equation}
		
	\end{proof}

	\section{The interior region} \label{sec:interior}
	
	The goal of this section is to prove the existence in the interior region $\dint$ of the solution $h^1_{\alpha\beta}$ to \eqref{h_equations} with data satisfying the hypothesis of theorem \ref{thm:Main}.
	The proof is based on a bootstrap argument in which the a-priori assumptions on the solutions are bounds on the higher order energies on truncated hyperboloids $\hin_s$ as well as pointwise decay bounds on a certain number of $Z$ derivatives acting on it. 
	
	We define the interior energy functional as follows
	\begin{equation} \label{interior_energy_functional}
		\begin{aligned}
			\enint(s,\hab) & :=\iint_{\hin_s} (s/t)^2 |\partial_t \hab|^2 + |\underline{\nabla} \hab|^2 \, dxdy  \\
			&  =\iint_{\hin_s} (s/t)^2|\nabla_x \hab|^2 + t^{-2}|\scal \hab|^2 + t^{-2}\sum_{1\le \bmi<\bmj\le 3}|\Omega_{\bmi\bmj}\hab|^2 + |\partial_y\hab|^2\,dxdy  ,
		\end{aligned}
	\end{equation}
	where $\scal = t\partial
	_t + x\cdot\nabla_x$ is the scaling vector field and $\Omega_{\bmi\bmj} = x_\bmi\partial_\bmj - x_\bmj\partial_\bmi$ are the Euclidean rotations in $\R^3$. Using Parseval's identity, we also define the energy functional associated to the zero-mode $\habf$ of the solution as well as that of its zero-average component $\habn$, so that
	\[
	\enint(s, \hab) = \enint(s, \habf) + \enint(s, \habn).
	\] 
	We fix $N, N_1\in \N$ two integers sufficiently large with $N\ge 14$ and $N_1=N-5$ and assume the existence of two positive constants $1\ll C_1\ll C_2$, as well as of a finite and increasing sequence of  parameters $0< \zeta_k,\gamma_k, \delta_k\ll 1$ with 
	\begin{equation} \label{condition_parameters}
		\zeta_i \ll \gamma_j\ll \delta_k, \quad \forall i,j,k, \qquad \gamma_i + \delta_j \ll \delta_k, \quad \forall i, j<k, 
	\end{equation}		
	such that:
	
	$\bullet$ for $s_0$ close to 2 (e.g. $s_0 =21/10$) and for any arbitrarily fixed $S_0>s_0$, the solution $\hab$ exists in the hyperbolic strip $\hin_{[s_0, S_0)}$,
	
	$\bullet$ for any $s\in [s_0, S_0)$, any multi-index $K = (I, J)$ of type $(N, k)$\footnote{We recall that a multi-index $K=(I, J)$ is said to be of type $(N, k)$ if $|I|+|J|\le N$ and $|J|\le k$}, it satisfies the following energy bounds
	\begin{gather}
		\enint(s,\partial Z^K\hab)^\frac12 + \enint(s, Z^K\hab)^\frac12 \le 2C_1\ep s^{\frac12+\zeta_k} \label{boot1_in}\\
		\enint(s, Z^K\habf)^\frac12 \le 2C_1\ep s^{\zeta_k}\label{boot2_in}
	\end{gather}
	and for multi-indexes $K$ of type $(N, k)$ with $k\leq N_1$
	\begin{equation}
		\enint(s, Z^K\hab)^\frac12 \le 2C_1\ep s^{\delta_k} \label{boot3_in}
	\end{equation}			
	
	$\bullet$	for any $s\in [s_0, S_0)$, it satisfies the following pointwise bounds
	\begin{equation} \label{bootW_in}
	\|t\, \Gamma^J\habf\|_{L^\infty_x(\hin_s)}\lesssim
			2 C_2\ep s^{\gamma_k}  \text{ with } |J|=k\le N_1, 
			\end{equation}
	\begin{equation} \label{bootKG_in}
		\begin{split}\| t^\frac12s\, \partial_{tx} (\partial^I \Gamma^J \habn)\|_{L^\infty_xL^2_y(\hin_s)} & + \| t^\frac32 \partial^{\le 1}_y (\partial^I \Gamma^J \habn)\|_{L^\infty_xL^2_y(\hin_s)}\\
	&	\le 		 \begin{cases}
		2C_2\ep, \ &\text{if } |I|\le N_1,\ |J|=0\\
		2C_2\ep s^{\gamma_k}, \ &\text{if } |I|+|J|\le N_1+1,\ |J|=k\le N_1. \\
	\end{cases}	
	\end{split}	
	\end{equation}

	The result we aim to prove states the following
	\begin{proposition} \label{prop:bootstrap_int}
		There exist two constants $1\ll C_1\ll C_2$ sufficiently large, a finite and increasing sequence of parameters $0\leq \zeta_k,\gamma_k, \delta_k\ll 1$ satisfying \eqref{condition_parameters} and $0<\ep_0\ll 1$ sufficiently small such that for every $0<\ep<\ep_0$, if $\hab$ is solution to \eqref{h_equations} in the hyperbolic strip $\hin_{[s_0, S_0)}$ that satisfies the bounds \eqref{boot1_in}-\eqref{bootKG_in} for all $s\in [s_0, S_0)$ and the energy bounds \eqref{Boot1_ext}-\eqref{Boot2_ext} globally in the exterior region, then for every $s\in [s_0, S_0)$ it actually satisfies the following:\\
		for multi-indexes $K$ of type $(N,k)$
		\begin{gather}
			\enint(s,\partial Z^K\hab)^\frac12 + \enint(s, Z^K\hab)^\frac12 \le C_1\ep s^{\frac12+\zeta_k} \label{boot1_in_enh} \\
			\enint(s, Z^K\habf)^\frac12 \le C_1\ep s^{\zeta_k} ;\label{boot2_in_enh}
		\end{gather}
		for multi-indexes $K$ of type $(N, k)$ with $k\leq N_1$
		\begin{equation}\label{boot3_in_enh}
			\enint(s, Z^K\hab)^\frac12 \le C_1\ep s^{\delta_k}
		\end{equation}
		and finally
		\begin{equation} \label{bootW_in_enh}
			\|t\,  \Gamma^J\habf\|_{L^\infty_x(\hin_s)}\le
				C_2\ep s^{\gamma_k} \quad   \text{if } |J|=k\le N_1, 
		\end{equation}
		\begin{equation} \label{bootKG_in_enh}
			\aligned
			\| t^\frac12s\, \partial_{tx} (\partial^I \Gamma^J \habn)\|_{L^\infty_xL^2_y(\hin_s)}  &+ \| t^\frac32 \partial^{\le 1}_y (\partial^I \Gamma^J \habn)\|_{L^\infty_xL^2_y(\hin_s)}\\
			&\le 
			\begin{cases}
				C_2\ep, \ &\text{if } |I|\le N_1,\ |J|=0\\
				C_2\ep s^{\gamma_k}, \ &\text{if } |I|+|J|\le N_1+1,\ |J|=k\le N_1 \\
			\end{cases}	
			\endaligned
		\end{equation}	
	\end{proposition}
	
	\begin{remark}
		The a-priori assumptions \eqref{boot1_in}-\eqref{bootKG_in} are satisfied when $s=s_0$ as a consequence of the assumptions on the initial data and the local existence result for the Einstein equations. The hyperbolic time $S_0$ in the above proposition is arbitrary. This implies the existence of the solution in the unbounded region $\hin_{[s_0, \infty)}$, hence in the full interior region. 
	\end{remark}
	\begin{remark}\label{remark_on_constants}
		The result stated above builds upon the energy and pointwise estimates the solution has been proved to satisfy in the exterior region. This can be already seen in the energy inequality \eqref{en_ineq_hab_high} below, where the energy flux through the separating hypersurface $\hin_{[s_0, s]}$, which is controlled by the exterior energies, appears in the right hand side of the inequality. Constants $C_1, \gamma_k, \delta_k$ in proposition \ref{prop:bootstrap_int} will in particular be chosen relative to $C_0, \sigma, \kappa$ so that $C_1\gg C_0$, $\sigma\ll  \gamma_k\ll\delta_k\ll \kappa$ and $\delta_k\ll \kappa-\sigma$ for all $k=0,\dots, N$. For this reason and throughout the rest of this section, we will often replace $C_0$ by $C_1$ in the inequalities obtained using bounds recovered in the exterior region. 
	\end{remark}
	
	In order to recover the enhanced energy bounds \eqref{boot1_in_enh}-\eqref{boot3_in_enh}, we compare the equation satisfied by $Z^K\hab$ and $Z^K \habf$ respectively with \eqref{linearized_wave} and apply the energy inequality of proposition \ref{prop:energy_ineq_hyperb}. 
	We recall that $Z^K\hab$ satisfies the following quasilinear wave equation
\begin{equation} \label{h1_eqt_higher1}
		\tilde{\Box}_g Z^Kh^1_{\alpha\beta} = F^K_{\alpha\beta} + F^{0,K}_{\alpha\beta}
	\end{equation}
	with source terms
	\begin{equation}\label{source_ext_high1}
 F^{0,K}_{\alpha\beta} = Z^K\tilde{\Box}_g h^0_{\alpha\beta}, \qquad		F^K_{\alpha\beta} = Z^KF_{\alpha\beta}(h)(\partial h, \partial h) - [Z^K, H^{\mu\nu}\partial_\mu\partial_\nu]h^1_{\alpha\beta} 
	\end{equation}
 and that the equation of $Z^K\habf$ 
	is obtained by averaging \eqref{h1_eqt_higher1} over $\S^1$
	\begin{equation}\label{eq:diff_hff}
		\Box_{x} Z^K \habf + (H^{\bmmu\bmnu})^\flat\cdot \partial_\bmmu\partial_\bmnu  Z^K\habf  + \big( (H^{\mu\nu})^\natural\cdot \partial_\mu\partial_\nu  Z^K\habn\big)^\flat=   F^{K, \flat}_{\alpha\beta} + F^{0,K}_{\alpha\beta}
	\end{equation} 
	where $F^{K,\flat}_{\alpha\beta} = \fint_{\S^1}F^K_{\alpha\beta}dy$.
	If the tensor $H$ satisfies suitable decay bounds in $\hin_{[s_0, S_0)}$, e.g. if for some $\delta>0$
	\[
	|H(t,x,y)|\lesssim \frac{\ep}{(1+t+r)^\frac34}, \quad |H^1_{LL}(t,x,y)|\lesssim \frac{\ep}{(1+t+r)^{1+\delta}}	
	\]
	we derive the following two energy inequalities, which hold for any $s\in [s_0, S_0)$
	\begin{equation} \label{en_ineq_hab_high}
		\begin{aligned}
			&\enint(s, Z^K\hab)\lesssim \enint(s_0, Z^K\hab) \\
			&+ \int_{\tilde{\hcal}_{s_0s}} \Big(\frac{1}{2( 1+r^2)} +\chi\left(\frac{r}{t}\right)\chi(r)\frac{M}{2r} \Big) |\partial_t Z^K\hab|^2 + |\underline{\nabla} Z^K\hab|^2\, dxdy \\
			&	+ \iint_{\hin_{[s_0,s]}} \hspace{-13pt}|F^K_{\alpha\beta} + F^{0,K}_{\alpha\beta} + {\partial^\mu H_\mu}^\sigma\cdot \partial_\sigma Z^K \hab| |\partial_t Z^K\hab| + \frac{1}{2}|{\partial_t H_\mu}^\sigma\cdot\partial_\sigma Z^K\hab \cdot \partial^\mu Z^K\hab| \, dtdxdy
		\end{aligned}
	\end{equation}
	and
	\begin{equation}\label{en_ineq_habf_high}
		\begin{aligned}
			&\enint(s, Z^K\habf)\lesssim \enint(s_0, Z^K\habf)\\
			& + \int_{\tilde{\hcal}_{s_0s}}\Big(\frac{1}{2( 1+r^2)} +\chi\left(\frac{r}{t}\right)\chi(r)\frac{M}{2r} \Big) |\partial_t Z^K\habf|^2 + |\underline{\nabla}_x Z^K\habf|^2\, dx \\
			&+ \iint_{\hin_{[s_0,s]}}  \hspace{-13pt } |F^{K,\flat}_{\alpha\beta} + F^{0,K}_{\alpha\beta} + {\partial^\mu H_\mu}^\sigma\cdot  \partial_\sigma Z^K \habf| |\partial_t Z^K\habf| + \frac{1}{2}|{\partial_t H_\mu}^\sigma\cdot \partial_\sigma Z^K\habf\cdot \partial^\mu Z^K\habf| \, dtdx \\
			&+ \iint_{\hin_{[s_0,s]}}|\big( (H^{\mu\nu})^\natural\cdot \partial_\mu\partial_\nu  Z^K\habn\big)^\flat| |\partial_t Z^K\habf|\, dtdx.
		\end{aligned}
	\end{equation}
	The energy flux through the boundary $\tilde{\hcal}_{s_0s}$, which appears in the right hand side of both of the above inequalities, is suitably controlled using \eqref{est_boundary_term} with $t = t_s = s^2/2$
	\begin{equation}\label{boundary_term_int}
		\int_{\tilde{\hcal}_{s_0s}}\Big(\frac{1}{2( 1+r^2)} +\chi\left(\frac{r}{t}\right)\chi(r)\frac{M}{2r} \Big)|\partial_t Z^K\hab|^2 + |\underline{\nabla}_{xy} Z^K\hab|^2\, dxdy \lesssim C_0^2\ep^2 s^{2\sigma + C\ep}.
	\end{equation}
	The current section is therefore mainly devoted to estimating the remaining integrals in the right hand side of \eqref{en_ineq_hab_high} and \eqref{en_ineq_habf_high}.

	\subsection{First sets of bounds} \label{sub:KSbounds}
	
	Below is a list of $L^2$ and $L^\infty$ bounds for $\hab, \hff_{\alpha\beta}$ and $\hnn_{\alpha\beta}$, which are a straightforward consequence of the a-priori bounds. All bounds stated below hold true also for tensor coefficients $H^{1,\mu\nu}$, after decomposition \eqref{dec_H} and bound \eqref{h0_estimate}.
	
	\subsubsection{$L^2_{xy}$ bounds on hyperboloids}

	From a-priori energy assumptions \eqref{boot1_in}-\eqref{boot2_in}, the Parseval identity and Poincar\'e inequality applied to the zero-average components $\habn$, we derive the following $L^2$ bounds on $\hin_s$, for any multi-index of type $(N,k)$ and $i=0,1$,
	\begin{align}
	& \left\| (s/t)\partial \partial^i Z^K \hab \right\|_{L^2(\hin_s)} + \left\| \pb \partial^i Z^K \hab \right\|_{L^2(\hin_s)}\le 2C_1 \epsilon s^{\frac12+\zeta_k}  \label{bootin1}\\
		& \left\| (s/t)\partial Z^K \habf \right\|_{L^2(\hin_s)}+ \left\| \pb Z^K \habf \right\|_{L^2(\hin_s)} \le 2C_1 \epsilon s^{\zeta_k} \label{bootin2}\\
		& \left\| (s/t)\partial Z^K \habn \right\|_{L^2(\hin_s)}+ \left\| \pb Z^K \habn \right\|_{L^2(\hin_s)} + \left\| Z^K \habn \right\|_{L^2(\hin_s)} \le 2C_1 \epsilon s^{\frac12+\zeta_k}\label{bootin3}
	\end{align}
	and
\begin{align}
	&  \left\|t^{-1} \scal Z^K \habf\right\|_{L^2(\hin_s)} +\left\|t^{-1} \Gamma  Z^K \habf\right\|_{L^2(\hin_s)} \le 2C_1\epsilon s^{\zeta_k} \label{bootin2.1}\\
	&  \left\| t^{-1}\scal Z^K \habn\right\|_{L^2(\hin_s)} + \left\|t^{-1}\Gamma Z^K \habn\right\|_{L^2(\hin_s)} \le 2C_1\epsilon s^{\frac12+\zeta_k}. \label{bootin3.1}
\end{align}

	\smallskip
	\noindent For multi-indexes $K$ of type $(N,k)$ with $k\leq N_1$ we have
	\begin{align}
		& \left\| (s/t)\partial  Z^K \hab \right\|_{L^2(\hin_s)} + \left\| \pb   Z^K \hab \right\|_{L^2(\hin_s)} \le 2C_1\epsilon s^{\delta_k} \label{bootin5}\\ 
		& \left\| (s/t)\partial  Z^K \habn \right\|_{L^2(\hin_s)}+ \left\| \pb  Z^K \habn \right\|_{L^2(\hin_s)} + \left\|  Z^K \habn \right\|_{L^2(\hin_s)} \le 2C_1 \epsilon s^{\delta_k} \label{bootin4}
	\end{align}
		and
\begin{align}
	& \left\| t^{-1}\scal  Z^K \hab\right\|_{L^2(\hin_s)} +  \left\|t^{-1} \Gamma  Z^K \hab\right\|_{L^2(\hin_s)} \le 2C_1\epsilon s^{\delta_k} \label{bootin5.1}.
\end{align}
	Moreover, provided that $\sigma,\ep\ll 1$ are sufficiently small so that $\sigma+C\ep\le \zeta_k$ for all $1\le k\le N$, from the Hardy inequality \eqref{Hardy_classical}, energy assumption \eqref{bootin1} and the exterior energy bound \eqref{enhanced_Boot1_ext} (recall that $t_s=s^2/2$) we also deduce the following bound when $|J|=k\le N$
	\begin{align}
		& \left\|r^{-1}\Gamma^J\hab \right\|_{L^2_{xy}(\hin_s)} \lesssim 2C_1\ep s^{\frac12+\zeta_k} + C_0\ep s^{\sigma+C\ep}\lesssim 2C_1\ep s^{\frac12+\zeta_k}\label{bootin_hardy1}\\
		&	\left\|r^{-1}\Gamma^J\habf \right\|_{L^2_{x}(\hin_s)} \lesssim 2C_1\ep s^{\zeta_k} + C_0\ep s^{\sigma+C\ep}\lesssim 2C_1\ep s^{\zeta_k}\label{bootin_hardy3}.
	\end{align}
	For $|J|=k\le N_1$, we instead get from \eqref{bootin5} that
	\begin{equation}\label{bootin_hardy2}
		\left\|r^{-1}\Gamma^J\hab \right\|_{L^2_{xy}(\hin_s)} \lesssim 2C_1\ep s^{\delta_k} + C_0\ep s^{\sigma+C\ep}\lesssim 2C_1\ep s^{\delta_k}.
	\end{equation}

	\subsubsection{$L^\infty_xL^2_y$ bounds on hyperboloids}\label{subsec:L2y_bounds}
	These are obtained using the Poincar\'e inequality, lemma \ref{lm:ksobolev}, relation $\pb_{\bmi}=t^{-1}\Omega_{0\bmi}$ and energy assumption \eqref{bootin3}. \\
	For multi-indexes $K$ of type $(N-2, k)$
	\begin{align}
		& \left\|t^\frac32\,  \partial_y^{\le1}  Z^K \habn \right\|_{L^\infty_xL^2_y(\hin_s)} + \left\|t^\frac12 s\, \partial_{tx}  Z^K \habn \right\|_{L^\infty_xL^2_y(\hin_s)}  \lesssim C_1\epsilon s^{\frac12+\zeta_{k+2}}; \label{bootinL2y.1}
	\end{align}
	for multi-indices $K$ of type $(N-3, k)$
	\begin{align}
		& \left\|t^\frac52\, \partial_y^{\le1} \pb_x  Z^K \habn \right\|_{L^\infty_xL^2_y(\hin_s)} + \left\|t^\frac32 s\,  \partial_{tx}  \pb_x Z^K \habn \right\|_{L^\infty_xL^2_y(\hin_s)} \lesssim C_1\epsilon s^{\frac12 + \zeta_{k+3}}; \label{bootinL2y.5}
	\end{align}
	for multi-indexes $K$ of type $(N-4, k)$
	\begin{align}
		& \left\|t^\frac72\,  \partial_y^{\le1} \pb^2_x  Z^K \habn \right\|_{L^\infty_xL^2_y(\hin_s)} + \left\|t^\frac52s \,  \partial_{tx} \pb^2_x Z^K \habn \right\|_{L^\infty_xL^2_y(\hin_s)} \lesssim C_1\epsilon s^{\frac12+\zeta_{k+4}} \label{bootinL2y.6}.
	\end{align}
		Moreover, for multi-indexes $K$ of type $(N-2, k)$ with $k\leq N_1-2$
	\begin{align}
		& \left\|t^\frac32\,  \partial_y^{\le1}  Z^K \habn \right\|_{L^\infty_xL^2_y(\hin_s)} + \left\|t^\frac12 s\, \partial_{tx}  Z^K \habn \right\|_{L^\infty_xL^2_y(\hin_s)}  \lesssim C_1\epsilon s^{\delta_{k+2}}; \label{bootinL2y.1bis}
	\end{align}
	for multi-indices $K$ of type $(N-3, k)$ with $k\le N_1-3$
	\begin{align}
		& \left\|t^\frac52\, \partial_y^{\le1} \pb_x  Z^K \habn \right\|_{L^\infty_xL^2_y(\hin_s)} + \left\|t^\frac32 s\,  \partial_{tx}  \pb_x Z^K \habn \right\|_{L^\infty_xL^2_y(\hin_s)} \lesssim C_1\epsilon s^{\delta_{k+3}}; \label{bootinL2y.5bis}
	\end{align}
	for multi-indexes $K$ of type $(N-4, k)$ with $k\le N_1-4$
	\begin{align}
		& \left\|t^\frac72\,  \partial_y^{\le1} \pb^2_x  Z^K \habn \right\|_{L^\infty_xL^2_y(\hin_s)} + \left\|t^\frac52s \,  \partial_{tx} \pb^2_x Z^K \habn \right\|_{L^\infty_xL^2_y(\hin_s)} \lesssim C_1\epsilon s^{\delta_{k+4}}. \label{bootinL2y.6bis}
	\end{align}
	
	\subsubsection{$L^\infty_{xy}$ bounds on hyperboloids}	
	
	These are obtained from the energy assumptions using Poincar\'e inequality, lemma \ref{lm:ksobolev} and Sobolev embedding on $\S^1$. \\
	For any multi-index $K$ of type $(N-3,k)$ and $i=0,1$
	\begin{align}
		& \label{KS1_h} \left\|t^\frac12 s^\frac12\, \partial \partial^i  Z^K \hab \right\|_{L^\infty_{xy}(\hin_s)}+ \left\|t^\frac32 s^{-\frac12}\, \pb \partial^i  Z^K \hab\right\|_{L^\infty_{xy}(\hin_s)}  \lesssim C_1\epsilon s^{\zeta_{k+3}} \\
		& \label{KS2_h} \left\|t^\frac12s\, \partial  Z^K \habn \right\|_{L^\infty_{xy}(\hin_s)} +\left\|t^\frac32\, \pb Z^K \habn \right\|_{L^\infty_{xy}(\hin_s)}+ \left\|t^\frac32   Z^K \habn\right\|_{ L^\infty_{xy}(\hin_s)} \lesssim C_1\epsilon s^{\frac12 + \zeta_{k+3}} \\
		& \label{KS1_hff} \left\|t^\frac12 s\, \partial  Z^K \habf\right\|_{L^\infty_x(\hin_s)} +  \left\|t^\frac32\, \pb  Z^K \habf \right\|_{L^\infty_x(\hin_s)} \lesssim C_1 \epsilon s^{\zeta_{k+2}}
	\end{align}
	and
	\begin{align}
		& \label{KS3_h} \|t^\frac12\, \scal \partial^i  Z^K \hab\|_{L^\infty_{xy}(\hin_s)} + \|t^\frac12\, \Gamma \partial^i  Z^K \hab\|_{L^\infty_{xy}(\hin_s)}\lesssim C_1\ep s^{\frac12 +\zeta_{k+3}} \\
		& \label{KS3_hff} \|t^\frac12\, \scal   Z^K \habf\|_{L^\infty_x(\hin_s)} + \|t^\frac12\, \Gamma   Z^K \habf\|_{L^\infty_x(\hin_s)}\lesssim C_1\ep s^{\zeta_{k+2}}.
	\end{align}					
	
	\smallskip
	\noindent From the pointwise bounds \eqref{bootKG_in} and the Sobolev embedding on $\S^1$ we also have that
	\begin{equation} \label{KS1_hnn}
		\left\| t^\frac32\, \partial^I \Gamma^J\habn\right\|_{L^\infty_{xy}(\hin_s)} \lesssim 
		\begin{cases}
		C_2\ep \ &\text{ if } |I|\le N_1,\ |J|=0\\
			C_2\ep s^{\gamma_k} \ &\text{ if } |I|+|J|\le N_1+1,\ |J|\le N_1
			\end{cases}
	\end{equation}
which coupled to \eqref{KS1_hff} gives that, for any $|I|+|J|\le N_1\le N-3, |J|=k\ge 0$,
	\begin{equation}
		\label{KS4_h}  \|t^\frac12 s \, \partial  (\partial^I\Gamma^J \hab) \|_{L^\infty_{xy}(\hin_s)}   + \|t^\frac32 \pb   (\partial^I\Gamma^J \hab) \|_{L^\infty_{xy}(\hin_s)}  \lesssim C_2\ep  s^{\max(\zeta_{k+2},\gamma_k)}.
	\end{equation}

	%
	
	\medskip
	\noindent For $|J|=k\le N-3$, we also have the following bound on coefficients without derivatives
	\begin{equation} \label{KS6_h}
		\|t^\frac12\, \Gamma^J \hab\|_{L^\infty_{xy}(\hin_s)}\lesssim C_1\epsilon s^{\delta_{k+2}}.
	\end{equation}
	Such a bound is satisfied by $\Gamma^J\habn$ thanks to \eqref{KS2_h}, while for $\Gamma^K\habf$ it is obtained by integration.
	More precisely, on the initial truncated hyperboloid $\hin_{s_0}$ such an estimate is obtained by integrating \eqref{KS1_hff} along the hyperboloid itself and up to the boundary $\partial\hin_{s_0} = S_{s_0, r_0}$ where $r_0:=\max\{r>0: S_{s_0, r}\subset \hin_{s_0}\} = \q O(1)$. In fact, for any $r\le r_0$ and $\omega = x/|x|$
	\[
	\begin{aligned}
		\big|\Gamma^J\habf(\sqrt{s_0^2+|x|^2},r\omega)\big| &\le \big|\Gamma^J\habf(\sqrt{s_0^2+r_0^2}, r_0\omega)\big| + \int_r^{r_0} \big|\pb \Gamma^J \habf(\sqrt{s_0^2+\rho^2}, \rho\omega)\big| d\rho \\
		&\lesssim \big|\Gamma^J\habf(\sqrt{s_0^2+r_0^2}, r_0\omega)\big|  + \int_r^\infty C_1\epsilon (s_0^2+\rho^2)^{-\frac{3}{4}+\frac{\zeta_{k+2}}{2}} d\rho \\
		&	\lesssim C_1\ep (s_0^2 + r^2)^{-\frac{1}{4} + \frac{\zeta_{k+2}}{2}}
	\end{aligned}
	\]
	where we estimated the first term in the above right hand side using the exterior bound \eqref{KS3_ext}. 	
	For all other points $(t,x)\in \hin_{(s_0, S_0)}$, the decay bound \eqref{KS6_h} is instead obtained by integrating \eqref{KS1_hff} along the rays with $t+r$ and $\omega$ fixed, i.e. along 
	\[ \delta: \lambda\in [r, \lambda^*] \mapsto\delta(\lambda)=(t+r-\lambda, \lambda \omega)\]
	where $\lambda^*$ is the first time $\delta(\lambda)$ intersects the lateral boundary $\tilde{\hcal}$ (in which case $\lambda^*=\frac{(t+r)(t+r-2)}{2(t+r-1)}$) or the initial hyperboloid $\hin_{s_0}$ (in which case $\lambda^*=\frac{(t+r)^2-s_0^2}{2(t+r)}$). In both cases $\lambda^* =O(t+r)$, so from the estimates on $\tilde{\hcal}$ following from \eqref{KS3_ext} or on the initial truncated hyperboloid $\hin_{s_0}$ derived above, we get
	\[
	\begin{aligned}
		|\Gamma^J \habf(t,x)|  \lesssim |\Gamma^J\habf(t+r-\lambda^*, \lambda^* \omega)| + \int_r^{\lambda^*} |(\partial \Gamma^J\habf)(\delta(\lambda))| d\lambda \\
		\lesssim (C_0+C_1)\ep\, (1+t+r)^{-\frac12}s^{\zeta_{k+2}} + C_1\epsilon (1+ t+r)^{-1+\zeta_{k+2}} \int_r^{\lambda^*} (t+r-2\lambda)^{-\frac12+\zeta_{k+2}}d\lambda \\
		\lesssim (C_0+ C_1)\epsilon\, (1+ t+r)^{-\frac12}s^{\zeta_{k+2}}.
	\end{aligned}
	\]
	
	
	\subsubsection{$L^\infty_{xy}$ bounds for the good metric coefficients} These refer to the enhanced bounds satisfied by the metric coefficients $H^1_{LT}$ as a consequence of the wave condition, more precisely of inequality \eqref{wave_cond_interior}, and of the pointwise bounds obtained above. As remarked above, these bounds are also satisfied by the $h^1_{LT}$ metric coefficients.
	
	\begin{proposition}\label{prop:HLL_int}
		Under the assumptions of proposition \ref{prop:bootstrap_int}, we have that for any $s\in [s_0, S_0)$, any multi-index $K$ of type $(N-3, k)$ and $i=0,1$
		\begin{equation}
			\| t\, \partial \partial^i Z^KH^1_{LT}\|_{L^\infty(\hin_s)} \lesssim C_1\ep s^{\zeta_{k+3}} \label{good_coefficient1}
		\end{equation} 
		and for any multi-index $K$ of type $(N_1, k)$
		\begin{equation}
			\| t^\frac32 \partial Z^K H^1_{LT}\|_{L^\infty(\hin_s)} \lesssim C_1\ep s^{\delta_{k+2}}.  \label{good_coefficient11}
		\end{equation}
		Furthermore, for any multi-index $K$ of type $(N-2, k)$ 
		\begin{gather}
			\|t^\frac{3}{2}\partial Z^K (H^1_{LT})^\flat\|_{L^\infty(\hin_s)} \lesssim C_1\ep s^{\zeta_{k+2}} \label{good_coeff1_flat}\\ 
			\|t^\frac12 (t/s)^2 Z^K (H^1_{LT})^\flat\|_{L^\infty(\hin_s)} \lesssim C_1\ep s^{\zeta_{k+2}}.\label{good_coefficient2}
		\end{gather}
	\end{proposition}
	\begin{proof}
		The proof of the above estimates is based on inequality \eqref{wave_cond_interior}.
		Estimate \eqref{good_coefficient1} (resp. \eqref{good_coefficient11}) is in fact obtained using \eqref{KS1_h} (resp. \eqref{KS4_h}) and \eqref{KS6_h}.
		Estimate \eqref{good_coeff1_flat} is deduced similarly, after taking the zero norm of both left and right hand side of \eqref{wave_cond_interior}. We recall, in particular, that for any two integrable functions $f$ and $g$ defined on $\S^1$, we have 
		\begin{equation}\label{dec_product}
			(f g)^\flat= f^\flat g^\flat + \big(f^\natural g^\natural\big)^\flat, \qquad (fg)^\natural = f^\flat g^\natural + f^\natural g^\flat + (f^\natural g^\natural)^\natural.
		\end{equation}
		Therefore, \eqref{good_coeff1_flat} follows from \eqref{KS2_h}, \eqref{KS1_hff} and \eqref{KS6_h}. 
		
		Finally, estimate \eqref{good_coefficient2} is satisfied in the interior of the cone $t=2r$ after \eqref{KS6_h}. In the portion of interior region where $t<2r$, it is instead obtained from the integration of \eqref{good_coeff1_flat} along the rays with $t+r = const$ and $\omega =const$ and up to the boundary of the interior region. From \eqref{est_hLT2} we derive that
		\[
		\begin{aligned}
			|Z^K(H^1_{LT})^\flat(t,x)| & \lesssim |Z^K(H^1_{LT})^\flat(\delta(\lambda^*))| + \int_r^{\lambda^*}|\partial Z^K(H^1_{LT})^\flat(\zeta(\lambda))|d\lambda \\
			& \lesssim |Z^K(H^1_{LT})^\flat(\delta(\lambda^*))| + \int_r^{\lambda^*}C_1\ep (t+r)^{-\frac32 + \frac{\zeta_{k+2}}{2}}(t+r-2\lambda)^{\frac{\zeta_{k+2}}{2}}d\lambda \\
			& \lesssim  C_0\ep (1+t+r)^{-\frac32+2\sigma}(t-r)^{\frac12-\kappa} + C_1\ep (t+r)^{-\frac32+\frac{\zeta_{k+2}}{2}}(t-r)^{1+\frac{\zeta_{k+2}}{2}}\\
			& \lesssim  C_1 \ep \frac{(t^2-r^2)^{1+\frac{\zeta_{k+2}}{2}}}{t^2} t^{-\frac12}.
		\end{aligned}
		\]
	\end{proof}
	
	\begin{remark}
		By combining together the wave gauge estimate \eqref{wave_cond_interior}, with the energy bounds \eqref{bootin2} and \eqref{bootin_hardy3} (respectively \eqref{bootin1} and \eqref{bootin_hardy1}) and the pointwise bounds \eqref{KS1_hff} and \eqref{KS6_h} (respectively \eqref{KS4_h} and \eqref{KS6_h}), we obtain the following estimate (resp. the second)
		\begin{equation}\label{good_coeff_L2}
			\big\|\partial Z^{K} H^{1, \flat}_{LT}  \big\|_{L^2(\hin_s)} \lesssim C_1\epsilon
			\begin{cases}
				s^{2\zeta_k} \quad  &\text{if } K \text{ is of type } (N, k) ,\\
				s^{\frac12+2\zeta_k} \quad  &\text{if } K \text{ is of type } (N+1, k) \text{ with } k\le N .
			\end{cases}
		\end{equation}
	\end{remark}

	\subsection{The null and cubic terms}
	The $L^2$ and $L^\infty$ bounds deduced in subsection \ref{sub:KSbounds} from the a-priori energy bounds, coupled with the a-priori pointwise bounds, allow us to suitably estimate the null and cubic contributions appearing in the equations for $Z^K\hab$ and $Z^K\habf$. 
	Quadratic and cubic interactions involving a $h^0$ factor are the simplest ones to analyze. They satisfy the following estimates, which follow from a straightforward application of the energy bounds \eqref{bootin1}, \eqref{bootin_hardy1} and the pointwise bounds \eqref{h0_estimate}, \eqref{KS4_h} and \eqref{KS6_h}.
	
	\begin{lemma}\label{lem:prod_h0}
		Let $Q = Q(\psi, \phi)$ and $C = C(\delta)(\phi, \psi)$ denote a quadratic and a cubic form respectively. Under the a-priori assumptions \eqref{boot1_in}-\eqref{bootKG_in}, there exists some small constant $0<\eta\ll 1$ depending linearly on $ \gamma_k, \delta_k,\zeta_k$, such that for $i=0,1$
		\begin{gather}
			\sum_{0\le l+m\le 1}\|\partial^i Z^{\le N}Q(\partial h^{l}, \partial h^m)\|_{L^2(\hin_s)}\lesssim C_1\ep^2 s^{-3/2+\eta}\label{est_quad_h0}\\
			\sum_{\substack{0\le l+m+n\le 2\\ l,m,n\le 1}}\|\partial^i Z^{\le N}C(h^l)(\partial h^m, \partial h^n)\|_{L^2(\hin_s)} \lesssim C_1^2\ep^3 s^{-2+\eta}.
		\end{gather} 
	\end{lemma}

	\begin{proposition}\label{prpnull}
		Under the a-priori assumptions \eqref{boot1_in}-\eqref{bootKG_in} there exists some small constant $0<\eta\le 3\delta_N\ll 1$ depending linearly on $\zeta_k,\gamma_k, \delta_k$, such that for $i=0,1$
		\begin{gather}
			\|\partial^iZ^{\le N} \mathbf{Q}_{\alpha\beta}(\partial h, \partial h)\|_{L^2(\hin_s)}\lesssim (C_1\ep)^2 s^{-1+\eta} \label{null_int} \\
			\|\partial^iZ^{\le N} G_{\alpha\beta}(h)(\partial h, \partial h)\|_{L^2(\hin_s)}\lesssim (C_1\ep)^3 s^{-3/2+\eta}.\label{cub_int}
		\end{gather}
		and multi-indexes $K$ of type $(N, k)$ with $k\le N_1$
		\begin{gather}
			\|Z^{K} \mathbf{Q}_{\alpha\beta}(\partial h, \partial h)\|_{L^2(\hin_s)}\lesssim (C_1\ep)^2 s^{-3/2+\eta} \label{null_lower_int}\\
			\|Z^{K} G_{\alpha\beta}(h)(\partial h, \partial h)\|_{L^2(\hin_s)}\lesssim (C_1\ep)^3 s^{-2+\eta}.\label{cub_lower_int}
		\end{gather}
		Moreover, for multi-indexes $K$ of type $(N,k)$
		\begin{equation} \label{null_zeromode}
			\|Z^K \mathbf{Q}^\flat_{\alpha\beta}(\partial h, \partial h)\|_{L^2(\hin_s)}\lesssim (C_1\ep)^2\left(\sum_{i=1}^4 s^{-1+\gamma_i+\zeta_{k-i}} + s^{-1+\zeta_k}.\right)
		\end{equation}
	\end{proposition}
	\begin{proof}
	Throughout the proof, $\eta$ will denote a small positive constant that depends linearly on $\gamma_k$ and $\delta_k$. We do not need to keep track of the explicit value of $\eta$, which may change from line to line.
	
		We start by decomposing each occurrence of $h$ in $\mathbf{Q}_{\alpha\beta}$ and $G_{\alpha\beta}$ into $h^0 + h^1$. Owing to lemma \ref{lem:prod_h0}, we only need to prove that the above estimates are satisfied for null and cubic interactions involving $h^1$ factors only.
		
		We recall that the admissible vector fields $Z$ preserve the null structure and that for any null form $Q$
		\begin{equation} \label{null_structure}
			|Q(\partial \phi, \partial\psi)|\le |\pb\phi| |\partial \psi| + |\partial\phi| |\pb \psi| + \frac{|t^2-r^2|}{t^2}|\partial\phi||\partial \psi|.
		\end{equation}
		For any $M\in \mathbb{N}$, we then have
		\[
		|\partial^i Z^{\le M} \mathbf{Q}_{\alpha\beta}(\partial h^1, \partial h^1)|\lesssim \hspace{-10pt}\sum_{\substack{|K_1|+|K_2|\le M \\ |I_1|+|I_2| = i}}\hspace{-10pt} |\pb \partial^{I_1}Z^{K_1}h^1| |\partial \partial^{I_2} Z^{K_2}h^1| + \frac{|t^2-r^2|}{t^2}|\partial \partial^{I_1}Z^{K_1}h^1| |\partial \partial^{I_2}Z^{K_2}h^1|.
		\]
		We observe that at least one of the two indexes $K_j$ in the above right hand side has length smaller than $\floor{M/2}$. When $M=N$ and since $N$ and $N_1$ are such that $\floor{N/2}+1\le N_1$, we deduce from energy bound \eqref{bootin1} and pointwise estimate \eqref{KS4_h} that
		\[
		\begin{aligned}
			&\sum_{\substack{|K_1|+|K_2|\le N \\ |I_1|+|I_2| = i}} \left\|\pb \partial^{I_1}Z^{K_1}h^1\cdot \partial \partial^{I_2} Z^{K_2}h^1\right\|_{L^2_{xy}(\hin_s)} + \left\|(s/t)^2\partial \partial^{I_1}Z^{K_1}h^1\cdot \partial \partial^{I_2}Z^{K_2}h^1\right\|_{L^2_{xy}(\hin_s)} \\
			&\lesssim \hspace{-10pt}\sum_{\substack{|K_1|+|K_2|\le N \\ |K_1|\le \floor{N/2}\\ |I_1|+|I_2| = i}} \Big(\left\|(t/s)\pb \partial^{I_1}Z^{K_1}\hab\right\|_{L^\infty_{xy}(\hin_s)} + \left\|(s/t)\partial \partial^{I_1}Z^{K_1}\hab\right\|_{L^\infty_{xy}(\hin_s)}\Big)\left\|(s/t)\partial \partial^{I_2}Z^{K_2}\hab\right\|_{L^2_{xy}(\hin_s)}\\
			& + \sum_{\substack{|K_1|+|K_2|\le N \\ |K_2|\le \floor{N/2}\\ |I_1|+|I_2| = i}} \left\|\partial\partial^{I_2}Z^{K_2}\hab\right\|_{L^\infty_{xy}(\hin_s)}\left\|\pb\partial^{I_1}Z^{K_1}\hab\right\|_{L^2_{xy}(\hin_s)} \lesssim (C_1\ep)^2 s^{-1+\eta}.
		\end{aligned}
		\]
		Similarly, when $K$ is of type $(N, k)$ with $k\le N_1$, we get from energy bound \eqref{bootin5} and pointwise bound \eqref{KS4_h} that
		\[
		\begin{aligned}
			\sum_{\substack{|K_1|+|K_2|\le N \\ |I_1|+|I_2| = i}} \left\|\pb \partial^{I_1}Z^{K_1}h^1\cdot \partial \partial^{I_2} Z^{K_2}h^1\right\|_{L^2_{xy}(\hin_s)} + \left\|(s/t)^2\partial \partial^{I_1}Z^{K_1}h^1\cdot \partial \partial^{I_2}Z^{K_2}h^1\right\|_{L^2_{xy}(\hin_s)} \\
			\lesssim (C_1\ep)^2 s^{-\frac32+\eta}.
		\end{aligned}
		\]
		Estimate \eqref{null_int} can be slightly improved if we only consider the zero mode of the quadratic null interactions. We recall that the zero-mode of a product decomposes as in \eqref{dec_product}, hence
		\[
		\mathbf{Q}^\flat_{\alpha\beta}(\partial h^1, \partial h^1) = \mathbf{Q}_{\alpha\beta}(\partial \hff, \partial \hff) + \big( \mathbf{Q}_{\alpha\beta}(\partial \hnn, \partial \hnn)\big)^\flat.
		\]
		The pure zero-mode interactions are treated using the null structure. From the energy bound \eqref{bootin2} and the pointwise bound \eqref{KS1_hff} we derive that
		\[
		\begin{aligned}
			\left\|Z^{\le N}\mathbf{Q}_{\alpha\beta}(\partial \hff, \partial \hff) \right\|_{L^2_x(\hin_s)}\lesssim \ep^2 s^{-\frac32+\eta}.
		\end{aligned}
		\]
		The null structure is instead irrelevant when estimating the quadratic interactions of pure non-zero modes. Using the Cauchy-Schwartz and Poincar\'e inequalities and assuming $N, N_1$ are such that $\floor{N/2}+1\le N_1$, we derive from the pointwise bound \eqref{bootKG_in}, the energy bounds  \eqref{bootin1}, \eqref{bootin3} and \eqref{bootin4} (recall that $N_1=N-5$) and relation \eqref{condition_parameters}  that 
		\[
		\begin{aligned}
			\left\|\big( Z^K \mathbf{Q}_{\alpha\beta}(\partial \hnn, \partial \hnn)\big)^\flat \right\|_{L^2_x(\hin_s)}&\lesssim \sum_{|K_1|+|K_2|\le |K|}\left\|\partial Z^{K_1}\hnn \cdot \partial Z^{K_2}\hnn\right\|_{L^1_yL^2_x(\hin_s)} \\
			&\lesssim \sum_{\substack{|K_1|+|K_2|\le |K|\\ |K_1|\le \floor{|K|/2}}}\big\|\partial Z^{K_1}\hnn\big\|_{L^\infty_{x}L^2_y(\hin_s)}\big\|\partial Z^{K_2}\hnn\big\|_{L^2_{xy}(\hin_s)}\\
			&\lesssim C_1C_2\ep^2 (s^{-\frac{3}{2}+\eta}+s^{-1+\zeta_k}+\sum_{i=0}^4 s^{-1+\gamma_i+\zeta_{k-i}}).
		\end{aligned}
		\]

		As concerns the cubic terms, we have that
		\[
		|\partial^i Z^{\le M}G_{\alpha\beta}(h^1)(\partial h^1, \partial h^1)|\lesssim \sum_{\substack{|K_1|+|K_2|+|K_3|\le M\\ |I_1|+|I_2|+|I_3| = i }} |\partial^{I_1}Z^{K_1}h^1| |\partial \partial^{I_2} Z^{K_2}h^1| |\partial \partial^{I_3} Z^{K_3}h^1| .
		\]
		When $M=N$, we get from energy bound \eqref{bootin1} and pointwise bounds \eqref{KS4_h}, \eqref{KS6_h} that
		\begin{multline*}
			\sum_{\substack{|K_1|+|K_2|+|K_3|\le N\\|K_1|+|K_2|\le \floor{N/2}\\  |I_1|+|I_2|+|I_3| = i }} \left\| \partial^{I_1}Z^{K_1}h^1 \cdot\partial \partial^{I_2} Z^{K_2}h^1\cdot \partial \partial^{I_3} Z^{K_3}h^1\right\|_{L^2_{xy}(\hin_s)}\\
			\lesssim \sum_{\substack{|K_1|+|K_2|+|K_3|\le N\\|K_1|+|K_2|\le \floor{N/2}\\  |I_1|+|I_2|+|I_3| = i }} \left\|(t/s)\partial^{I_1}Z^{K_1}h^1\cdot \partial\partial^{I_2}Z^{K_2}h^1\right\|_{L^\infty_{xy}(\hin_s)}\left\| (s/t)\partial\partial^{I_3}Z^{K_3}h^1 \right\|_{L^2_{xy}(\hin_s)} \\
			\lesssim (C_1\ep)^3s^{-\frac32+\eta}
		\end{multline*}
		and from \eqref{bootin_hardy1}, \eqref{KS1_h} and \eqref{KS4_h} 
		\[
		\begin{aligned}
			&\sum_{\substack{|K_1|+|K_2|+|K_3|\le N\\|K_2|+|K_3|\le \floor{N/2}\\  |I_2|+|I_3| = i }} \left\| Z^{K_1}h^1 \cdot \partial \partial^{I_2} Z^{K_2}h^1\cdot \partial \partial^{I_3} Z^{K_3}h^1\right\|_{L^2_{xy}(\hin_s)}\\
			&\lesssim \sum_{\substack{|K_1|+|K_2|+|K_3|\le N\\|K_2|+|K_3|\le \floor{N/2}\\  |I_2|+|I_3| = i }} \left\|r^{-1}Z^{K_1}h^1\right\|_{L^2_{xy}(\hin_s)}\left\|r\partial \partial^{I_2} Z^{K_2}h^1\cdot \partial \partial^{I_3} Z^{K_3}h^1\right\|_{L^\infty_{xy}(\hin_s)}\lesssim (C_1\ep)^3 s^{-\frac32+\eta}.
		\end{aligned}
		\]
		Similarly, we deduce from \eqref{bootin5}, \eqref{KS4_h}, \eqref{KS6_h} that when $K$ is of type $(N, k)$ with $k\le N_1$
		\[
		\sum_{\substack{|K_1|+|K_2|+|K_3|\le N\\|K_1|+|K_2|\le \floor{N/2}\\  |I_1|+|I_2|+|I_3| = i }} \left\| \partial^{I_1}Z^{K_1}h^1 \cdot \partial \partial^{I_2} Z^{K_2}h^1\cdot \partial \partial^{I_3} Z^{K_3}h^1\right\|_{L^2_{xy}(\hin_s)}\lesssim (C_1\ep)^3 s^{-2+\eta}
		\]
		while from \eqref{bootin_hardy2} and \eqref{KS4_h}
		\[
		\sum_{\substack{|K_1|+|K_2|+|K_3|\le N\\|K_2|+|K_3|\le \floor{N/2}\\  |I_2|+|I_3| = i }} \left\| Z^{K_1}h^1 \cdot \partial \partial^{I_2} Z^{K_2}h^1\cdot \partial \partial^{I_3} Z^{K_3}h^1\right\|_{L^2_{xy}(\hin_s)}\lesssim (C_1\ep)^3 s^{-2+\eta}.
		\]
		
	\end{proof}

	\begin{lemma}
	There exists some small constant $0<\eta\le 3\delta_N\ll 1$ depending linearly on $\zeta_k, \gamma_k, \delta_k$ such that for multi-indexes $K$ of type $(N_1,k)$ we have
		\begin{align}
			\|t^2 Z^K  \mathbf{Q}_{\alpha\beta}(\partial h, \partial h)\|_{L^\infty_{xy}(\hin_s)}& \lesssim C_2^2 \ep^2 s^{-1+\eta}  \label{eq:est_null_point}\\
			\|t^\frac32 Z^K G_{\alpha\beta}(h)( \partial h, \partial h)\|_{L^\infty_{xy}(\hin_s)} &\lesssim  C_1 C_2^2\ep^3 s^{-2 + \eta}. \label{eq:est_cub_point}
		\end{align}
		while for any quadratic form $N$ we have
		\begin{equation}\label{eq:Linf_quadratic}
			\left\|t\, \partial Z^K N(\partial h, \partial h) \right\|_{L^\infty_{xy}(\hin_s)} \lesssim C_1^2\epsilon^2 s^{-2 + \eta}.
		\end{equation}
	\end{lemma}
	\begin{proof}
		This is a direct consequence of bounds \eqref{KS4_h} and \eqref{KS6_h}.
	\end{proof}

	\subsection{Second order derivatives of the zero modes}
	As expected for solutions to wave equations on $\R^{1+3}$, the second order derivatives of the differentiated coefficients $Z^K\habf$ enjoy better decay estimates compared to \eqref{bootin2} and \eqref{KS4_h} respectively. 
	
	\begin{lemma}\label{lem:sup_second_derivatives_wave}
		There exists some small constant $0<\eta\le 2\delta_N\ll 1$ depending linearly on $\zeta_k, \gamma_k, \delta_k$ such that for any multi-index $K$ of type $(N_1, k)$ 
		\begin{equation}\label{second_derivatives2}
			\left\|t^\frac32 (s/t)^2\partial_t^2 Z^K \hff_{\alpha\beta}\right\|_{L^\infty(\hin_s)}\lesssim C_2\epsilon s^{-1+\eta}.
		\end{equation}
	\end{lemma}
	
	\begin{proof}
		The flat wave operator can be expressed in terms of the hyperbolic derivatives as 
		\begin{equation}\label{flat_wave_hyp}
			-\Box_{tx}  = (s/t)^2\partial^2_t  +2 (x^\bma/t)\pb_\bma \partial_t - \pb^\bma\pb_\bma    + \frac{r^2}{t^3}\partial_t  + \frac{3}{t}\partial_t 
		\end{equation}
		and from \eqref{null-hyp2} the curved part can be written as
		\[
		(H^{\bmmu\bmnu})^\flat\partial_\bmmu\partial_\bmnu = (H^{UV})^\flat c^{\bmalph\bmbeta}_{UV}\pb_\bmalph\pb_\bmbeta + (H^{UV})^\flat d^\bmmu_{UV} \pb_\bmmu
		.
		\]
		Equation \eqref{eq:diff_hff} for $ Z^K \hff_{\alpha\beta}$ becomes
		\begin{equation}\label{eq_second_derivatives}
			\begin{aligned}
				\left((s/t)^2 + (H^{UV})^\flat c_{UV}^{00}\right)\partial_t^2  Z^K\hff_{\alpha\beta}  & = 
				S_1( Z^K\hff_{\alpha\beta})+ S_2( Z^K\hff_{\alpha\beta})
				\\
				&+  F^{K, \flat}_{\alpha\beta} + (F^{0,K}_{\alpha\beta})^\flat - \big( (H^{\mu\nu})^\natural\cdot \partial_\mu\partial_\nu  Z^K\habn\big)^\flat
			\end{aligned}
		\end{equation}
		where $F^{K, \flat}_{\alpha\beta}$ is the average over $\S^1$ of the source term in \eqref{source_ext_high1}, and $S_1(p), S_2(p)$ are defined as follows for an arbitrary two tensor $p$
		\begin{equation}\label{eq:def_S1_S2}
			\gathered
			S_1(p) := -\big(2(x^\bma/t)\pb_\bma \partial_t -\pb^\bma\pb_\bma + \frac{r^2}{t^3}\partial_t  + \frac{3}{t}\partial_t \big)p
			\\
			S_2(p) :=- \big( (H^{UV})^\flat c_{UV}^{\bma \bmbeta}\pb_\bma\pb_\bmbeta+  (H^{UV})^\flat c_{UV}^{\bmalph \bmb}\pb_\bmalph\pb_\bmb+ (H^{UV})^\flat d_{UV}^{\bmmu}\pb_\bmmu\big)p.
			\endgathered
		\end{equation}
		
		We note that, if $\ep$ is sufficiently small, relation \eqref{H1-UV-cUV-exp} with $\pi = (Z^KH^{UV})^\flat$, bounds \eqref{h0_estimate}, \eqref{KS6_h} and \eqref{good_coefficient2} yield
		\begin{equation}\label{good_coeff_comm}
			\big|(Z^{K}H^{UV})^\flat c_{UV}^{00}\big|\lesssim C_1\epsilon t^{-1/2}(s/t)^2s^{\delta_{k+2}}\lesssim (1/2)(s/t)^2,
		\end{equation}
		hence it is enough to prove that the right hand side in \eqref{eq_second_derivatives} is bounded by $C_1\ep t^{-3/2}s^{-1+2\zeta_{k+3}}$.	
		
		This is the case for the $S_1, S_2$ terms with $p=Z^K\hff_{\alpha\beta}$, as follows by using the pointwise bounds \eqref{KS1_hff}, \eqref{KS6_h} and \eqref{null-hyp-coeff} together with the fact that  $\pb_\bma = \Omega_{0\bma}/t$. 
		
		All quadratic and cubic terms in $F^{K,\flat}_{\alpha\beta}$, except for the commutator terms,  are estimated using \eqref{eq:Linf_quadratic} and \eqref{eq:est_cub_point} respectively, while from \eqref{dec_F0}, \eqref{h0_estimate} and \eqref{KS1_hff} we have
		\[
		\| t^3 (F^{0,K}_{\alpha\beta})^\flat\|_{L^\infty(\hin_s)}\lesssim \ep.
		\]
		
		From the pointwise bounds \eqref{bootKG_in}, \eqref{KS2_h} and relation \eqref{condition_parameters}, we have that for multi-indexes $K$ of type $(N_1, k)$
		\[
		\begin{aligned}
			& \big|  Z^K \big( (H^{\mu\nu})^\natural\cdot \partial_\mu\partial_\nu\hnn_{\alpha\beta}  \big)^\flat\big| \\
			& \lesssim \sum_{\substack{|K'|\le |K| }}  \big\| H^{1,\natural}\big\|_{L^\infty_xL^2_y}\ \big\|\partial_\mu\partial_\nu  Z^{K'}\hnn_{\alpha\beta} \big\|_{L^\infty_xL^2_y}  + \hspace{-10pt} \sum_{\substack{|K_1|+|K_2|\le |K|\\ |K_2|<|K| }}  \big\| ( Z^{K_1}H^{1,\mu\nu})^\natural\big\|_{L^\infty_xL^2_y} \big\|\partial_\mu\partial_\nu  Z^{K_2}\hnn_{\alpha\beta} \big\|_{L^\infty_xL^2_y}\\
			&  \lesssim C_1 C_2\ep^2 t^{-2}s^{-\frac12+ \zeta_{k+3}}+ C_2^2\ep^2 t^{-2}s^{-1+\gamma_k} .
		\end{aligned}
		\]
		We observe that the above bound can be improved if $|K|\le N_1-1$ to $C_2^2\ep^2 t^{-2}s^{-1+\gamma_k} $.
		
		Finally, for the purpose of this proof it is enough to write the commutator terms which only involve zero-modes as
		\[
		\sum_{\substack{|K_1| + |K_2|\le |K| \\ |K_2|<|K|}} Z^{K_1}H^{1,\flat}\cdot \partial^2 Z^{K_2}\habf	
		\]
		which is estimated by $C_1C_2\ep^2 t^{-3/2}s^{-1+2\zeta_{k+2}}$ using \eqref{KS1_hff} whenever $Z^{K_1}$ contains at least a usual derivative, and \eqref{bootW_in} together with \eqref{KS1_hff} otherwise. 	
		
		The conclusion of the proof follows by assuming $\ep$ sufficiently small so that $C_2\ep<1$.
	\end{proof}
	
	We highlight that the proof of the above Lemma also yields the following
	\begin{corollary}
		For every multi-index $K$ of type $(N_1-1,k)$ we have
		\begin{equation}\label{comm_int_point_zeromode}
			\big\|t^2 \big( (H^{1, \mu\nu})^\natural\cdot \partial_\mu\partial_\nu Z^K\habn\big)^\flat\big\|_{L^\infty(\hin_s)} + \big\|t^2 \big([Z^K, H^{1,\mu\nu}\partial_\mu\partial_\nu]\hab\big)^\flat\big\|_{L^\infty(\hin_s)}\lesssim C_1^2\ep^2 s^{-1+\eta}.
		\end{equation}
	\end{corollary}

	\begin{lemma}
		There exists some small constant $0<\eta\le 2\delta_N\ll 1$ depending linearly on $\zeta_k, \gamma_k, \delta_k$ such that for any multi-index $K$ of the type $(N-1, k)$ we have
		\begin{align}\label{second_derivatives_L2}
		&	\left\|(s/t)^2\partial^2_t Z^K\hff_{\alpha\beta} \right\|_{L^2(\hin_s)} 
			\lesssim C_1\ep
			s^{-1+\eta} .\\
		\label{third_derivative_L2}	&	\left\|(s/t)^2\partial^2_t \partial Z^K\hff_{\alpha\beta} \right\|_{L^2(\hin_s)} 
			\lesssim C_1\ep
			s^{-\frac{1}{2}+\eta} .
		\end{align}
	\end{lemma}
	\begin{proof}
		This is based on \eqref{eq_second_derivatives} and \eqref{eq:def_S1_S2}.
		We only detail the proof of estimate \eqref{second_derivatives_L2}, 	since \eqref{third_derivative_L2} is obtained in a similar way by replacing energy bound \eqref{bootin2} with \eqref{bootin1} whenever it occurs.

		We make use of \eqref{eq_second_derivatives}, \eqref{eq:def_S1_S2} and \eqref{good_coeff_comm} to estimate the $S_1, S_2$ terms. From $\pb_\bma = (1/t) \Omega_{0\bma}$, the energy bound \eqref{bootin2} and the pointwise bounds \eqref{h0_estimate}, \eqref{KS6_h} we derive that 
		\[
		\sum_{i=0}^1\left\|s S_i(Z^K\hff_{\alpha\beta})\right\|_{L^2(\hin_s)}\lesssim \|(s/t)\partial \Gamma^{\le 1} Z^K\hff_{\alpha\beta}\|_{L^2(\hin_s)}\lesssim C_1\ep s^{\zeta_{k+1}}.
		\]
		
		We recall that the quadratic null terms satisfy \eqref{null_zeromode}, while the cubic terms verify \eqref{cub_int}. In general, the zero-mode of a quadratic interaction can be estimated as follows: using \eqref{bootin2} and \eqref{KS1_hff} we find
		\[
		\sum_{|K_1|+|K_2|\le |K|}\hspace{-10pt} \|\partial Z^{K_1}\hff \cdot \partial Z^{K_2}\hff\|_{L^2_{x}(\hin_s)}\lesssim \hspace{-5pt} \sum_{\substack{|K'|\le |K|}}C_1\ep s^{-1+\delta_N}\|(s/t)\partial Z^{K'}\hff\|_{L^2(\hin_s)}\lesssim C_1^2\ep^2 s^{-1+2\delta_N}
		\]
		while, since $\floor{N/2}+1\le N_1$, from \eqref{bootKG_in} and \eqref{bootin3} we get		
		\[\aligned
		\sum_{|K_1|+|K_2|\le |K|} \|\big(\partial Z^{K_1}\hnn \cdot \partial Z^{K_2}\hnn\big)^\flat\|_{L^2_x(\hin_s)} \lesssim  \sum_{|K'|\le |K|}C_2\ep  s^{-\frac32+\gamma_{N_1}} \|(s/t)\partial Z^{K'}\hnn\|_{L^2_{xy}(\hin_s)}
		\\	\lesssim C_1C_2\ep^2 s^{-1+2\delta_N}.
		\endaligned
		\]
		From \eqref{bootKG_in} and \eqref{bootin3} we also have
		\[
		\left\| \big( (H^{\mu\nu})^\natural\cdot \partial_\mu\partial_\nu  Z^K\habn\big)^\flat\right\|_{L^2(\hin_s)}\lesssim C_1C_2 \ep^2 s^{-1+\zeta_k}.
		\]
		Finally, from \eqref{dec_F0}, estimate \eqref{h0_estimate}, the Hardy inequality \eqref{Hardy_classical} and energy estimates \eqref{Boot1_ext} and \eqref{bootin2} we get
		\begin{equation}\label{est_F0_int}
			\|Z^K \tilde{\Box}_g h^0\|_{L^2(\hin_s)}\lesssim \ep s^{-\frac32}.
		\end{equation}
	\end{proof}

	\subsection{The $h^1_{TU}$ coefficients}\label{subsec:good_coeff_II}
	
	In this subsection we show that for any $T\in \q T$ and $U\in \q U$, the differentiated metric coefficients $\partial Z^{\le N_1-1}h^1_{TU}$ satisfy better lower order decay bounds than the ones in \eqref{KS4_h} obtained via Klainerman-Sobolev inequalities. More precisely, we want to prove the following:
	\begin{proposition} \label{prp:good_coeffII}
		There exists a constant $C>0$ such that, for any multi-index $K$ of type $(N_1-1,k)$ and any multi-indexes $I$ with $|I|\le N_1-1$, we have that
		\begin{align}
			& \sup_{\S^1} |\partial_\alpha  Z^K h^1_{TU}|\lesssim C_1 \epsilon t^{-1+C\epsilon} \label{good1} \\
			& \sup_{\S^1} |\partial_\alpha \partial^I h^1_{TU}|\lesssim C_1 \epsilon t^{-1}. \label{good2}
		\end{align}
	\end{proposition}
	We observe that the above bounds are satisfied by any $\partial Z^K\hnn_{TU}$ with $|K|\le N_1$ as a consequence of \eqref{KS1_hnn}. 
	In the interior of the cone $t=2r$ they follow immediately from \eqref{KS1_hff}, where one has $t^2-r^2\ge 3t^2/4$ and consequently
	\[
	|\partial Z^{\le N-3} \hff_{a\beta}(t,x,y)|\lesssim C_1\ep t^{-3/2+\zeta_{N}}.
	\] 
	
	For all other points in the interior region, \eqref{good1}-\eqref{good2} are instead obtained by integration along characteristics, see lemma \ref{lem:characteristics}. The difference between $h^1_{TU}$ and any general coefficient $\hab$ relies on the fact that weak null terms do not appear in the equation \eqref{h1TU_higher} satisfied by the former. We only sketch the proof of the following lemma, see \cite{LR10} for additional details.
	
	\begin{lemma}\label{lem:characteristics}
		Let $s>s_0$ and $\q D_s$ be the set of points $(t,x)$ in the cone $t/2<|x|<2t-3$ with $t\ge2$ such that
		\[
		|x|\le \sqrt{t^2-s^2} \text{ if } (t,x)\in \dint \quad \text{ or }\quad t\le s^2/2
		\text{ otherwise.} 		\]
		We denote by $\partial_B\q D_s$ the lateral boundary of $\q D_s$, i.e.
		\[
		\partial_B\q D_s :=\{(t, x) :  |x|=t/2 \text{ and } t\le 8/3 \text{ or } |x| =2t-3 \text{ and } t\le s^2/2\}.
		\]
		Let $u$ be a solution to the wave equation on the curved 4-dimensional spacetime $\tilde{\Box}_g u = F$,
		where $g = (g_{\bmmu\bmnu})$ is a Lorenztian metric, $g^{-1} = (g^{\bmmu\bmnu})$ is its inverse and $F$ is some smooth source term. Let $m = (m_{\bmmu\bmnu})$ denote the Minkowski metric, $m^{-1}=(m^{\bmmu\bmnu})$ its inverse and $\pi^{\bmmu\bmnu}:= g^{\bmmu\bmnu} - m^{\bmmu\bmnu}$. For any spacetime point $(t, x)\in \q D_s$, let $(\tau, \varphi(\tau; t, x))$ be the integral curve of the vector field
		\[
		\partial_t + \frac{1-\pi_{LL}/4}{1+\pi_{LL}/4}\partial_r
		\]
		passing through $(t, x)$, i.e. $\varphi(t;t,x)=x$. Assume that 
		\[
		|\pi^{L\Lb}(t,x)|<1/4 \quad \text{and}\quad	|\pi_{LL}(t,x)|\le \ep \frac{|t-r|}{t+r}, \qquad \forall (t,x)\in \q D_s.
		\] 
		Then for any $(t,x)\in \q D_s$ one has that
		\begin{equation}\label{eq:est_characteristics}
			\begin{aligned}
				t|(\partial_t-\partial_r)u(t,x)| & \lesssim \sup_{\partial_B \q D_s} |(\partial_t -\partial_r)(ru)| +  \int_2^t |M[u, \pi](\tau) |d\tau + \int_2^t \tau |F|_{(\tau, \varphi(\tau;t,x))} d\tau 
			\end{aligned}
		\end{equation}
		where
		\begin{equation}\label{M}
			M[u,\pi](\tau) =\Big(r |\Delta_{\S^1} u| + |\pi|_{\q L \q T}\big( r |\op \partial  u|+ |\partial u|\big)+ |\pi|r |\op^2 u|   \Big)|_{(\tau, \varphi(\tau;t,x))}.
		\end{equation}		
	\end{lemma}
	\begin{proof}
		From the hypothesis on $\pi$, $-2g^{L\Lb} = 1-2\pi^{L\Lb}>1/2$ and $\tilde{g}^{\alpha\beta} := \frac{g^{\alpha\beta}}{-2g^{L\Lb}}$ is well-defined. 
		From $\tilde{\Box}_gu=F$ one has that
		\[
		\Box_x u + \theta^{\alpha\beta}\partial_\alpha\partial_\beta u = \frac{F}{-2g^{L\Lb}}
		\]
		where $\theta^{\alpha\beta}:= \tilde{g}^{\alpha\beta} - m^{\alpha\beta}$ satisfies the following
		\[
		\begin{gathered}
			\theta_{L\Lb}=0, \quad \theta_{LT}=(-2g^{L\Lb})^{-1}\pi_{LT}, \quad 
			\overline{\text{tr}}\theta = (-2g^{L\Lb})^{-1}(\overline{\text{tr}}\theta + \pi_{L\Lb}), \quad |\theta|\lesssim |\pi|\\
			\Big|\theta^{\alpha\beta}\partial_\alpha\partial_\beta u - \frac{1}{r}\theta^{\Lb\Lb}\Lb^2 (ru)\Big|\lesssim |\pi|_{\q L \q T}|\op\partial u| + |\pi||\op^2 u| + |\pi|r^{-1}|\partial u|.
		\end{gathered}
		\]
		We recall that the flat wave operator can be written as follows
		\[
		\Box_x u= \frac{1}{r}(\partial_t +\partial_r)(\partial_t - \partial_r)(r u) + \Delta_{\S^1}u.
		\]
		Therefore
		\[
		\begin{aligned}
			\Big|\Big(\partial_t + \frac{1-\pi_{LL}/4}{1+\pi_{LL}/4}\partial_r\Big)(\partial_t - \partial_r)(ru) \Big| \lesssim r |\Delta_{\S^1} u| + |\pi|_{\q L \q T}\big( r |\op \partial  u|+ |\partial u|\big)+ |\pi|r |\op^2 u|  + |rF|.
		\end{aligned}
		\]
		Due to the smallness assumption on $\pi_{LL}$, any integral curve $(\tau,\varphi(\tau; t, x))$ passing through a point $(t,x)\in \q D_s$ must intersect the boundary $\partial_B \q D_s$. 
		The result of the lemma finally follows from integration along the characteristic curve, from which we get
		\[
		|(\partial_t - \partial_r)(ru)(t,x)| \lesssim  |(\partial_t - \partial_r)(ru)(t_0,x_0)| + \int_{t_0}^t |M[u, \pi](\tau)| + \tau |F|_{(\tau, \varphi(\tau;t,x))} d\tau
		\]
		where $(t_0, x_0)$ is the first point at which the intersection with the lateral boundary occurs.
	\end{proof}

	\begin{proof}[Proof of Proposition \ref{prp:good_coeffII}]
		Throughout this proof we will denote by $\eta$ any small positive constant that linearly depends on $\zeta_k, \gamma_k, \delta_k$.
		After the above observations, we only need to prove that  \eqref{good1} and \eqref{good2} are satisfied in the exterior of the cone $t=2r$. Such estimates are satisfied by the $(s/t)^2\partial_t$ and $\pb_\bma$ derivatives as a consequence of \eqref{KS1_hff}. Moreover, since
		\[
		\begin{aligned}
			\partial_t & = \frac{t-r}{t}\partial_t + \frac{x^\bma}{t+r}\pb_\bma + \frac{r}{t+r}(\partial_t-\partial_r), \qquad
			\partial_\bma & = \pb_\bma - \frac{x_\bmb}{t}\partial_t
		\end{aligned}
		\]
		we can reduce to proving them for the $\partial_t-\partial_r$ derivative, which we do by applying Lemma \ref{lem:characteristics}.
		We remark that for a point $(t,x)\in \q D_s\cap \dint$, the integral curve $(\tau, \varphi(\tau;t,x))$ may have a non-empty intersection with the exterior region, which explains why in the following we invoke some pointwise estimates obtained in section \ref{sec:exterior}.
		
		The integration of \eqref{h1TU_higher} along $\S^1$ shows that $Z^K\habf$ is solution to the following equation
		\begin{equation}\label{eq_ZKhTU}
			\Box_{x} Z^K h^{1, \flat}_{TU} + (H^{\bmmu\bmnu})^\flat \partial_\bmmu\partial_\bmnu   Z^K h^{1, \flat}_{TU} =  F^{K,\flat}_{TU} + F^{0,K}_{TU} - \big((H^{\mu\nu})^\natural\cdot \partial_\mu\partial_\nu  Z^K h^{1, \natural}_{TU}\big)^\flat
		\end{equation}
		where $F^{K, \flat}_{TU} = (F^K_{TU})^\flat$ and $F^K_{TU}$ is given by \eqref{FKTU}.
		We recall that tensor $H^{\mu\nu}$ decomposes as in \eqref{dec_H}. The hypothesis of Lemma \ref{lem:characteristics} are met thanks to the pointwise bounds \eqref{h0_estimate}, \eqref{est_hLT2} and \eqref{good_coefficient2}, therefore for all $s> s_0$ and all $(t,x)\in \q D_s$ 
		\begin{equation}\label{eq:zero_mode_good_coeff}
			\begin{aligned}
				|(\partial_t-\partial_r) Z^K h^{1, \flat}_{TU}(t,x)|   \le t^{-1}\sup_{\partial_B \q D_s} |(\partial_t -\partial_r)(r Z^K h^{1, \flat}_{TU})| + t^{-1}\int_2^t |M[Z^K\hff_{TU}, H^\flat](\tau)|d\tau \\
				+ t^{-1}\int_2^t \tau\,  (\text{RHS of } \eqref{eq_ZKhTU}) d\tau
			\end{aligned}
		\end{equation}
		where $M[\cdot, \cdot]$ is given by the formula in \eqref{M}. 
		
		From the interior pointwise bounds \eqref{KS1_hff}, \eqref{KS6_h} and the exterior pointwise bounds \eqref{KS1_ext}, \eqref{KS3_ext} we see that
		\[
		\sup_{\partial_B  \q D_s} |(\partial_t -\partial_r)(r Z^K h^{1, \flat}_{TU})| \lesssim C_1\ep t^{-\frac12+\eta}.
		\]
		
		As concerns the contribution coming from $M[Z^K\hff_{TU}, H^\flat]$, we see from formula \eqref{M} together with the fact that $\gd_\bmj = \frac{x^\bmi}{r^2}\Omega_{\bmi\bmj}$ and the exterior pointwise bounds \eqref{KS1_ext}-\eqref{KS3_ext} that for points $(t,x)\in \q D_s\cap \dext$
		\[
		|M[Z^K\hff_{TU}, H^\flat](t,x)| \lesssim C_0\ep t^{-1+2\sigma}\sqrt{l(t)} + C_0^2\ep^2 t^{-2+2\sigma}.
		\]	
		In $\q D_s\cap \dint$, we rewrite \eqref{M} using inequality \eqref{der_null_semihyp}
		\[
		\begin{aligned}
			|M[Z^K\hff_{TU}, H^\flat](t,x)| \lesssim |\pb Z^{\le 1}Z^K\hff_{TU}| + |H^\flat|_{\q L \q T}\Big(r\Big(\frac{s}{t}\Big)^2 |\partial^2 Z^K\hff_{TU}| + |\partial Z^{\le 1}Z^K\hff_{TU}|\Big) \\
			+ |H^\flat|  \Big(r\Big(\frac{s}{t}\Big)^4|\partial^2 Z^K\hff_{TU} | + \Big(\frac{s}{t}\Big)^2|\partial Z^{\le 1}Z^K\hff_{TU}| + |\pb Z^{\le 1}Z^K\hff_{TU}|\Big)
		\end{aligned}
		\]
		and hence deduce from pointwise bounds \eqref{h0_estimate}, \eqref{KS1_hff}, \eqref{KS1_hff}, \eqref{KS6_h}, \eqref{good_coefficient2} that 
		\[
		|M[Z^K\hff_{TU}, H^\flat](t,x)| \lesssim C_1\ep t^{-\frac32 + \eta}.
		\]
		Overall, $M[Z^K\hff_{TU}, H^\flat]|_{(\tau, \varphi(\tau;t,x))}$ is an integrable function of $\tau$.	 	
		
		We next show that the right hand side of \eqref{eq_ZKhTU} multiplied by $t$ is integrable in $\tau$ along the characteristic curve.
		Concerning the contributions to $F^{K,\flat}_{TU}$, see formula \eqref{FKTU}: the weak null terms do not appear in $F^{K,\flat}_{TU}$, hence from \eqref{null_ext_point}, \eqref{cub_ext_point} and \eqref{eq:est_null_point}, \eqref{eq:est_cub_point} we have that
		\[
		| Z^K F^\flat_{TU}(t,x)| \lesssim C_1^2\ep^2 t^{-\frac52+\eta} +C_0^2 \ep^2 t^{-2+2\sigma}\sqrt{l(t)}.
		\]
		
		The terms arising from the commutation of the null frame with the wave operator are estimated, on the one hand, using \eqref{h0_estimate},  pointwise interior bounds \eqref{KS1_hff}, \eqref{KS6_h} and the exterior bounds \eqref{KS2_ext}, \eqref{KS3_ext}
		\[
		\sum_{|K'|\le |K|}C^{\bmi\alpha\beta}_{TU,K'}\big| \gd_\bmi  Z^{K'} \habf\big| + \big|D^{\alpha\beta}_{TU,K'}  Z^{K'}\habf \big| \lesssim C_1\ep t^{-\frac52 +\eta} + C_0\ep t^{-2+\sigma}\sqrt{l(t)}.
		\]
		On the other hand, using additionally the a-priori bound \eqref{bootW_in} we see that
		\[
		\begin{aligned}
			\sum_{|K_1|+|K_2|\le |K|}  \big|E^{\bmi\alpha\beta}_{TU\mu\nu, K_1K_2} \big(Z^{K_1} H^{\mu\nu}\cdot \gd_\bmi  Z^{K_2}h^1_{\alpha\beta}\big)^\flat\big| + \big| F^{\alpha\beta}_{TU\mu\nu, K_1K_2}  \big(Z^{K_1}H^{\mu\nu} \cdot Z^{K_2}h^1_{\alpha\beta}\big)^\flat \big| \\
			\lesssim C_1C_2\ep^2 t^{-\frac52 + \eta} + C_0^2\ep^2 t^{-2+2\sigma}\sqrt{l(t)} .
		\end{aligned}
		\]
		From \eqref{comm_int_point_zeromode}, together with \eqref{KS1_ext} and \eqref{KS3_ext}, we have that
		\[
		\big| \big( (H^{1, \mu\nu})^\natural\cdot \partial_\mu\partial_\nu Z^K\hnn_{TU}\big)^\flat(t,x)\big| + \big| \big([Z^K, H^{1,\mu\nu}\partial_\mu\partial_\nu]\hnn_{TU}\big)^\flat(t,x)\big|\lesssim C_1^2\ep^2 t^{-\frac52 + \eta}.
		\]
		Thus all together
		\begin{multline*}
			t |F^{K, \flat}_{TU} - [Z^K, H^{0, \bmmu\bmnu}\partial_\bmmu\partial_\bmnu]\hff_{TU}| + t \big| \big( (H^{1, \mu\nu})^\natural\cdot \partial_\mu\partial_\nu Z^K\hnn_{TU}\big)^\flat(t,x)\big| \\
			\lesssim C_1\ep t^{-\frac32+\eta} + C_0\ep t^{-1+2\sigma}\sqrt{l(t)}.
		\end{multline*}
		
		Using the structure highlighted in Lemma \ref{lem:box_h0}, one can easily show that
		\[
		|F^{0, K}_{TU}(t,x)|\lesssim \ep t^{-3}.
		\]	
		
		Finally, if $K$ is such that $Z^K = \partial^{I}$ is a product of derivatives only, with $|I|\le N_1$, we derive from \eqref{h0_estimate}, \eqref{KS1_ext} and \eqref{KS1_hff} that
		\[
		t\big|[ \partial^I, (H^{0,\bmmu\bmnu}) \partial_\bmmu\partial_\bmnu]\hff_{TU}\big|\lesssim \sum_{\substack{|I_1|+|I_2|=|I|\\ |I_1|\ge 1}} t|\partial^{I_1}H^0|\, |\partial^2 \partial^{I_2}\hff_{TU}|\lesssim C_1\ep^2 t^{-2+\delta_{k+2}}.
		\]
		This is an integrable quantity, as all others above, therefore we obtain \eqref{good2}. 
		
		Bound \eqref{good1} follows then by induction on the number $k$ of vector fields in $Z^K$, since 
	\begin{align*}
	 |[Z^K, (H^{0,\bmmu\bmnu}) \partial_\bmmu\partial_\bmnu]\hff_{TU}|& \le \sum_{\substack{|I_1|+|K'_2|\le |K|\\|I_1|\ge 1,\, |K'_2|<|K|}}\hspace{-5pt} |\partial^{I_1 }H^0| |\partial^2 Z^{K'_2}\hff_{TU}| + \sum_{\substack{|K_1|+|K''_2|\le |K|\\ |K''_2|<|K|}}\hspace{-5pt} |Z^{K_1}H^0| |\partial^2 Z^{K''_2}\hff_{TU}| \\
	 &\le\sum_{\substack{|K'|<|K|}} \frac{1}{(1+t+r)^2}|\partial^2 Z^{K'}\hff_{TU}| + \sum_{\substack{|K''|<|K|}} \frac{1}{(1+t+r)}|\partial^2 Z^{K''}\hff_{TU}| 
		\end{align*}
		where $K'$ is a multi-index of type $(|K|-1, k)$ and $K''$ is of type $(|K|-1, k-1)$.
		
	\end{proof}
	
	\subsection{The weak null terms}\label{subsec:weak_null}
	
	In this section we prove estimates on the weak null terms. 
	\begin{lemma}\label{lem:weak_null_int}
		For any multi-index $K$ of type $(N,k)$ and $i=0,1$, we have that
		\begin{equation}
			 \left\|\partial^i Z^K P_{\alpha\beta}(\partial h, \partial h)\right\|_{L^2_{xy}(\hin_s)}\lesssim C_1^2\epsilon^2 s^{-\frac12 + \zeta_k},\label{weak_est1}  
		\end{equation}
		while for any multi-index $K$ of type $(N, k)$ with $k\le N_1$
		\begin{align}
			&\left\| Z^K P_{\alpha\beta}(\partial h, \partial h)\right\|_{L^2_{xy}(\hin_s)}\lesssim C_1^2\epsilon^2 s^{-1 + \delta_k} \label{weak_est3}.
		\end{align}
		Moreover, for $K$ a multi-index of type $(N, k)$ we have
		\begin{multline}\label{weak_est4}
			\left\| Z^K P^\flat_{\alpha\beta}(\partial h, \partial h)\right\|_{L^2_{xy}(\hin_s)}\\ 
			\lesssim C_1\epsilon s^{-1}\sum_{K'} \enint(s,Z^{K'}\hff)^{1/2}+ C_1\ep s^{-1+C\ep}\sum_{K''} \enint(s,Z^{K''}\hff)^{1/2}+C_1^2\ep^2s^{-\frac32+2\delta_N} \\
			+C_1^2\ep^2\delta_{k>N_1}\left(\sum_{i=1}^4 s^{-1+\gamma_i+\zeta_{k-i}} + s^{-1+\zeta_k}\right)
		\end{multline}	
		where $K'$ is of type $(|K|, k)$, $K''$ of type $(|K|-1, k-1)$, and where $\delta_{k>N_1}=1$ when $k> N_1$, 0 otherwise.
	\end{lemma}
	\begin{proof}
		We start by decomposing each occurrence of $h$ into the sum $h^0+h^1$ and observe that the quadratic interactions involving at least one factor $h^0$ verify \eqref{est_quad_h0}. We hence focus on estimating the weak null terms only involving factors $h^1$ and distinguish between the region inside the cone $t=2r$ and its complement in $\dint$.
		
		In the interior of the cone $t = 2r$ (where $s\approx t$), the bounds of the statement do not depend on the weak null structure: consequently they are the same as the bounds for the null terms \eqref{null_int}, \eqref{null_lower_int} and \eqref{null_zeromode}. 

		The estimates in the region $\dint \cap \{t< 2r\}$ follow from the particular structure of the weak null terms and the wave condition. We have already seen that these two yield inequality \eqref{Z^Kweak} which, after \eqref{der_null_semihyp}, can be also written in the following form
		\begin{equation}\label{weak_null_ineq}
			\begin{aligned}
				\big|\partial^i Z^{K}P_{\alpha\beta}(\partial h^1, \partial h^1)\big|\lesssim \sum_{\substack{|K_1|+|K_2| \le |K|\\ |I_1|+|I_2|=i}} |\partial\partial^{I_1}Z^{K_1}h^1|_{\q T \q U} |\partial\partial^{I_2}Z^{K_2}h^1|_{\q T \q U} \\
				+ \sum_{\substack{|K_1|+|K_2| \le |K|\\ |I_1|+|I_2|=i}} \Big(\frac{s}{t}\Big)^2 |\partial\partial^{I_1}Z^{K_1}h^1||\partial\partial^{I_2}Z^{K_2}h^1| + |\pb\partial^{I_1}Z^{K_1}h^1||\partial\partial^{I_2}Z^{K_2}h^1| \\
				+ \sum_{\substack{|K_1|+|K_2| \le |K|\\ |I_1|+|I_2|=i}} \hspace{-20pt} r^{-1}|\partial^{I_1}Z^{K_1}h^1||\partial\partial^{I_2}Z^{K_2}h^1| + \hspace{-20pt} \sum_{\substack{|K_1|+|K_2|+|K_3|\le |K|\\ |I_1|+|I_2|+|I_3| = i}}\hspace{-20pt} |\partial^{I_1}Z^{K_1}h^1|  |\partial \partial^{I_2}Z^{K_2}h^1|  |\partial\partial^{I_3}Z^{K_3}h^1| \\
				+ \frac{M\chi_0(t/2\le r\le 3t/4)}{(1+t+r)^2}|\partial \partial^i Z^{\le K}h^1|.
			\end{aligned}
		\end{equation}
		We observe that the quadratic terms on the second line of \eqref{weak_null_ineq} are null and hence satisfy \eqref{null_int}, \eqref{null_lower_int} and \eqref{null_zeromode}; the ones on the third line are cubic (or cubic-like) and satisfy \eqref{cub_int}, \eqref{cub_lower_int}; moreover
		\[
		\left\|\frac{M\chi_0(t/2\le r\le 3t/4)}{(1+t+r)^2}\, \partial \partial^i Z^{\le K}h^1 \right\|_{L^2}\lesssim \ep s^{-2}\left\|(s/t)\partial \partial^i Z^{\le K}h^1\right\|_{L^2}\lesssim C_1\ep^2 s^{-\frac32 + \delta_k}.
		\]
		Finally, from the enhanced bounds \eqref{good1}, \eqref{good2} satisfied by the $h^1_{TU}$ coefficients we derive 
		\[
		\begin{aligned}
			\sum_{\substack{|K_1|+|K_2|\le |K|\\ |I_1|+|I_2|=i}}\left\|\partial \partial^{I_1}Z^{K_1}h^1_{TU}\cdot \partial \partial^{I_2}Z^{K_2}h^1_{TU}\right\|_{L^2}\lesssim \sum_{\substack{|I_1|+|K_2|\le |K|+i\\ |I_1|\le \floor{(|K|+i)/2}}}\|(t/s)\partial \partial^{I_1}h^1_{TU}\|_{L^\infty}\|(s/t)\partial Z^{K_2}h^1_{TU}\|_{L^2}\\
			+ \sum_{\substack{|I_1|+ |J_1|+|K_2|\le |K|+i\\ |I_1| + |J_1|\le \floor{(|K|+i)/2}\\ |J_1|\ge 1 }}\|(t/s)\partial \partial^{I_1}\Gamma^{J_1} h^1_{TU}\|_{L^\infty}\|(s/t)\partial Z^{K_2}h^1_{TU}\|_{L^2}\\
			\lesssim C_1 \ep s^{-1}\sum_{j=0}^i\|(s/t)\partial \partial^j Z^{\le K}h^1_{TU}\|_{L^2} + C_1\ep s^{-1+2C\ep}\sum_{j=0}^i\|(s/t)\partial \partial^j Z^{\le K-1}h^1_{TU}\|_{L^2}.
		\end{aligned}
		\]
		

	\end{proof}

	\subsection{The commutator terms} \label{sub:commutators_int}
	
	The goal of this section is to get suitable estimates for the trilinear terms involving commutators in the right hand side of energy inequalities \eqref{en_ineq_hab_high} and \eqref{en_ineq_habf_high}. More precisely, we will prove the following.
	
	\begin{proposition}\label{prp:commutators}
		There exist some small positive parameters $\delta_k$ with $\sigma\ll \delta_k\ll\delta_{k+1}\ll \kappa$ and $\delta_k\ll \kappa-\sigma$ for $k=0, \dots, N$ such that, under the assumptions of Proposition \ref{prop:bootstrap_int} we have the following inequalities for all $s\in [s_0, S_0)$: 
		
		\noindent for multi-indexes $K$ of type $(N+1, k)$ with $k\le N$ 
		\begin{equation}
			\iint_{\hin_{[s_0, s]}}\hspace{-10pt} \big|[Z^K, H^{1, \mu\nu}\partial_\mu \partial_\nu]\hab\big| |\partial_t Z^K\hab| dxdydt \lesssim (C_0^2+C_1^2)C_2^2\ep^3 {\Big(\sum_{i=1}^4 s^{1+\gamma_i+\zeta_{k-i}+\zeta_k} + s^{1+2\zeta_k}\Big)},\label{comm_high} 
		\end{equation}
		for multi-indexes $K$ of type $(N, k)$
		\begin{multline}
			\iint_{\hin_{[s_0, s]}} \big|\big([Z^K, H^{1, \mu\nu}\partial_\mu\partial_\nu]\hab\big)^\flat \big| |\partial_t Z^K\habf| dx dt \\ 
			+ \iint_{\hin_{[s_0, s]}} \big|\big((H^{1, \mu\nu})^\natural \cdot\partial_\mu\partial_\nu Z^K\habn\big)^\flat \big| |\partial_t Z^K\habf| dx dt \lesssim  (C_0^2+C_1^2)C_2^2\ep^3 {\Big(\sum_{i=1}^4 s^{\gamma_i+\zeta_{k-i}+\zeta_k} + s^{2\zeta_k}\Big)} \label{comm_zeromode}
		\end{multline}
		for multi-indexes $K$ of type $(N, k)$ with $k\le N_1$
		\begin{equation}
			\iint_{\hin_{[s_0, s]}} \big|[Z^K, H^{1, \mu\nu}\partial_\mu\partial_\nu]\hab\big| |\partial_t Z^K\hab| dxdydt \lesssim  (C_0^2+C_1^2)C_2^2\ep^3 s^{2\delta_k} \label{comm_low}
		\end{equation}
		and for multi-indexes $K$ of type $(N-1, k)$ with $k\le N_1$
		\begin{multline}
			\iint_{\hin_{[s_0, s]}} \big|\big([Z^K, H^{1, \mu\nu}\partial_\mu \partial_\nu]\hab\big)^\flat \big| |\partial_t Z^K\habf| dx dt \\
			+ \iint_{\hin_{[s_0, s]}} \big|\big((H^{1, \mu\nu})^\natural \cdot\partial_\mu\partial_\nu Z^K\habn\big)^\flat \big| |\partial_t Z^K\habf| dx dt\lesssim (C_0^2+C_1^2)C_2^2\ep^3. \label{comm_zeromode_low}
		\end{multline}
	\end{proposition}
	We postpone the proof of the above proposition and first observe that, because of \eqref{dec_product},
	we will need to estimate quadratic terms (in fact, commutators) that are either pure products of zero modes, or pure products of non-zero modes, or mixed products. We proceed to the analysis of those separately, in the lemmas that follow.

	\begin{lemma} \label{lem:low-order-wave-commutator}
		There exists $0< \eta\le 2\delta_N\ll 1$, linearly depending on $\zeta_k, \gamma_k, \delta_k$, such that for any multi-index $K$ of type $(N, k)$ we have
		\begin{equation}\label{comm_est_int1}
			\begin{aligned}
				\big\|[\partial^i Z^K, (H^{1, \bmmu\bmnu})^\flat\partial_\bmmu \partial_\bmnu] \hff_{\alpha\beta} -\hspace{-10pt} \sum_{\substack{|K_1|+|K_2|\le |K|\\ |K_2|\le \floor{|K|/2}}}\hspace{-10pt} Z^{K_1}H^{1, \flat}_{LL}\cdot\partial^2_t \partial^i Z^{K_2}\hff_{\alpha\beta} \big\|_{L^2_{x}(\mathcal{H}_s)} 	\\
				\lesssim C_1C_2\ep^2	
				\begin{cases}
					s^{-1+\eta}  &\text{if } i=1 \text{ and } |K|=N \\
					s^{-\frac32+\eta}   &\text{if } i=0
			\end{cases}				\end{aligned}
		\end{equation}
	\end{lemma}
	\begin{proof}
		The proof is based on inequality \eqref{comm_int} applied to $\pi = H^{1, \flat}$ and $\phi =\hff_{\alpha\beta}$.
		We will focus on discussing only the following terms, the remaining ones being simpler:
		\[
		\sum_{\substack{|K_1|+|K_2|\le |K| \\|K_1|\le \floor{|K|/2}, |K_2|<|K|}}\big\|  Z^{K_1}H^{1, \flat}_{LL}\cdot\partial^2_t \partial^i Z^{K_2}\habf\big\|_{L^2(\hin_s)} 
		\]		
		and	
		\[
		\sum_{\substack{|K_1|+|K_2|\le |K|}}\big\| \partial Z^{K_1}H^{1, \flat}_{LL}\cdot\partial^2_t Z^{K_2}\habf\big\|_{L^2(\hin_s)} .
		\]		
We remark that the latter sum only appears in the case $i=1$.
		
		Concerning the first sum with $i=0$,  pointwise bound \eqref{good_coefficient2} and $L^2$ bound \eqref{second_derivatives_L2} yield
		\[
		\begin{aligned}
			\sum_{\substack{|K_1|+|K_2|\le |K|\\ |K_1|\le \floor{|K|/2}, \, |K_2|<|K|}}\hspace{-10pt} \big\|Z^{K_1}H^{1, \flat}_{LL}\cdot\partial^2_t Z^{K_2}\hff_{\alpha\beta}\big\|_{L^2(\hin_s)}
			\lesssim C_1\ep s^{-\frac12 +\eta}\sum_{|K'|< |K|} \|(s/t)^2\partial^2_t Z^{K'}\habf\big\|_{L^2}\\ \lesssim C_1^2\ep^2 s^{-\frac32+\eta}.
		\end{aligned}
		\]
		
		When $i=1$, the $L^2$ bound \eqref{third_derivative_L2} gives
		\[
		\sum_{\substack{|K_1|+|K_2|\le |K|\\ |K_1|\le \floor{|K|/2}, \, |K_2|<|K|}}\hspace{-10pt} \big\| Z^{K_1}H^{1, \flat}_{LL}\cdot\partial^2_t\partial  Z^{K_2}\hff_{\alpha\beta}\big\|_{L^2(\hin_s)} \lesssim C_1C_2\ep^2 s^{-1 + \eta}.			
		\]
		
		The second sum is estimated using \eqref{bootin1} and \eqref{good_coeff1_flat} whenever $|K_1|\le N-2$, and using \eqref{KS1_hff} and \eqref{good_coeff_L2} otherwise
		\[
		\sum_{\substack{|K_1|+|K_2|\le |K|}}\hspace{-10pt} \big\|\partial Z^{K_1}H^{1, \flat}_{LL}\cdot\partial^2_t Z^{K_2}\hff_{\alpha\beta}\big\|_{L^2(\hin_s)}\lesssim C_1^2\ep s^{-1+\eta}	.
		\]	
	
	\end{proof}
	
	\begin{lemma}
		For any fixed fixed multi-index $K$ and any smooth function $\phi$, we have
		\begin{equation}\label{comm_est_int2}
			\begin{aligned}
				\sum_{\substack{|K_1|+|K_2|\le |K|\\ |K_2|\le \floor{|K|/2}}}\iint_{\hin_{[s_0, s]}} |Z^{K_1} H^{1, \flat}_{LL}| |\partial^2_t Z^{K_2}\hff_{\alpha\beta}| |\partial_t \phi|dxdt \lesssim (C_0^2+C_1^2)C_2^2\ep^3 s^{\kappa_\phi}
			\end{aligned}
		\end{equation}
		with
		\begin{equation} \label{kappa_phi}
			\kappa_\phi = 
			\begin{cases}
				1, \quad &\text{if } \phi = Z^K\hab \text{ and K is of type } (N+1, k),\ k\le N \\[2pt]
				0, \quad &\text{if } \phi = Z^K\habf \text{ and K is of type } (N, k)\\[2pt]
				0, \quad &\text{if } \phi = Z^K\hab \text{ and K is of type } (N, k), \ k\le N_1
			\end{cases}
		\end{equation}
	\end{lemma}
	\begin{proof}
		We restrict our attention to the case where $Z^{K_1} = \Gamma^{K_1}$, as the ones where $Z^{K_1} = \partial^{I_1}\Gamma^{J_1}$ with $|I_1|\ge 1$ can be estimated as in the proof of lemma \ref{lem:low-order-wave-commutator}.
		
For any $K_1$, $K_2$ in the selected range of indexes and any fixed $\nu$ such that $2\delta_k<\nu\ll 1$, we write the following
		\[
		\begin{aligned}
			\iint_{\hin_{[s_0, s]}}\hspace{-10pt} |Z^{K_1} H^{1, \flat}_{LL}| |\partial^2_t Z^{K_2}\hff_{\alpha\beta}| |\partial_t \phi|dxdt \le \int_{s_0}^s \hspace{-5pt}\enint(\tau, \phi)^{\frac12}\big\|Z^{K_1} H^{1, \flat}_{LL} \cdot \partial^2_t Z^{K_2}\hff_{\alpha\beta}\big\|_{L^2(\hin_\tau)}d\tau\\
			\lesssim \int_{s_0}^s \ep \tau^{-1-\nu}\enint(\tau, \phi)d\tau + 
			\frac{1}{\ep}\int_{s_0}^s\int_{\hin_\tau}\tau^{1+\nu}|Z^{K_1} H^{1, \flat}_{LL}|^2 |\partial^2_t Z^{K_2}\hff_{\alpha\beta}|^2 dxd\tau \\
			\lesssim C_1^2\ep^3 s^{\kappa_\phi} + \frac{1}{\ep}\int_{s_0}^{t_s}\int_{\q C_t}|Z^{K_1} H^{1, \flat}_{LL}|^2 |\partial^2_t Z^{K_2}\hff_{\alpha\beta}|^2 t^{1+\nu} dxdt
		\end{aligned}
		\]
		where $\q C_t = \{x\in \R^3 : \sqrt{(t^2-s^2)^+}\le |x|\le \sqrt{(t-1)^2-1}\}$ and $t_s = s^2/2$. The latter inequality is obtained by injecting the energy assumptions \eqref{boot1_in}, \eqref{boot2_in}, \eqref{boot3_in} in the first integral on the second line and by performing a change of coordinates in the second one.
		
		We use \eqref{second_derivatives2} and the Hardy inequality of corollary \ref{cor:Hardy} with $\mu=1-\eta$ and $\alpha = 1-\eta-\nu$ to estimate the above integral. For any fixed $\mu'>0$, we get that
		\[
		\begin{aligned}
			\frac{1}{\ep}\int_{s_0}^{t_s}\int_{\q C_t}|Z^{K_1} H^{1, \flat}_{LL}|^2 |\partial^2_t Z^{K_2}\hff_{\alpha\beta}|^2 t^{1+\nu} dxdt \\
			\lesssim C_2^2 \ep \int_{s_0}^{t_s}\int_{\q C_t} \frac{|Z^{K_1}H^{1, \flat}_{LL}|^2}{(1+t-r)^{2+(1-\eta)}} \frac{dxdt}{(1+t+r)^{1-\eta-\nu}} \\
			\lesssim  C_2^2\ep\int_{s_0}^{t_s}\int_{\hin_\tau}\frac{|\partial_r Z^{K_1}H^{1, \flat}_{LL}|^2}{\tau^{2(1-\eta-\nu)}}dxd\tau + C_2^2\ep\int_{s_0}^{t_s}\int_{\Sext_t}|\partial_rZ^{K_1}H^{1, \flat}_{LL}|^2\frac{(1+|r-t|)^{1+\mu'}}{(1+t+r)^{1-\eta-\nu}}dxdt.
		\end{aligned}
		\]
		
		We estimate the first integral using \eqref{good_coeff_L2}:
		\[
		C_2^2\ep\int_{s_0}^{t_s}\int_{\hin_\tau}\frac{|\partial_r Z^{K_1}H^{1, \flat}_{LL}|^2}{\tau^{2(1-\eta-\nu)}}dxd\tau\lesssim \int_{s_0}^{t_s}C_1^2C_2^2\ep^3 \tau^{-2(1-\eta-\nu) +4\delta_{k_1}}d\tau\lesssim C_1^2C_2^2\ep^3								.
		\]
For the latter one, we pick $2\delta_k\le \nu\ll \kappa$ and $\mu':=2\kappa-\eta-2\nu$ so that $\mu'>0$. From inequality \eqref{ineq_H1lt_higher} we get \small{
			\[
			\begin{aligned}
				C_2^2\epsilon\int_{s_0}^{t_s}\int_{\Sext_t}|\partial_rZ^{K_1}H^{1, \flat}_{LL}|^2\frac{(1+|r-t|)^{1+\mu'}}{(1+t+r)^{1-\eta-\nu}}dxdt\\
				\lesssim C_2^2\ep\hspace{-10pt}\sum_{|K'|\le |K_1|} \int_{s_0}^{t_s}\hspace{-5pt}\int_{\Sext_t}t^{-\nu}(2+r-t)^{2\kappa}|\op Z^{K'} H^{1, \flat}|^2 dxdt  + \int_{s_0}^{t_s}\hspace{-5pt}\int_{\Sext_t}t^{-\nu} r^{-1}(2+r-t)^{2\kappa-1}|Z^{K'}H^{1, \flat}|^2 dxdt \\
				+ \sum_{|K'_1|+|K'_2|\le |K_1|}C_2^2\ep\int_{s_0}^{t_s}\int_{\Sext_t}t^{-\nu} (2+r-t)^{2\kappa} |\big(Z^{K'_1}H^1\cdot \partial Z^{K'_2}H^1\big)^\flat|^2 dxdt \\
				+ C_2^2\ep\int_{s_0}^{t_s}\int_{\Sext_t}(2+r-t)^{2\kappa}M^2\chi_0^2(1/2\le r/t\le 3/4)r^{-4-\nu}dxdt.
			\end{aligned}
			\]}
		\normalsize
		The above right hand side is bounded by $C_2^2C_0^2\ep^3$. For the first integral, this follows using \eqref{boot2_ex} and the fact that $\nu>2\delta_k>\sigma$; for the second one, it follows from the weighted Hardy inequality \eqref{hardy_ext} and \eqref{boot1_ex}; for the third one, from the pointwise bounds \eqref{KS1_ext}, \eqref{KS3_ext}, the weighed Hardy inequality \eqref{hardy_ext} and the energy bounds \eqref{boot1_ex}, \eqref{boot2_ex}; for the last one, from the fact that the domain of integration of the latter one is uniformly bounded in $(t,x)$.
	\end{proof}

	We now estimate the different contributions to the commutator $[Z^K, (H^{1, \mu\nu})^\flat \partial_\mu\partial_\nu]\habn$, which appear in the equation satisfied by $Z^K\hab$. 
	
	\begin{lemma}
		For any fixed multi-index $K$ of type $(N+1, k)$ with $k\le N$, one has
		\begin{equation}\label{comm_WKGh}
		\|[Z^K, (H^{1, \mu\nu})^\flat \partial_\mu\partial_\nu]\habn\|_{L^2(\hin_s)} \lesssim  C_1C_2\ep^2 \left(\sum_{i=1}^4 s^{-\frac12+\gamma_i+\zeta_{k-i}} + s^{-\frac12+\zeta_k}\right)
		\end{equation}
		and  for $K$ of type $(N, k)$ with $k\le N_1$ 
\begin{equation}\label{comm_WKGl}
\|[Z^K, (H^{1, \mu\nu})^\flat \partial_\mu\partial_\nu]\habn\|_{L^2(\hin_s)} \lesssim C_1C_2\ep^2s^{-1+\delta_k}.
\end{equation}		
	\end{lemma}	

	\begin{proof}
The terms in the commutator are of the form
		$$Z^{K_1} {H}^{1, \flat}\,   \partial^2 Z^{K_2}\habn$$
		with each $K_j$ of type $(|K_j|, k_j)$ and such that $|K_1|+|K_2|\leq |K|$, $|K_2|<|K|$. We focus on the case where $|K_1|=k_1$, that is where $Z^{K_1}=\Gamma^{K_1}$, the remaining ones being simpler, and also recall that $N=N_1+5$.
		
In the case where $|K_1|\le N_1$, we estimate $\Gamma^{K_1} {H}^{1, \flat} $ in $L^\infty$ with \eqref{bootW_in}. Together with the energy bounds \eqref{bootin3} and \eqref{bootin5}, and the relation \eqref{condition_parameters}, this yields
\[
\aligned
\|\Gamma^{K_1} {H}^{1, \flat}  \, \partial^2 Z^{K_2}\habn\|_{L^2(\hin_s)}
	\lesssim C_1C_2\ep^2
	\begin{cases}s^{-1+\delta_k} , \ &\text{if } |k_2|\le N_1 \\
 \sum_{i=0}^4 s^{-\frac12+\gamma_i+\zeta_{k-i}}, \ & \text{if } N_1<|k_2|\le N
	\end{cases}
	\endaligned
\]	

		In the case where $N_1 <|K_1|\leq N$, we estimate $\partial^2 Z^{K_2}\habn$ in $L^\infty$
		with \eqref{bootKG_in} and $\Gamma^{K_1}H^{1, \flat}$ in $L^2$ using \eqref{bootin_hardy3}
		$$
		\aligned
		\|\Gamma^{K_1} {H}^{1, \flat}\,  \partial^2 Z^{K_2}\habn\|_{L^2(\hin_s)}
			&\lesssim  \ep^2C_1C_2\left(s^{-\frac{1}{2}+\zeta_k} +\sum_{i=1}^4 s^{-\frac{1}{2} +\gamma_{i}+\zeta_{k-i}}\right)
			\endaligned
		$$

The conclusion of the proof follows from relation \eqref{condition_parameters} and the observation that, if $K$ is a multi-index of type $(N, k)$ with $k\le N_1$ and $|K_1|>N_1$, then $\Gamma^{K_1}$ contains at least one usual derivative.
\end{proof}

	We conclude this subsection with some estimates on commutator terms involving $H^{1,\natural}$.
	
	\begin{lemma}
We have that
		\begin{multline}
		\label{KG-KG in wave}
			\big\| Z^K\big( (H^{1, \mu\nu})^\natural\cdot\partial_\mu\partial_\nu\hnn_{\alpha\beta}\big)^\flat \big\|_{L^2(\hin_s)} \\
			\lesssim C_1C_2\ep^2
			\begin{cases} 
	\sum_{i=1}^4 s^{-1+\gamma_i+\zeta_{k-i}} + s^{-1+\zeta_k}\ &\text{if } K \text{ of type } (N, k) \\
				s^{-\frac32 +\delta_k} \ &\text{if } K \text{ of type } (N-1, k) \text{ with } k\le N_1
			\end{cases}
		\end{multline}
		and 
		\begin{multline}\label{KG-KG in general}
			\big\|\big[Z^K, (H^{1, \mu\nu})^\natural \partial_\mu \partial_\nu\big]\habn \big\|_{L^2_{xy}(\hin_s)} \\
	\qquad	
		\\	\lesssim C_1C_2 \ep^2 
			\begin{cases}
			\sum_{i=1}^4 s^{-1+\gamma_i+\zeta_{k-i}} + s^{-1+\zeta_k}\ &\text{if } K \text{ of type } (N+1, k) \text{ with } k\le N \\
				s^{-\frac32+\delta_k} \ &\text{if } K \text{ of type } (N, k) \text{ with } k\le N_1
			\end{cases}
		\end{multline}
	\end{lemma}
	\begin{proof}
		We observe that if $K$ is a multi-index of type $(N-1, k)$ then $\partial Z^K = Z^{K'}$ with $K'$ of type $(N, k)$.
		The first bound (resp. the second) in \eqref{KG-KG in wave} simply follows from \eqref{condition_parameters}, \eqref{bootKG_in}, \eqref{bootin3} (resp. \eqref{bootin4}) and the fact that $\floor{N/2}+2\le N_1$. Similarly, the first bound (resp. the second) in \eqref{KG-KG in general} follows from \eqref{bootin1} (resp. \eqref{bootin4}) and \eqref{KS1_hnn}. 
	\end{proof}
	
	\begin{lemma}
		We have that
		\begin{equation} \label{comm_KGw}
			\big\|\big[Z^K, (H^{1, \bmmu\bmnu})^\natural \partial_\bmmu \partial_\bmnu\big]\habf \big\|_{L^2_{xy}(\hin_s)} \lesssim C_1^2 \ep^2 
			\begin{cases}
				s^{-1+2\delta_N} \quad  &\text{if } K \text{ of type } (N+1, k) \text{ with } k\le N \\
				s^{-\frac32+2\delta_N} \quad  &\text{if } K \text{ of type } (N, k) \text{ with } k\le N_1.
			\end{cases}
		\end{equation}
	\end{lemma}
	\begin{proof}
		We begin by writing that
		\[
		\big|\big[Z^K, (H^{1, \bmmu\bmnu})^\natural \partial_\bmmu \partial_\bmnu\big]\habf \big|\lesssim \sum_{\substack{|K_1|+|K_2|\le |K|\\ |K_2|<|K|}} |Z^{K_1}H^{1,\natural}| |\partial^2 Z^{K_2}\habf|.
		\]
		For any multi-index $K$ of type $(N+1, k)$ with $k\le N$, energy bound \eqref{bootin2} and pointwise bound \eqref{bootinL2y.1} yield
		\begin{multline*}
			\sum_{\substack{|K_1|+|K_2|\le |K|\\ |K_1|\le\floor{|K|/2}, |K_2|<|K|}} \big\|Z^{K_1}H^{1,\natural}\cdot \partial^2 Z^{K_2}\habf\big\|_{L^2_{xy}} \\
			\lesssim \sum_{\substack{|K_1|+|K_2|\le |K|\\ |K_1|\le\floor{|K|/2}, |K_2|<|K|}} \big\|(t/s)Z^{K_1}H^{1, \natural}\big\|_{L^\infty_xL^2_y}\big\|(s/t)\partial^2 Z^{K_2}\habf\big\|_{L^2_x}\lesssim C_1^2\ep^2 s^{-\frac32+2\delta_N}
		\end{multline*}
		while energy bound \eqref{bootin1} and pointwise bound \eqref{KS4_h}  yield
		\begin{multline*}
			\sum_{\substack{|K_1|+|K_2|\le |K|\\ |K_2|\le\floor{|K|/2}}} \big\|Z^{K_1}H^{1,\natural}\cdot \partial^2 Z^{K_2}\habf\big\|_{L^2_{xy}}\\
			\lesssim \sum_{\substack{|K_1|+|K_2|\le |K|\\ |K_2|\le\floor{|K|/2}}} \big\|\partial^2 Z^{K_2}\habf\big\|_{L^\infty_x}\big\|Z^{K_1}H^{1, \natural}\big\|_{L^2_{xy}} \lesssim C_1^2\ep^2 s^{-1+2\delta_N}.
		\end{multline*}
		If $K$ is instead a multi-index of type $(N, k)$ with $k\le N_1$, the above estimate can be improved to $C_1^2\ep^2 s^{-\frac32+2\delta_N}$ using energy bound \eqref{bootin4}.
		
	\end{proof}

	\begin{proof}[Proof of Proposition \ref{prp:commutators}]
		We decompose $H^{1,\mu\nu}$ and $\hab$ appearing in the commutator terms into their zero mode and their zero-average component and express all commutators involving $(H^{1,\mu\nu})^\flat$ with respect to the null framework. 
		
		From \eqref{comm_est_int1}, \eqref{comm_est_int2} and the energy bounds \eqref{bootin1}, \eqref{bootin2} and \eqref{bootin5} we derive that
		\[
		\iint_{\hin_{[s_0, s]}} |[Z^K, (H^{1,\bmmu\bmnu})^\flat \partial_\bmmu\partial_\bmnu]\habf| |\partial_t \phi| dxdydt \lesssim (C_1^3 + C_0^2)\ep^3 s^{\kappa_\phi}
		\]
		with $\kappa_\phi$ given by \eqref{kappa_phi}, while from \eqref{KG-KG in wave} and \eqref{bootin2}
		\[
		\iint_{\hin_{[s_0, s]}} |Z^K\big( (H^{1, \mu\nu})^\natural\cdot\partial_\mu\partial_\nu\hnn_{\alpha\beta}\big)^\flat| |\partial_t Z^K\habf| dx dt \lesssim C_1^3\ep^2 \Big(\sum_{i=0}^4 s^{\mu(\gamma_i+\zeta_{k-i}+\zeta_k)} + s^{2\mu\zeta_k}\Big)
		\]
		where $\mu=0$ if $K$ is of type $(N-1, k)$ with $k\le N_1$, $\mu=1$ otherwise. These two estimates together imply \eqref{comm_zeromode} and \eqref{comm_zeromode_low}.
		
		Additionally, for $K$ of type $(N+1, k)$ with $k\le N$ we get from \eqref{comm_WKGh} and \eqref{bootin1} that 
		\[
		\iint_{\hin_{[s_0, s]}}  |[Z^K, (H^{1,\mu\nu})^\flat \partial_\mu\partial_\nu]\habn| |\partial_t Z^K\hab| dxdydt \lesssim (C_1^3+C_0^2)\ep^3 s\Big(\sum_{i=0}^4 s^{(\gamma_i+\zeta_{k-i}+\zeta_k)} + s^{2\zeta_k}\Big)
		\]
		while for multi-indexes of type $(N, k)$ with $k\le N_1$, estimates \eqref{comm_WKGl} and \eqref{bootin5} yield
		\[
		\iint_{\hin_{[s_0, s]}}  |[Z^K, (H^{1,\mu\nu})^\flat \partial_\mu\partial_\nu]\habn| |\partial_t Z^K\hab| dxdydt \lesssim (C_1^3+C_0^2)\ep^3 s^{2\delta_k}.
		\]
		Finally, from \eqref{comm_KGw} and the energy bounds \eqref{bootin1}, \eqref{bootin5} we get
		\[
		\iint_{[s_0, s]}|[Z^K, (H^{1,\bmmu\bmnu})^\natural \partial_\bmmu\partial_\bmnu]\habf| |\partial_t Z^K\hab| dxdydt \lesssim C_1^3\ep^3 s^{\mu}
		\]
		with $\mu=1$ if $K$ is of type $(N+1, k)$ with $k\le N$, $\mu=1$ if $K$ is of type $(N, k)$ with $k\le N_1$. This concludes the proof of \eqref{comm_high} and \eqref{comm_low}.
	\end{proof}

	\subsection{Propagation of the pointwise bound \eqref{bootKG_in}}

	This section is devoted to the propagation of the pointwise estimates \eqref{bootKG_in} on the zero-average component $\habn$ of the solution. The equation satisfied by $Z^K\habn$ is obtained from the subtraction of \eqref{eq:diff_hff} from \eqref{h1_eqt_higher1}: 
	\begin{equation}\label{eq:habn_diff}
		\begin{aligned}
			\Box_{xy}Z^K\habn + (H^{\mu\nu})^\flat\partial_{\mu}\partial_{\nu} Z^K\habn = & F^{K, \natural}_{\alpha\beta}- (H^{1, \mu\nu})^\natural\partial_{\mu}\partial_{\nu} Z^K\habf - \big((H^{\mu\nu})^\natural\partial_{\mu}\partial_{\nu} Z^K\habn\big)^\natural
		\end{aligned}
	\end{equation}
	where $F^{K,\natural}_{\alpha\beta} = F^K_{\alpha\beta} - F^{K,\flat}_{\alpha\beta}$ and $F^K_{\alpha\beta}$ is defined in \eqref{source_ext_high1}. We observe that, after \eqref{dec_product}, pure zero-mode interactions do not appear in the above right hand side.

	The proof relies on the following result, which is motivated by the early work of Klainerman \cite{klainerman:global_existence} and can be found in slightly different forms in the works of LeFloch-Ma \cite{LeFloch-Ma16}, Dong-Wyatt \cite{DW20} and Huneau-Stingo \cite{HS}.

	\begin{proposition}\label{prp:KG_uniform}
		Suppose that $\phi$ is a solution of the equation
		\begin{equation}\label{eq:decayKG_model_PDE}
			\Box_{xy} \phi + (H^{\mu\nu})^\flat
			\partial_\mu \partial_\nu \phi= F , \quad (t,x,y) \in\mathbb{R}^{1+3}\times\mathbb{S}^1
		\end{equation}
		such that $\int_{\mathbb{S}^1}\phi\, dy=0$. For each $(t,x)$ in the cone $\{t>r\}$, let $s=\sqrt{t^2-r^2}$ and $Y_{tx}, A_{tx}, B_{tx}$ be functions defined as follows
		\begin{align*}
			Y^2_{tx}(\lambda)&:=\int_{\mathbb{S}^1}\lambda\left|\frac32 \phi_\lambda + (\scal\phi)_\lambda\right|^2+\lambda^3|\partial_y\phi_\lambda|^2 dy
			\\
			A_{tx}(\lambda)&:=\sup_{\mathbb{S}^1}\left|\lambda^{-1}\left(\scal\left((t/s)^2 (H^{1,UV})^\flat c^{00}_{UV}\right)\right)_\lambda\right| + \sup_{\mathbb{S}^1}\left|\lambda^{-1}(\scal (H^{1,44})^\flat)_\lambda\right|
			\\& \quad + \sup_{\mathbb{S}^1}\left|\lambda^{-1}(\chi(t/r)\chi'(r)(t/s)^2 M )_\lambda\right|
			\\
			B^2_{tx}(\lambda)&:=\int_{\mathbb{S}^1}\lambda^{-1}|(R[\phi])_\lambda|^2dy,
		\end{align*}
		where $f_\lambda(t,x,y) = f(\tfrac{\lambda t}{s}, \tfrac{\lambda x}{s}, y)$ and $R[\phi]:= R_0[\phi] +\sum_{i=1}^2 R^0_i[\phi] + \sum_{i=1}^3 R^1_i[\phi] - s^2 F$
		with
		\[
		\begin{aligned}
			R_0[\phi]&:=s^2 \pb^\bma \pb_\bma \phi + x^\bma x^\bmb\pb_\bma \pb_\bmb \phi+\tfrac{3}{4} \phi + 3 x^\bma \pb_\bma \phi
			\\
			R^0_1[\phi]& := s^2 \chi(t/r)\chi(r)\frac{M}{r} \left( (x^\bma/t)\partial_t\pb_\bma\phi +(x^\bma/t)\pb_\bma\partial_t\phi- \pb^\bma\pb_\bma\phi+(3/t)\partial_t \phi\right) \\
			R^0_2[\phi]&:= \chi(t/r)\chi(r)\frac{M}{r}\frac{t^2+r^2}{s^2}\cdot Q[\phi]
			\\
			R^1_1[\phi ] &:= s^{2}(H^{1,UV})^\flat\left( c_{UV}^{\bma \beta}\pb_\bma\pb_\beta\phi+   c_{UV}^{\alpha \bmb}\pb_\alpha\pb_\bmb \phi + c_{UV}^{4\bma} \partial_y \pb_\bma \phi+ d_{UV}^{\mu}\pb_\mu\phi\right)
			\\
			R^1_2[\phi] &:= - (t/s)(H^{1,UV})^\flat c^{40}_{UV}\cdot s\left(  \tfrac32  \partial_y \phi + x^\bma\partial_y \pb_\bma\phi \right)
			\\
			R^1_3[\phi] &:= -(t/s)^2(H^{1,UV})^\flat c^{00}_{UV} \cdot Q[\phi]
		\end{aligned}
		\]
		and
		\[
		Q[\phi] = \left( \tfrac34 +s^2 \left( 2(x^\bma/t) \pb_\bma\partial_t + s^{-2} x^\bma x^\bmb \pb_\bma \pb_\bmb + (r^2/t^3)\partial_t + (3/t) \partial_t + 3s^{-2} x^\bma \pb_\bma\right) \right)\phi.
		\]
		Then, in the hyperbolic region $\q H_{[s_0, S_0)}$, the following inequality holds 
		\[
		s^{\frac32} \left( \| \phi\|_{L^2(\mathbb{S}^1)} + \|\partial_y \phi\|_{L^2(\mathbb{S}^1)}\right) + s^{\frac12} \|\scal \phi\|_{L^2(\mathbb{S}^1)}\lesssim
		\left( Y_{tx}(s_0) + \int_{s_0}^s B_{tx}(\lambda) d \lambda\right) \exp^{\int_{s_0}^s A_{tx}(\lambda) d \lambda}.
		\]
	\end{proposition}
	\begin{proof}
		The wave operator in question writes in terms of the $\partial_t$ and $\pb_\bma$ derivatives as follows
		\[
		-\Box_{xy}  = (s/t)^2\partial^2_t  +2 (x^\bma/t)\pb_\bma \partial_t - \pb^\bma\pb_\bma + (r^2/t^3)\partial_t  + (3/t)\partial_t -\Delta_y.
		\]
		For some $\lambda>0$ and fixed $(t,x,y)$, we define $\omega_{txy}(\lambda):=\lambda^{3/2} \phi(\tfrac{\lambda t}{s}, \tfrac{\lambda x}{s}, y)$ to be the evaluation of $\phi$ on the hyperboloid $\mathcal{H}_\lambda$ dilated by $\lambda^{3/2}$. 
		We compute
		\[
		\begin{aligned}
			\dot{\omega}_{txy} &= \lambda^{1/2}\left(\tfrac32 \phi + (\scal \phi) \right)_\lambda
			\\
			\ddot{\omega}_{txy} &= \lambda^{-1/2}(P \phi)_\lambda := \lambda^{-1/2} \left(\tfrac34 \phi + 3 (\scal \phi) + (t^2\partial_t^2 +2tx^\bma\partial_\bma\partial_t + x^\bma x^\bmb\partial_\bma\partial_\bmb)\phi \right)_\lambda.
		\end{aligned}
		\]
		A calculation shows that 
		\begin{equation}\label{Pphi}
			\begin{aligned}
				P\phi&=  s^2 \left( -\Box_{txy} + \pb^\bma\pb_\bma + \Delta_y\right)\phi + x^\bma x^\bmb \pb_\bma\pb_\bmb \phi+ 3x^\bma\pb_\bma \phi + \tfrac34 \phi .
			\end{aligned}
		\end{equation}
		Using equation \eqref{eq:decayKG_model_PDE} we derive that  $\omega_{txy}(\lambda)$ satisfies 
		\[
		\ddot{\omega}_{txy} - \Delta_y \omega_{txy}
		= \lambda^{-1/2} (s^2 (H^{\mu\nu})^\flat\partial_\mu\partial_\nu \phi)_\lambda +\lambda^{-1/2}R_0[\phi]_\lambda  - \lambda^{3/2} F_\lambda.
		\]
		
		For the curved part in the above expression, we expand $(H^{\mu\nu})^\flat = (H^{1, \mu\nu})^\flat + H^{0, \mu\nu}$ where $H^0$ is defined in \eqref{dec_H}. Starting with the $H^0$ piece, we compute
		\[
		H^{0,\mu\nu}\partial_\mu \partial_\nu \phi= H^{0,\bmmu\bmnu}\partial_\bmmu \partial_\bmnu \phi= -a(t,x) \left(1+(r/t)^2\right)\partial_t^2 \phi + s^{-2} R^0_1[\phi].
		\]
		where for simplicity we put $a(t,x) := \chi(r/t)\chi(r)\tfrac{M}{r}$. 
		Using the calculation for $P\phi$ given in \eqref{Pphi}, we find
		\[\aligned
		- \lambda^{-1/2} (s^2 a(t,x)\left(1+(r/t)^2\right)\partial_t^2 \phi)_\lambda
		&= - \lambda^{-1/2}\left(a(t,x)(t/s)^2(1+r^2/t^2)\cdot s^{2} (s/t)^2 \partial_t^2 \phi\right)_\lambda
		\\&= - \left(a(t,x)(t/s)^2(1+r^2/t^2)\right)_\lambda \ddot{\omega}_{txy} + \lambda^{-1/2} R^0_2[\phi]_\lambda.
		\endaligned
		\]
		For the $H^{1, \flat}$ part, we use \eqref{c_coefficients} and \eqref{eq:curved_part_semihyp}. We find
		\[
		(H^{1,{\mu\nu}})^\flat \partial_\mu\partial_\nu \phi= \big[(H^{1,UV})^\flat c_{UV}^{00}\partial_t^2 + (H^{1,UV})^\flat c_{UV}^{04} \partial_t\partial_y  + (H^{1,44})^\flat \partial_y^2 \big]\phi + s^{-2}R^1_1[\phi].
		\]
		Since we can write $\dot{\omega}_{txy} = \lambda^{1/2}\left( \tfrac32 \phi + (s^2/t)\partial_t\phi + x^\bma\pb_\bma\phi\right)_\lambda$, we find
		\[\aligned
		\lambda^{-1/2} (s^2(H^{1,UV})^\flat c^{04}_{UV} \partial_t\partial_y\phi )_\lambda 
		&= \left((t/s)(H^{1,UV})^\flat c^{04}_{UV} \partial_y ( \lambda^{1/2} \tfrac{s^2}{t}\partial_t \phi)\right)_\lambda
		\\&= \left((t/s)(H^{1,UV})^\flat c^{04}_{UV}\right)_\lambda \partial_y \dot{\omega}_{txy}+\lambda^{-1/2}R^1_2[\phi]_\lambda .
		\endaligned
		\]
		In a similar way, using also the calculation for $P\phi$ given in \eqref{Pphi}, we find
		\[\aligned
		\lambda^{-1/2} (s^2(H^{1,UV})^\flat c^{00}_{UV} \partial_t^2 \phi)_\lambda
		&= \left( (t/s)^2(H^{1,UV})^\flat c^{00}_{UV} \right)_\lambda\ddot{\omega}_{txy} + \lambda^{-1/2} R^1_3[\phi]_\lambda.
		\endaligned
		\]
		
		For simplicity, we henceforth write $\omega_{txy} = \omega(\lambda)$ and suppress also the $|_\lambda$ notation. Putting the above computations together, we have
		\begin{equation}\label{eq:decayKG_lambda_PDE}
			\begin{aligned}
				b(t,x) \ddot{\omega}&- (1+ (H^{1,44})^\flat)\Delta_y \omega- (t/s)(H^{1,UV})^\flat c^{04}_{UV} \partial_y \dot{\omega}= \lambda^{-1/2} R[\phi]
			\end{aligned}
		\end{equation}
		where 
		\begin{equation}\label{def_b}	
			b(t,x) := 1 -(t/s)^2(H^{1,UV})^\flat c^{00}_{UV}+\chi(t/r)\chi(r)\frac{M}{r}(t/s)^2(1+r^2/t^2).
		\end{equation}	
		We multiply  \eqref{eq:decayKG_lambda_PDE}  by $\partial_\lambda \omega$, integrate over $\mathbb{S}^1$ and integrate by parts to get:
		\[
		\begin{aligned}
			&\int_\Circle \partial_\lambda \omega \Big( b \partial_\lambda^2 \omega
			- (1+(H^{1,44})^\flat)\Delta_y \omega -  (t/s)(H^{1,UV})^\flat c^{04}_{UV}\partial_y\partial_\lambda \omega \Big) dy 
			\\&= 
			\frac{d}{d\lambda}\Big(\frac12 \int_\Circle b|\partial_\lambda \omega|^2 dy +
			(1+(H^{1,44})^\flat) |\partial_y\omega|^2\Big) - \frac12\int_{\S^1}(\partial_\lambda b)| \partial_\lambda \omega |^2  dy.
		\end{aligned}
		\]
		Recalling the definition of $b$ in \eqref{def_b}, we obtain
		\[
		\aligned
		&\frac{d}{d\lambda}\Big( \int_\Circle b|\partial_\lambda \omega|^2 + \left( 1 +(H^{44})^\flat \right)|\partial_y \omega|^2 dy \Big)
		\\&= - \int_\Circle \partial_\lambda \left((t/s)^2 (H^{1,UV})^\flat c^{00}_{UV}\right) |\partial_\lambda \omega|^2 dy 
		+ \int_\Circle \partial_\lambda \left(\chi(t/r)\chi(r)(t/s)^2(1+r^2/t^2)\right)\tfrac{M}{r} |\partial_\lambda \omega|^2 dy 
		\\&
		+ 	\int_\Circle \chi(t/r)\chi(r)(t/s)^2(1+r^2/t^2)\partial_\lambda \left(\tfrac{M}{r}\right) |\partial_\lambda \omega|^2 dy +  2 \int_\Circle \lambda^{-1/2}R[\phi]\partial_\lambda\omega dy .
		\endaligned
		\]
		We crucially can drop the third term on the RHS above using the fact that $\chi \geq 0$, $M>0$ and the identity $\partial_\lambda (\tfrac{M}{ r})_\lambda =- (\tfrac{M}{s r})_\lambda$. Note also that the cut-off function $\chi$ is supported for  $2r>t>2$ and so in the region $\{t\geq r+1\}$ we have $|\chi(t/r)\chi(r)\tfrac{M}{r}(t^2/s^2)|\lesssim \epsilon$.  
		
		By relation \eqref{H1-UV-cUV-exp} with $\pi = H^{1, \flat}$, the estimates \eqref{null-hyp-coeff}, \eqref{KS6_h},  and \eqref{good_coefficient2}, and the fact that $t/s^2\leq 1$ in the interior of the light cone,  we find
		\[
		\sup_{\mathbb{S}^1}|(t/s)^2(H^{1,UV})^\flat c^{00}_{UV}| + |(H^{1,44})^\flat|
		\lesssim \varepsilon.
		\]
		All together, we obtain
		\[
		\frac{d}{d\lambda}Y^2_{tx}(\lambda)\lesssim A_{tx}(\lambda)Y^2_{tx}(\lambda)+B_{tx}(\lambda)Y_{tx}(\lambda)
		\]
		with $A_{tx}, B_{tx}, Y_{tx}$ as in the statement. From Gr\"onwall's lemma, we have
		\begin{equation}\label{eq:decayKG_Gron}
			Y_{tx}(s)\lesssim \left( Y_{tx}(s_0) + \int_{s_0}^s B_{tx}(\lambda) d \lambda\right) \exp \left( \int_{s_0}^s A_{tx}(\lambda) d \lambda \right),
		\end{equation}
		and the conclusion follows from the Poincar\'e inequality. 
	\end{proof}
	
	\begin{proposition}\label{prp:KG_imp_decay}
		There exists a constant $C_2$ sufficiently large, a finite and increasing sequence of parameters $0< \gamma_k, \delta_k,\zeta_k\ll 1$ satisfying \eqref{condition_parameters}, and $0< \ep_0\ll 1$ sufficiently small such that, under the assumptions of Proposition \ref{prop:bootstrap_int}, we have \eqref{bootKG_in_enh}.
	\end{proposition}
	
	\begin{proof}
		Throughout the proof, $0<\eta\le 2\delta_N\ll 1$ will denote a constant that linearly depends on $\zeta_k, \gamma_k, \delta_k$.
		We apply Proposition \ref{prp:KG_uniform} to the variable $\mathbf{W}=\partial^I\Gamma^J\habn$ with $|I|+|J|\le N_1+1 = N-4$ and $|J|\le N_1$, governed by the PDE \eqref{eq:habn_diff}.
		
		\textit{1. The $A_{tx}(\lambda)$ term:}
		this is the same for all values of $\mathbf{W}$. 
		Bound \eqref{KS3_hff} gives that 
		\[
		\sup_{\mathbb{S}^1} \lambda^{-1} |(\scal (H^{1,44})^\flat)_\lambda|
		\lesssim \epsilon\lambda^{-\frac32 + \delta_2}.
		\]
		For the other piece of $A_{tx}(\lambda)$ we use \eqref{H1-UV-cUV-exp} with $\pi = H^{1, \flat}$, the identities $\scal(s) = s$ and $\scal(t+r)=t+r$, the pointwise bounds \eqref{KS3_hff}, \eqref{good_coefficient11}, \eqref{good_coefficient2} and the fact that $t/s^2 \leq 1$ in the interior of the lightcone, to derive
		\[
		\begin{aligned}
			|\scal\left((t/s)^2(H^{1,UV})^\flat c^{00}_{UV}\right)|
			&\lesssim  |(s/t)^2 \scal H^1_{\Lb\Lb}| + |(t/s)^2 \scal (H^1_{LL})^\flat| + |\scal H^1_{L\Lb}|
			\lesssim \epsilon s^{-\frac12+\gamma_1}.
		\end{aligned}
		\]
		Note that the estimate of the second term in the above right hand side is obtained using also the following decomposition of the scaling vector field
		\begin{equation}\label{rewritten-scal}
			\scal= (t-r)\partial_t + (r-t)\partial_r + (x^\bma/r)\Omega_{0\bma}.
		\end{equation}
		Consequently
		\[
		\begin{aligned}
			\lambda^{-1} |\left(\scal\left((t/s)^2(H^{1,UV})^\flat c^{00}_{UV}\right)\right)_\lambda|
			&\lesssim \epsilon \lambda^{-\frac32+\gamma_1}.
		\end{aligned}
		\]
		Finally, we use that $\chi(r) =O(1)$ on its support, so that we can bound $|\chi'(r)|\lesssim r^{-2}$ and thus get
		\[
		\lambda^{-1} \left|(\chi(t/r)\chi'(r)(t/s)^2 M )_\lambda\right| \lesssim \epsilon\lambda^{-2}.
		\]
		Bringing all this together gives
		\begin{equation} \label{uniformKG-A}
			\int_{s_0}^s A_{tx}(\lambda)d \lambda \lesssim \epsilon.
		\end{equation}

		\textit{2. The $Y_{tx}(s_0)$ term:}  We evaluate all the expressions here on the hyperboloid $\hin_{s_0}$ for $s_0$ close to 2. We observe that $1+t+r=O(1)$ on $\hin_{s_0}$, so by \eqref{rewritten-scal} and lemma \ref{lm:ksobolev}
		\begin{equation}
			|Y_{tx}(s_0)|
			\lesssim \| \mathbf{W}_{s_0}\|_{L^2_y}
			+ \|(\scal\mathbf{W})_{s_0}\|_{L^2_y}+ \|(\partial_y\mathbf{W})_{s_0}\|_{L^2_y}
			\lesssim
			\Big(\frac{s}{t}\Big)^\frac32 \enint(s_0, Z^{\le 2}\Wf)^\frac12.
			\label{uniformKG-Y2}
		\end{equation}

		\textit{3.a. The $R_0$ term in $B_{tx}(\lambda)$:}
		By \eqref{bootinL2y.1}, \eqref{bootinL2y.5}, and \eqref{bootinL2y.6} we obtain
		\begin{equation}\label{KG-uniform:R1.1}
			\| \lambda^{-1/2}( R_0[\mathbf{W}])_\lambda \|_{L^2_y}\lesssim C_1\ep \lambda^{-\frac32+\zeta_{k+4}} \Big(\frac{s}{t}\Big)^\frac32.
		\end{equation}
		
		\textit{3.b. The $R^1_1$ term in $B_{tx}(\lambda)$:} 
		First note that in the region $\{r<t\}$ we have $|c^{\bma \beta}_{UV}|\lesssim 1$ using \eqref{null-hyp-coeff} and straightforward computations. Thus, by \eqref{bootW_in}, \eqref{bootinL2y.1} and \eqref{bootinL2y.5},
		\begin{equation}\label{KG-uniform:R2.1}
			\begin{aligned}
				\| \lambda^{-1/2}( R^1_1[\mathbf{W}])_\lambda \|_{L^2_y}
				& \lesssim
				\lambda^\frac32 \left| H^{1,\flat}\right| \Big( \left\|\partial \pb \Wf\right\|_{L^2_y} + t^{-1}\|\Wf\|_{L^2_y}\Big)\\
				&\lesssim C_1^2 \ep^2 \lambda^{-\frac32+\gamma_0+\zeta_{k+3}} \Big(\frac{s}{t}\Big)^\frac52.
			\end{aligned}
		\end{equation}
		
		\textit{3.c The $R^1_2$ term in $B_{tx}(\lambda)$:} using \eqref{bootW_in}, \eqref{bootinL2y.1}, \eqref{bootinL2y.5} we get
		\[
		\aligned
		\| \lambda^{-1/2}( R^1_2[\mathbf{W}])_\lambda \|_{L^2_y}& \lesssim \lambda^\frac12 \big|\big((t/s)((H^{1,UV})^\flat c^{40}_{UV}\big)_\lambda\big| \big(\left\|\partial_y\Wf\right\|_{L^2_y} + t \left\|\partial_y \pb\Wf\right\|_{L^2_y}\Big)\\
		& \lesssim C_1^2\ep^2 \lambda^{-\frac32+\gamma_0+\zeta_{k+3}}\Big(\frac{s}{t}\Big)^\frac32.
		\endaligned
		\]

		\textit{3.d.i The $R^1_3$ term in $B_{tx}(\lambda)$:}
		we observe that from \eqref{bootinL2y.1}, \eqref{bootinL2y.5}, \eqref{bootinL2y.6} 
		\[
		\aligned
		\|Q[\Wf]\|_{L^2_y} & \lesssim \big(\|\Wf\|_{L^2_y} + \left\|x^2\pb^2\Wf\right\|_{L^2_y} + \left\| x\pb \Wf\right\|_{L^2_y} \big) + \big(\left\|s^2 \pb\partial \Wf\right\|_{L^2_y} + \left\| (s^2/t)\partial \Wf\right\|_{L^2_y}\big) \\
		& \lesssim C_1\ep t^{-\frac32}s^{\frac12+\zeta_{k+4}} + C_1\ep t^{-\frac52}s^{\frac52+\zeta_{k+3}}
		\endaligned
		\]		
		which, coupled to \eqref{bootW_in} and the fact that $s^{-2}\lesssim t^{-1}$ yields again
		\[
		\aligned
		\| \lambda^{-1/2}( R^1_3[\mathbf{W}])_\lambda \|_{L^2_y} & \lesssim \big|\big(s^{-\frac12}(t/s)^2 (H^{1,UV})^\flat c^{00}_{UV}\big)_\lambda\big|	\|Q[\Wf]_\lambda\|_{L^2_y} \lesssim C_1^2\ep^2 \lambda^{-\frac32 + \gamma_0+\zeta_{k+3}}\Big(\frac{s}{t}\Big)^\frac32.
		\endaligned
		\]
		%
		
		\textit{3.e The $R^0_1$ and $R^{0}_2$ terms in $B_{tx}(\lambda)$:} satisfy the same bounds as $R^1_1$ and $R^1_3$ respectively.
		%
		%
		%
		
		\medskip
		In summary, for $\mathbf{W}=\partial^I\Gamma^J\habn$ with $|I|+|J|\le N-4 = N_1+1$
		\begin{equation}\label{uniform-KG:allRs}
			\| \lambda^{-1/2}(R[\mathbf{W}]+s^2 F)_\lambda \|_{L^2(\S^1)}
			\lesssim C_1\ep \lambda^{-\frac32+\eta}\Big(\frac{s}{t}\Big)^\frac32.
		\end{equation}
		
		\textit{4.a The source term $F$:}
		we simply distribute derivatives and vector fields across the nonlinearities given in \eqref{eq:habn_diff}. We begin by analyzing the quadratic interactions of zero-average component with itself, which are of the form 
		\[
		\sum_{\substack{|I_1|+|I_2|=|I|\\ |J_1|+|J_2|\le |J|}} \Big(\partial (\partial^{I_1}\Gamma^{J_1}\hnn) \cdot  \partial (\partial^{I_2}\Gamma^{J_2}\hnn) \Big)^\natural +  \Big((\partial^{I_1}\Gamma^{J_1}\hnn) \cdot\partial^2  ( \partial^{I_2}\Gamma^{J_2}\hnn)\Big)^\natural
		\]
		and recall that there must exist an index $l=1,2$ such that $|I_l|+|J_l|\le \floor{N_1+1/2}$. 
		
		To estimate the first sum, we use \eqref{bootKG_in} and \eqref{KS1_hnn} to obtain that
		\[
		\aligned
		\sum_{\substack{|I_1|+|I_2|=|I|\\ |J_1|+|J_2|\le |J|}}\lambda^\frac32\left\| \Big(\partial (\partial^{I_1}\Gamma^{J_1}\hnn) \cdot  \partial (\partial^{I_2}\Gamma^{J_2}\hnn)\Big)^\natural\right\|_{L^2_y}& \lesssim \lambda^\frac32 \left\|\partial Z^{\le N_1-1}\hnn\right\|_{L^\infty_y}\left\| \partial Z^{\le N_1}\hnn\right\|_{L^2_y}
		\\
		& \lesssim C_1^2\ep^2 \lambda^{-\frac32+\eta}\Big(\frac{s}{t}\Big)^2.
		\endaligned
		\]
	All terms in the second sum are estimated in the same way, besides the one corresponding to $|I_2| + |J_2| = |I|+|J| = N_1+1$. For this one we use \eqref{bootinL2y.1} and \eqref{KS1_hnn}, together with the assumption \eqref{condition_parameters} on the parameters (i.e. $\zeta_i\ll \gamma_j$ for any $i, j$), and derive that
		\[
		\lambda^\frac32 \left\| \Big(\hnn \cdot\partial^2  ( \partial^{I}\Gamma^{J}\hnn)\Big)^\natural\right\|_{L^2_y}\lesssim \lambda^\frac32 \|\hnn\|_{L^\infty_y}\left\|\partial^2 ( \partial^I \Gamma^J\hnn) \right\|_{L^2_y}\lesssim C_1C_2\ep^2 \lambda^{-1+\gamma_k}\Big(\frac{s}{t}\Big)^2.
		\]
Let us note that this term is highly specific to the Kaluza--Klein problem and is absent in the Einstein-Klein Gordon equations. We also observe that, in the case where $|I|+|J|=N_1+1$ but $ |J| \leq N_1-2$ we can use \eqref{bootinL2y.1bis} instead of \eqref{KS1_hnn} to obtain
		$$	\lambda^\frac32 \left\| \Big(\hnn \cdot\partial^2  ( \partial^{I}\Gamma^{J}\hnn)\Big)^\natural\right\|_{L^2_y}\lesssim C_1C_2\ep^2 \lambda^{-\frac{3}{2}+\eta}\Big(\frac{s}{t}\Big)^2.$$
		
		\medskip			
		Next we turn to the mixed interactions between the zero mode and the zero-average component. We begin with the commutator terms $[\partial^I \Gamma^J, (H^{1, \mu\nu})^\flat \partial_\mu\partial_\nu]\habn$, which we rewrite using \eqref{comm_int} with $\pi = H^{1, \flat}$. We focus on treating the following products
		\begin{equation}\label{terms_comm}
			\aligned
		&	\partial^{I_1}\Gamma^{J_1}H^{1, \flat}_{LL} \cdot\partial^2_t (\partial^{I_2}\Gamma^{J_2}\habn),\quad  \partial^{I_1}\Gamma^{J_1}H^{1, \flat}_{4L} \cdot\partial_t \partial_y (\partial^{I_2}\Gamma^{J_2}\habn),\quad \partial^{I_1}\Gamma^{J_1}H^{1, \flat}_{44} \cdot \partial^2_y (\partial^{I_2}\Gamma^{J_2}\habn)\\
	&\text{for }		|I_1|+|I_2| = |I|, \quad |J_1|+|J_2|\le |J|, \quad |I_2|+|J_2|<|I|+|J|
			\endaligned
		\end{equation}
		 the remaining ones being simpler. The latter term can be rewritten using the equation, i.e.
		\[
		\aligned
		\partial_y^2(\partial^{I_2}\Gamma^{J_2}\habn) =  \Big((s/t)^2\partial^2_t  +2 (x^\bma/t)\pb_\bma \partial_t - \pb^\bma\pb_\bma + (r^2/t^3)\partial_t  + (3/t)\partial_t\Big)(\partial^{I_2}\Gamma^{J_2}\habn) \\ + \Box_{xy}(\partial^{I_2}\Gamma^{J_2}\habn).
		\endaligned
		\]
		If $|I_1|>0$, we estimate the first two terms in \eqref{terms_comm} using \eqref{bootKG_in} and \eqref{good_coeff1_flat} and the latter one using \eqref{KS1_hff} and \eqref{bootKG_in}, hence getting that they are all bounded by $C_1C_2\ep^2 \lambda^{-\frac32+\eta}(s/t)^2$.
		If $|I_1|=0$ and $|J_1|>0$, we use \eqref{bootW_in} and \eqref{bootKG_in} for all terms, together with the algebraic relation \eqref{condition_parameters}, obtaining
		\[
		\lambda^\frac32 \left|\Gamma^{J_1}\hff \right| \left\|\partial^2(\partial^I\Gamma^{J_2}\habn)\right\|_{L^2_y} \lesssim  C_2^2\ep^2 \Big(\frac{s}{t}\Big)^\frac{3}{2}
			\lambda^{-1+\gamma_{k}}.
		\]
		
Turning now to the semilinear interactions between the zero-mode and the zero-average component, we immediately obtain from \eqref{bootKG_in} and \eqref{KS1_hff} that
		\[
		\lambda^\frac32 |\partial (\partial^{I_1}\Gamma^{J_1}\hff)| \left\| \partial (\partial^{I_2}\Gamma^{J_2}\hnn)\right\|_{L^2_y}\lesssim C_1C_2\ep^2 \lambda^{-\frac32+\eta}\Big(\frac{s}{t}\Big)^2, \quad |I_2|+|J_2|\le N_1.
		\]
		When $|I_2|+|J_2| = |I|+|J| = N_1+1$, \eqref{bootKG_in} does not provide us with the right power of $(s/t)$, which we instead get using the structure of the semilinear terms. On the one hand, using \eqref{null_structure} together with \eqref{bootKG_in}, relation $\pb_\bma = (1/t)\Omega_{0\bma}$, \eqref{KS2_h} and \eqref{KS1_hff}, we easily derive that
		\[
		\lambda^\frac32 \left\| \mathbf{Q}(\partial \hff,\partial( \partial^{I}\Gamma^{J}\hnn) )\right\|_{L^2_y}\lesssim C_1C_2\ep^2 \lambda^{-\frac32+\eta}\Big(\frac{s}{t}\Big)^2.
		\]
		The cubic terms also satisfy the same estimate as above, we leave the details to the reader.
		
		On the other hand, using lemma \ref{lem:weak_null_frame} we see that
		\begin{multline*}
			\left\|P_{\alpha\beta}(\partial \hff, \partial(\partial^I\Gamma^J\hnn))\right\|_{L^2_y} \\
			\lesssim |\partial \hff_{TU}| \|\partial (\partial^I\Gamma^J\hnn)_{TU}\|_{L^2_y} + |\partial \hff_{LL}| \|\partial (\partial^I \Gamma^J \hnn)_{\Lb\Lb}\|_{L^2_y} + |\partial\hff_{\Lb\Lb}|\|\partial(\partial^I\Gamma^J \hnn)_{LL}\|_{L^2_y}.
		\end{multline*}
		From \eqref{bootKG_in} and \eqref{good2}
		\[
		\lambda^\frac32 |\partial \hff_{TU}| \|\partial (\partial^I\Gamma^J\hnn)_{TU}\|_{L^2_y}\lesssim C_1C_2 \lambda^{-1+\gamma_k}\Big(\frac{s}{t}\Big)^\frac32;
		\]
		from \eqref{bootKG_in} and \eqref{good_coeff1_flat}
		\[
		\lambda^\frac32|\partial \hff_{LL}| \|\partial (\partial^I \Gamma^J \hnn)_{\Lb\Lb}\|_{L^2_y}\lesssim C_1C_2\ep^2 \lambda^{-\frac32+\eta}\Big(\frac{s}{t}\Big)^2;
		\]
		finally, since
		\[
		|\partial(\partial^I\Gamma^J \hnn)_{LL}| \le |\partial(\partial^I\Gamma^J h^1)_{LL}| + |\partial(\partial^I\Gamma^J \hff)_{LL}|
		\]
		from lemma \ref{lem:wave_cond_h1} and pointwise estimates \eqref{bootW_in}, \eqref{bootKG_in}, \eqref{KS1_hff}, \eqref{good_coeff1_flat} we deduce that
		\[
		\|\partial(\partial^I\Gamma^J \hnn)_{LL}\|_{L^2_y}\lesssim C_2 \ep t^{-\frac32}s^{\gamma_{k}}
		\]
		and consequently that
		\[
		|\partial\hff_{\Lb\Lb}|\|\partial(\partial^I\Gamma^J \hnn)_{LL}\|_{L^2_y}\lesssim C_1C_2\ep^2 \lambda^{-\frac32+\eta}\Big(\frac{s}{t}\Big)^2.
		\]
		
		In summary,
		\[
		\lambda^\frac32 \|F_\lambda\|_{L^2_y}\lesssim C_2^2\ep^2 \Big(\frac{s}{t}\Big)^\frac32
		\begin{cases}
			\lambda^{-\frac32+\eta}, \ &\text{if  } |I|\le N_1,\ |J|=0\\
			\lambda^{-1+\gamma_k}, \ &\text{if  } |I|+|J|\le N_1+1,\ |J|=k\le N_1\\

		\end{cases}
		\]
		and
		\[
		\int_{s_0}^s B_{tx}(\lambda)d\lambda \lesssim (C_1\ep +  C_2^2\ep^2) \Big(\frac{s}{t}\Big)^\frac32 
		\begin{cases}
			1 \ &\text{if  } |I|\le N_1,\ |J|=0\\
			s^{\gamma_k}, \ &\text{if  } |I|+|J|\le N_1+1, |J|=k\le N_1\\
		\end{cases}
		\]
		
		Finally, from Proposition \ref{prp:KG_uniform} we obtain that there exists a constant $C>0$ such that
		\[
		\aligned
		s^\frac32\|\partial_y^{\le 1}(\partial^I\Gamma^J\hnn)\|_{L^2} + s^\frac12 \|\scal (\partial^I\Gamma^J\hnn)\|_{L^2} \le C\Big(\frac{s}{t}\Big)^\frac32 \enint(s_0, Z^{\le 2}(\partial^I\Gamma^J\hnn))^\frac12 \\
		+ C(C_1\ep +  C_2^2\ep^2) \Big(\frac{s}{t}\Big)^\frac32 
		\begin{cases}
			1 \ &\text{if  } |I|\le N_1,\ |J|=0\\
			s^{\gamma_k}, \ &\text{if  } |I|+|J|\le N_1+1, |J|=k\le N_1\\
		\end{cases}
		\endaligned
		\]
		The conclusion of the proof then follows from the following relation
		\[
		\partial_t = \frac{t}{s^2}\scal - \frac{tx^\bmj}{s^2}\pb_\bmj, \quad \partial_\bmj = \pb_\bmj -\frac{x_\bmj}{s^2}\scal + \frac{x_\bmj x^\bmk}{s^2}\pb_\bmk
		\]
		and by choosing $C_2$ sufficiently large so that $CC_1\ll C_2$ and  $C\enint(s_0, Z^{\le 2}(\partial^I\Gamma^J\hnn))^\frac12\ll (C_2\ep)$, together with $\ep_0$ sufficiently small so that $CC_2\ep\ll 1$.
	\end{proof}

\begin{remark}\label{remark_lossKG}
It will be useful in view of Proposition \ref{prpenhancedsup} to observe that the loss $s^{\gamma_k}$ in the estimate of $\partial^I\Gamma^J\habn$ when $|I|+|J|\le N_1$ and $|J|=k$ is only due to the following contributions, which arise from the commutator term between zero modes and zero-average components
\[
\Gamma^{J_1}\hff \cdot \partial^2 (\partial^I\Gamma^{J_2}\habn), \quad |J_1|>0.
\]
\end{remark}

	\subsection{Enhanced energy bounds for the zero modes}

	The goal of this subsection is to show that the lower order energies of the zero-modes enjoy enhanced energy estimates compared to \eqref{boot2_in}.

	\begin{proposition}\label{prop:low-order-zero-mode-energy}
		Under the assumptions of Proposition \ref{prop:bootstrap_int}, we have that for any fixed $s\in [s_0, S_0)$ and any multi-index $K$ of type $(N-1,N_1)$
		\begin{equation}\label{energy_wave_lower}
			\enint (s, Z^K \hff)^{1/2} \lesssim C_1 \epsilon s^{3\sigma}
		\end{equation}
		where $0<\sigma\ll \gamma_0$ is the rate of growth of the exterior energies.
	\end{proposition}
	The proof of the above statement is based on energy inequality \eqref{en_ineq_habf_high}. We recall the estimates already obtained in the previous subsections on the quadratic null terms \eqref{null_lower_int}, on quadratic weak null terms \eqref{weak_est4} and on the cubic terms \eqref{cub_lower_int} appearing in $F^{K, \flat}_{\alpha\beta}$, as well as on the commutator terms \eqref{comm_zeromode_low} and on the contributions coming from $F^{0,K}_{\alpha\beta}$ \eqref{est_F0_int}. We complete the picture with the estimates of the remaining trilinear terms appearing in the right hand side of \eqref{en_ineq_habf_high}.

	\begin{lemma}
		For any multi-index $K$ of type $(N, k)$ we have that
		\begin{gather}\label{int_H_energy}
			\iint_{\hin_{[s_0, s]}}\hspace{-10pt} |{\partial^\mu H_\mu}^\sigma\cdot \partial_\sigma Z^K \hab| |\partial_t Z^K\hab| + \frac{1}{2}|{\partial_t H_\mu}^\sigma \cdot \partial_\sigma Z^K\hab\cdot \partial^\mu Z^K\hab| \, dtdx \lesssim C_1^3 \ep^3 s^{\frac12 + 3\delta_N} \\
			\label{int_H_energy_flat}
			\iint_{\hin_{[s_0, s]}}|{\partial^\mu H_\mu}^\sigma\cdot \partial_\sigma Z^K \habf| |\partial_t Z^K\habf| + \frac{1}{2}|{\partial_t H_\mu}^\sigma\cdot\partial_\sigma Z^K\habf\cdot \partial^\mu Z^K\habf| \, dtdx \lesssim C_1^3 \ep^3,
		\end{gather}
		and for multi-indexes $K$ of type $(N_1, k)$
		\begin{equation} \label{int_H_energy_lower}
			\iint_{\hin_{[s_0, s]}}|{\partial^\mu H_\mu}^\sigma \cdot \partial_\sigma Z^K \hab| |\partial_t Z^K\hab| + \frac{1}{2}|{\partial_t H_\mu}^\sigma\cdot\partial_\sigma Z^K\hab\cdot \partial^\mu Z^K\hab| \, dtdx \lesssim C_1^3 \ep^3 .
		\end{equation}
	\end{lemma}
	\begin{proof} It is a straightforward consequence of \eqref{eq_null1}, \eqref{eq_null2} applied to $\phi = Z^K\hab$ and $\phi = Z^K\hff_{\alpha\beta}$ respectively, coupled with the pointwise bounds \eqref{h0_estimate}, \eqref{KS1_hff},  \eqref{good_coeff1_flat} and with the energy bounds \eqref{bootin1} to get \eqref{int_H_energy}, with \eqref{bootin2} to get \eqref{int_H_energy_flat} and with \eqref{bootin5} to get \eqref{int_H_energy_lower}.

	\end{proof}
	
	\begin{proof}[Proof of Proposition \ref{prop:low-order-zero-mode-energy}]
		This follows by plugging the estimates obtained so far in the energy inequality \eqref{en_ineq_habf_high}. In fact, from Lemma \ref{lem:prod_h0}, estimates \eqref{null_lower_int}, \eqref{cub_lower_int}, \eqref{weak_est4}, \eqref{comm_zeromode_low} and the a-priori energy bound \eqref{bootin2} there exists a constant $C>0$ such that for any fixed multi-index $K$ of type $(N_1-1, k)$  we have
		\[
		\begin{aligned}
			\iint_{\hin_{[s_0,s]}}\hspace{-10pt} | F^{K, \flat}_{\alpha\beta}(h)(\partial h, \partial h)||\partial_t Z^K\hff_{\alpha\beta}|dxdt +  \iint_{\hin_{[s_0, s]}} \big|(H^{1, \mu\nu})^\natural\cdot \partial_\mu\partial_\nu Z^K\habn\big)^\flat \big| |\partial_t Z^K\habf| dx dt\\
			\le \int_{s_0}^s CC_1\ep \tau^{-1}\sum_{K'}\enint(\tau, Z^{K'}\hff)^{1/2}\enint(\tau, Z^K\hff_{\alpha\beta})^{1/2}d\tau\\
			+ \int_{s_0}^s  CC_1\ep\tau^{-1+C\ep}\sum_{K''}\enint(\tau, Z^{K''}\hff)^{1/2}\enint(\tau, Z^K \hff_{\alpha\beta})^{1/2} d\tau + (C_0^2+C_1^2)C_2^2\ep^3
		\end{aligned}
		\]
		where $K'$ denote multi-indexes of type $(|K|, k)$ and $K''$ multi-indexes of type $(|K|-1, k-1)$.
		Furthermore, from \eqref{h0_estimate} 
		\begin{multline*}
			\iint_{\hin_{[s_0s]}} |[Z^K, H^{0, \bmmu\bmnu} 
			\partial_\bmmu\partial_\bmnu]\hff_{\alpha\beta}||\partial_t Z^K\hff_{\alpha\beta}|dxdt \\
			\le \int_{s_0}^s C\ep \tau^{-1}\sum_{K'}\enint(\tau, Z^{K'}\hff)^{1/2}\enint(\tau, Z^K \hff_{\alpha\beta})^{1/2}d\tau. 
		\end{multline*}
		Summing the above estimates up together with \eqref{boundary_term_int}, \eqref{est_F0_int} and \eqref{int_H_energy_flat} we get that there exists some positive constant $C>0$ such that
		\begin{multline*}
			\enint(s, Z^{\le K}\hff_{\alpha\beta})\le C\enint(s_0, Z^{\le K}\hff_{\alpha\beta}) + CC_0^2\ep^2s^{2\sigma+C\ep} + C(C_0^2+C_1^2)C_2^2\ep^3  \\
			+ \int_{s_0}^s CC_1\ep \tau^{-1}\sum_{K'}\enint(\tau, Z^{K'}\hff)^{1/2}\enint(\tau, Z^{\le K}\hff)^{1/2} d\tau\\
			+ \int_{s_0}^s  CC_1\ep\tau^{-1+C\ep}\sum_{K''}\enint(\tau, Z^{K''}\hff)^{1/2}\enint(\tau, Z^{\le K} \hff)^{1/2} d\tau.
		\end{multline*}
		Performing an induction argument on $k$, it then follows that there exist some positive constants $c_1<c_2<\dots<c_N$ such that
		\[
		\enint(s, Z^{\le K}\hff_{\alpha\beta})\le C\big(\enint(s_0, Z^{\le K}\hff_{\alpha\beta})+  CC_0^2\ep^2 + (C_0^2+C_1^2)C_2^2\ep^3\big)
		s^{2\sigma+c_k\ep} 
		\]
		so the end of the proof follows from the smallness assumptions on the data and after choosing $0<\ep_0\ll 1$ sufficiently small so that $c_N\ep_0\le \sigma$.
	\end{proof}
	
	The improved lower-order energy estimate \eqref{energy_wave_lower} leads to the following improved sup-norm estimates. These are obtained  following the proofs of their analogues with $s^{\delta_k}$ losses and using the energy bound \eqref{energy_wave_lower} in place of \eqref{bootin2}.
	For any multi-index $K$ of type $(N-3, N_1)$, bounds \eqref{KS1_hff} and \eqref{KS6_h} are enhanced to the following ones
	\begin{align}\label{KS-hff1}
		\big\|s t^\frac12 \partial Z^K \hff_{\alpha\beta}\big\|_{L^\infty(\hin_s)} + \big\|t^\frac32\pb Z^K \hff_{\alpha\beta}\big\|_{L^\infty(\hin_s)} + \big\|t^\frac12 Z^K \hff_{\alpha\beta}\big\|_{L^\infty(\hin_s)} &\lesssim C_1 \epsilon  s^{3\sigma} 
	\end{align}
	and bounds \eqref{good_coeff1_flat}, \eqref{good_coefficient2} are improved to
	\begin{gather}
		\|t^\frac{3}{2}\partial Z^K (H^1_{LT})^\flat\|_{L^\infty(\hin_s)} + \|t^\frac12 (t/s)^2 Z^K (H^1_{LT})^\flat\|_{L^\infty(\hin_s)} \lesssim C_1\ep s^{3\sigma} \label{KS-hff4}.
	\end{gather}
	Furthermore, for multi-indexes $K$ of type $(N-4, k)$ we can also improve \eqref{second_derivatives2} to the following
	\begin{equation} \label{KS-hff6}
		\|t^\frac32 (s/t)^2\partial^2_t Z^K\hff_{\alpha\beta}\|_{L^\infty(\hin_s)}\lesssim C_1\ep s^{-1+6\sigma}.
	\end{equation}

	\subsection{Propagation of pointwise bound \eqref{bootW_in}} 
	
	This section is dedicated to the proof of the enhanced pointwise bound \eqref{bootW_in}, see proposition \ref{prp:Linf-Linf_wave}. We will make use of the following lemmas, which are due respectively to Alinhac \cite{Alinhac06},  Asakura \cite{asakura86} and Katayama-Yokoyama \cite{KY06}.
	
	\begin{lemma}\label{lem:Alinhac}
		Let $u = u(t,x)$ be the solution to the inhomogeneous wave equation $\Box_{tx} u =F$ on the flat space $\R^{1+3}$ with zero initial data. Suppose that $F$ is spatially compacted supported and that there exist some constants $C_0>0$, $\mu,\nu\ge 0$ such that the following pointwise bound is satisfied
		\[
		|F(t,x)|\le C_0 t^{-\nu}(t-r)^{-\mu}.
		\]
		Defining $\Phi_\mu(s)=1, \log s, s^{1-\mu}/(1-\mu)$ according to $\mu>1, =1, <1$ respectively, we then have 
		\begin{itemize}
			\item[(i)] If $\nu>2$, 
			\[
			|u(t,x)|\le C C_0 \Phi_\mu\big(\langle t-r\rangle \big) \frac{\langle t-r\rangle^{\nu -2}}{\nu -2}(1+t)^{-1}
			\]
			\item[(ii)] If $\nu=2$, 
			\[
			|u(t,x)|\le C C_0 \Phi_\mu\big(\langle t-r\rangle \big)(1+t)^{-1}\log(1+t)
			\]
			\item[(iii)] If $\nu<2$,
			\[
			|u(t,x)|\le C C_0 \Phi_\mu\big(\langle t-r\rangle \big)\frac{(1+t)^{1-\nu}}{2-\nu}.
			\]
		\end{itemize}
	\end{lemma}
	
	\begin{lemma}\label{lem:asakura}
		Let $\phi, \psi$ be smooth functions on $\R^3$ such that
		\[
		|\phi|\le C(1+|x|)^{-1-\kappa}, \qquad |\nabla \phi| + |\psi |\le C(1+|x|)^{-2-\kappa}
		\]
		for some constant $C>0$ and some fixed $0<\kappa<1$.
		Let $u$ be the solution to the homogeneous wave equation $\Box u = 0$ with initial $(u,\partial_t u)|_{t=0} = (\phi, \psi)$. There exists a constant $\tilde{C}>0$ such that $u$ satisfies the following inequality
		\[
		|u(t,x)|\le \frac{C\tilde{C}}{(1+t+|x|)(1+|t-r|)^\kappa}.
		\]
	\end{lemma}

	\begin{lemma}\label{lem:kata-yoko}
		Let $u$ be the solution to the inhomogeneous wave equation $\Box u = F$ with zero data and $F$ be a smooth function on $\R^{1+3}$. Let $\mu, \nu>1$ be fixed constants. Provided the following right hand side is finite, $u$ satisfies the following inequality
		\[
		\langle t+|x|\rangle\langle t-|x|\rangle^{\mu-1}|u(t,x)|\lesssim \sup_{\tau\in[0,t]}\sup_{|x-y|\leq |t-\tau|} |y|\langle \tau+|y|\rangle^\mu \langle \tau-|y|\rangle ^\nu |F(\tau,y)|.
		\]
	\end{lemma}

	The idea of the proof of Proposition \ref{prp:Linf-Linf_wave} is to look at $u= \Gamma^J\hff_{\alpha\beta}$, for any fixed $J$ with $|J|=j\le N_1$, as the solution to a Cauchy problem of the following form
	\[
	\begin{cases}
		\Box_{tx}u = F\\
		(u, \partial_t u)|_{t=2} = (\phi, \psi)
	\end{cases}
	\]
	for some given smooth initial data $\phi, \psi$ and source term $F$, and to successively decompose $u$ as the sum of three waves $v_1, v_2, v_3$ such that
	\[
	\begin{cases}
		\Box_{tx}v_1 = \chi((r+1/2)/t)F \\
		(v_1, \partial_t v_1)|_{t=0} =(0,0)
	\end{cases} \quad\text{and}\qquad 
	\begin{cases}
		\Box_{tx}v_2 = (1-\chi((r+1/2)/t)) F \\
		(v_2, \partial_t v_2)|_{t=0} =(0,0)
	\end{cases}
	\]
	and $v_3$ is the solution to the homogeneous wave equation with data $(\phi,\psi)$.
	
	In the above systems, $\chi$ is a cut-off function equal to 1 on the ball $B_{1/2}(0)$ and supported in $\overline{B_1(0)}$, so that the source term in the equation of $v_1$ is supported in the interior of the cone $t= r+1/2$, while the source term in the equation of $v_2$ is supported in the portion of exterior region such that $t\le r+1/2$.
	
	The solutions $v_1, v_2, v_3$ are estimated using lemma \ref{lem:Alinhac},  \ref{lem:asakura} and \ref{lem:kata-yoko} respectively. The combination of such estimates will provide us with the desired estimate on $u=Z^K\hff_{\alpha\beta}$.
	
	
	Let us denote by $D^J_{\alpha \beta}$ the nonlinearity in equation \eqref{eq:diff_hff} satisfied by $\Gamma^J\hff_{\alpha\beta}$, i.e.
	\[
	D^{J, \flat}_{\alpha \beta} := - (H^{\bmmu\bmnu})^\flat\cdot \partial_\bmmu\partial_\bmnu  \Gamma^J\habf  +   F^{J, \flat}_{\alpha\beta} + F^{0,J}_{\alpha\beta} - \big( (H^{\mu\nu})^\natural\cdot \partial_\mu\partial_\nu  \Gamma^J\habn\big)^\flat.
	\]
	We start by estimating the source term $D^{J, \flat}_{\alpha\beta}$ in the interior of the cone $t = r+1/2$. Since the intersection of this cone with the exterior region $\dext$ is non-empty, we will make use of some estimates obtained in Section \ref{sec:exterior}.
	
	\begin{lemma}\label{lem:Dab-sourcing}
		For any multi-index $J$ with $|J|=k\le N_1$, any $s\in [s_0, S_0)$ and any $(t,x)$ with $t^2-r^2=s^2$ and $t>r+1/2$ we have that
		\[\aligned
		|D^{J, \flat}_{\alpha\beta}(t,x)|
		&\lesssim (C_2\epsilon)^2 t^{-1} s^{-2+\gamma_k}
		\endaligned	\]
	\end{lemma}
	
	\begin{proof}
		From the pointwise bounds \eqref{h0_estimate}, \eqref{KS1_ext} and \eqref{KS3_ext} we get that for any $(t,x,y)\in \dext$ such that $t>r+1/2$ and $t^2-r^2=s^2$
		\[
		\gathered
		\sum_{|J_1|+|J_2|\le k} |\partial \Gamma^{J_1}h \cdot \partial \Gamma^{J_2} h|\lesssim C_0^2\ep^2 t^{-2+2\sigma}(2+r-t)^{-2-2\kappa}\lesssim C_0^2\ep^2 t^{-1}s^{-2+4\sigma}\\
		\sum_{|J_1|+|J_2|\le k}|\Gamma^{J_1}H\cdot\partial^2 \Gamma^{J_2}h|\lesssim C_0^2 \ep^2 t^{-2+2\sigma}(2+r-t)^{-\frac32-2\kappa}\lesssim C_0^2\ep^2 t^{-1}s^{-2+4\sigma}.
		\endgathered
		\]
		
		In the interior region, 
		we recall the pointwise decay estimates already obtained in lemma \ref{lem:prod_h0} for the quadratic and cubic terms involving at least one $h^0$.
		
Turning next to the pure quadratic zero-mode interactions, we derive from \eqref{KS-hff1} that 
			\[
			\sum_{|J_1|+|J_2|\le k}  \big|\partial \Gamma^{J_1}\hff\cdot\partial \Gamma^{J_2}\hff\big| \lesssim C_1^2\ep^2 t^{-1}s^{-2+ 6 \sigma}
			\]
	and from \eqref{comm_int} and bounds \eqref{KS-hff1}, \eqref{KS-hff6} that
		\[
		\begin{aligned}
			\big|\Gamma^J\big((H^{1, \bmmu\bmnu})^\flat\cdot \partial_\bmmu\partial_\bmnu \hff_{\alpha\beta} \big)\big| \lesssim \sum_{\substack{ |J_1|+|J_2|\le k }} | \Gamma^{J_1}H^{1,\flat}_{LL}|\, |\partial^2_t \Gamma^{J_2}\hff_{\alpha\beta}| + \Big(\frac{s}{t}\Big)^2 |\Gamma^{J_1}H^{1,\flat}| \, |\partial^2_t \Gamma^{J_2}\hff_{\alpha\beta}| \\
			+ \sum_{\substack{ |J_1|+|J_2|\le j}} (1+t+r)^{-1}|\Gamma^{J_1}H^{1,\flat}| \, |\partial \Gamma^{J_2}\hff_{\alpha\beta}|\lesssim C_1^2\ep^2 t^{-1}s^{-2+9\sigma}.
		\end{aligned}
		\]

As concerns instead the pure interaction of the zero-average components, we derive from \eqref{bootKG_in} and the algebraic relation \eqref{condition_parameters} that
\[
\sum_{|J_1|+|J_2|=k} |(\partial \Gamma^{J_1}\hnn \cdot \partial \Gamma^{J_2}\hnn)^\flat| + \big|\big( (H^{\mu\nu})^\natural\cdot \partial_\mu\partial_\nu  \Gamma^J\habn\big)^\flat \big| \lesssim C_2^2\ep^2 t^{-3}s^{\gamma_k}.
\]
The conclusion of the proof follows from the fact that $9\sigma\ll \gamma_0$.
	\end{proof}

\begin{remark}\label{remark_lossW}
It is important to observe, in view of Proposition \ref{prpenhancedsup}, that the loss $s^{\gamma_k}$ in the above estimate for $D^J_{\alpha\beta}$ is only caused by the pure interactions of the zero-average components of the metric perturbations. All other interactions cause a smaller loss $s^{9\sigma}$.
\end{remark}

	
	We are now able to propagate the a-priori pointwise bound \eqref{bootW_in}.
	\begin{proposition}\label{prp:Linf-Linf_wave}
		There exist two constants $1\ll C_1\ll C_2$ sufficiently large, a finite and increasing sequence of parameters $0<\zeta_k, \gamma_k, \delta_k\ll 1$ satisfying \eqref{condition_parameters} and $0< \ep_0\ll 1$ sufficiently small such that, under the assumptions of Proposition \ref{prop:bootstrap_int}, the enhanced estimate \eqref{bootW_in_enh} is satisfied.
	\end{proposition}
	\begin{proof}
		We split $\Gamma^J\hff_{\alpha\beta}$ into the sum of three waves $v^J_{1, \alpha\beta}$, $v^J_{2, \alpha\beta}$ and $v^J_{3, \alpha\beta}$, where
		\[
		\begin{cases}
			\Box_{tx}v^J_{1,\alpha\beta} = \chi((r+1/2)/t) D^{J,\flat}_{\alpha\beta} \\
			(v^J_{1,\alpha\beta}, \partial_t v^J_{1,\alpha\beta})|_{t=2} =(0,0)
		\end{cases}
		\begin{cases}
			\Box_{tx}v^J_{2,\alpha\beta} =\big(1-\chi((r+1/2)/t)\big)D^{J, \flat}_{\alpha\beta} \\
			(v^J_{2,\alpha\beta}, \partial_t v^J_{2,\alpha\beta})|_{t=2} =(0,0)|_{t=2}
		\end{cases}
		\]			
		and
		\[
		\begin{cases}
			\Box_{tx}v^J_{3, \alpha\beta} = 0 \\
			(v^J_{3,\alpha\beta}, \partial_t v^J_{3,\alpha\beta})|_{t=2} =(\Gamma^J\hff_{\alpha\beta}, \partial_t \Gamma^J\hff_{\alpha\beta})|_{t=2}.
		\end{cases}
		\]
		
		From Lemma \ref{lem:Dab-sourcing}, we get that in the interior of the cone $t=r+1/2$ 
		\begin{equation}\label{Day}
			|D^{J, \flat}_{\alpha \beta}(t,x)|  \lesssim C_2^2 \epsilon^2 t^{-1}s^{-2+ \gamma_k}
			\lesssim C_2^2 \epsilon^2  t^{-2+\frac{\gamma_k}{2}} \langle t-r\rangle^{-1+\frac{\gamma_k}{2}}.
		\end{equation}
		Then, Lemma \ref{lem:Alinhac} with $\nu=2-\frac{\gamma_k}{2}$ and $\mu=1-\frac{ \gamma_k}{2}$ yields
		\begin{equation}\label{vay}
			|v^J_{\alpha y}(t,x)|\lesssim C_2^2\epsilon^2 (1+t)^{-1+\frac{ \gamma_k}{2}}\langle t-r\rangle^{\frac{{\gamma_k}}{2}}.
		\end{equation}
		
		As concerns the region exterior to the cone $t= r+1/2$, we recall the estimates \eqref{null_ext_point} for the quadratic null terms, \eqref{cub_ext_point} for the cubic interactions, \eqref{comm_ext_point} for the commutator terms and \eqref{weak_point_ext} for the weak null interactions. For any $(t,x)$ with $t<r+1/2$, we at least have
		\[
		|D^{J, \flat}_{\alpha\beta}(t,x)|\lesssim C_0^2\ep^2 t^{-2+2\sigma}(2+r-t)^{-\frac32}
		\]
		so that
		\[
		\sup_{\tau\in[0,t]}\sup_{|x-z|\leq |t-\tau|} |z|\langle \tau+|z|\rangle^\mu \langle \tau-|z|\rangle ^\nu |D^{J, \flat}_{\alpha\beta}(z,\tau)| \lesssim C_0^2\epsilon^2 \langle t + |x|\rangle^{3\sigma}
		\]
		provided that $\mu,\nu>1$ are chosen so that $\mu =1+\sigma$ and $1<\nu<3/2$. 
		From Lemma \ref{lem:kata-yoko} we then have 
		\[
		|v^J_{2,\alpha\beta}(t,x)|\lesssim  C_0^2\epsilon^2 \langle t+r\rangle^{-1+3\sigma}\langle t-r\rangle^{-\sigma}	.
		\]	
		
		Finally, the initial data satisfy
		\[
		|\Gamma^J \habf(2,x)|\lesssim C_0\ep (1+|x|)^{-1-\kappa}, \qquad |\nabla_{tx}\Gamma^J \habf(2,x)|\lesssim C_0 (1+|x|)^{-2-\kappa}
		\]
		as a consequence of the assumptions on the initial data \eqref{data_maintheo} and the pointwise bounds \eqref{KS1_ext} and \eqref{KS3_ext},
		so that Lemma \ref{lem:asakura} implies
		\[
		|v^J_{3, \alpha\beta}(t,x)|\lesssim \frac{C_0\epsilon}{(1+t + |x|)(1+|t-r|)^\kappa}.
		\]
		Summing all up, we find that there exists a constant $C>0$ such that
		\[
		|\Gamma^Jh^\flat_{\alpha\beta}(t,x)|\le C(C_0\ep + C_0^2\epsilon^2) \langle t+r\rangle^{-1+3\sigma}\langle t-r\rangle^{-\sigma} + CC_2^2\epsilon^2 \langle t+r\rangle^{-1+\frac{\gamma_k}{2}}\langle t-r\rangle^{\frac{\gamma_k}{2}} 
		\]
		so the result of the statement follows from the fact that $\sigma\ll\gamma_0$ and by choosing $C_2\gg 1$ sufficiently large so that $C(C_0+C_0\ep)< (C_2\ep)/2$ and $0<\ep_0\ll 1$ sufficiently small so that $CC_2\ep_0<1/2$.
	\end{proof}
	
Following Remarks \ref{remark_lossKG} and \ref{remark_lossW},	we conclude this section with enhanced pointwise bounds for the metric perturbation.
	\begin{proposition} \label{prpenhancedsup}
	There exists a constant $c>9$ such that the metric perturbation satisfies the following
			\begin{equation}\label{W_sigma_improved}
			\|t\,  \Gamma^J\habf\|_{L^\infty_x(\hin_s)}\lesssim
			C_2\ep s^{c\sigma}\quad   \text{if } |J|=k\le N_1-1, 
		\end{equation}
				\begin{equation}\label{KG_sigma_improved}
			\aligned
			\| t^\frac12s\, \partial_{tx} (\partial^I \Gamma^J \habn)\|_{L^\infty_xL^2_y(\hin_s)}  & + \| t^\frac32 \partial^{\le 1}_y (\partial^I \Gamma^J \habn)\|_{L^\infty_xL^2_y(\hin_s)}\\
			& \le 
			\begin{cases}
				2C_2\ep, \ &\text{if } |I|\le N_1, |J|=0\\
				2C_2\ep s^{c\sigma}, \ &\text{if } |I|+|J|\le N_1, |J|=k\le N_1-1.
			\end{cases}	
			\endaligned
		\end{equation}
	\end{proposition}
\begin{proof}	
The proof proceeds by induction. We assume that there exists a finite increasing sequence $c_k$ with $9\le c_k\ll c_{k+1}$ and $c_{i}+c_{j}<c_k$ whenever $i, j<k$, such that
\begin{equation*}
			\|t\,  \Gamma^J\habf\|_{L^\infty_x(\hin_s)}\lesssim
			2C_2\ep s^{c_k\sigma}  \quad \text{if } |J|=k\le N_1-1, 
		\end{equation*}
				\begin{equation*}
			\aligned
			\| t^\frac12s\, \partial_{tx} (\partial^I \Gamma^J \habn)\|_{L^\infty_xL^2_y(\hin_s)}  & + \| t^\frac32 \partial^{\le 1}_y (\partial^I \Gamma^J \habn)\|_{L^\infty_xL^2_y(\hin_s)}\\
			& \le 
			\begin{cases}
				2C_2\ep, \ &\text{if } |I|\le N_1, |J|=0\\
				2C_2\ep s^{c_k\sigma}, \ &\text{if } |I|+|J|\le N_1, |J|=k\le N_1-1
			\end{cases}	
			\endaligned
		\end{equation*}
We then have that, whenever $|I|+|J_1|+|J_2|\le N_1$ with $|J_1|=k_1>0$ and $|J_2|=k_2\le N_1-2$,
\[
|\Gamma^{J_1}\hff| \|\partial^2 \partial^I\Gamma^{J_2}\habn\|_{L^2_y}\lesssim C_2^2\ep^2 t^{-\frac32} s^{-1+(c_{k_1}+c_{k_2})\sigma}\lesssim C_2^2\ep^2 t^{-\frac32}s^{-1+c_k\sigma}.
\]	
The arguments in the proof of Proposition \ref{prp:KG_imp_decay} show that
\[
		\aligned
		s^\frac32\|\partial_y^{\le 1}(\partial^I\Gamma^J\hnn)\|_{L^2} + s^\frac12 \|\scal (\partial^I\Gamma^J\hnn)\|_{L^2} \le C\Big(\frac{s}{t}\Big)^\frac32 \enint(s_0, Z^{\le 2}(\partial^I\Gamma^J\hnn))^\frac12 \\
		+ C(C_1\ep +  C_2^2\ep^2) \Big(\frac{s}{t}\Big)^\frac32 {\magenta \begin{cases}
			1 \ &\text{if  } |I|\le N_1,\ |J|=0\\
			s^{c_k\sigma}, \ &\text{if  } |I|+|J|\le N_1+1, |J|=k\le N_1\\
		\end{cases}}
		\endaligned
		\]
and an appropriate choice of constants allows us to get \eqref{KG_sigma_improved}. 
Furthermore, whenever $|J_1|+|J_2|\le |J|\le N_1-1$
\[
\sum_{|J_1|+|J_2|=k} |(\partial \Gamma^{J_1}\hnn \cdot \partial \Gamma^{J_2}\hnn)^\flat| + \big|\big( (H^{\mu\nu})^\natural\cdot \partial_\mu\partial_\nu  \Gamma^J\habn\big)^\flat \big| \lesssim C_2^2\ep^2 t^{-3}s^{c_k\sigma}
\]
which implies, following the proof of Lemma \ref{lem:Dab-sourcing}, that
\[
|D^J_{\alpha\beta}|\lesssim (C_1^2+C_2^2)\ep^2 t^{-1}s^{-2+c_k\sigma}.
\]
Using lemma \ref{lem:Alinhac} with $\nu=2-\frac{c_k\sigma}{2}$ and $\mu=1-\frac{c_k\sigma}{2}$, we can then replace \eqref{vay} with
\[
|v^J_{\alpha y}(t,x)|\lesssim C_2^2\epsilon^2 (1+t)^{-1+\frac{{\magenta c_k}}{2}}\langle t-r\rangle^{\frac{{c_k}}{2}}
\]
and the arguments in the proof of Proposition \ref{prp:Linf-Linf_wave} yield \eqref{W_sigma_improved}.
\end{proof}

	\subsection{Propagation of the energy bounds.}
	In this section we propagate the a-priori energy bounds \eqref{boot1_in}-\eqref{boot3_in} and hence conclude the proof of Proposition \ref{prop:bootstrap_int}. 
	
	\begin{proof}[Proof of Proposition \ref{prop:bootstrap_int}]	We note first that, using bound \eqref{KG_sigma_improved} instead of \eqref{bootKG_in} and the fact that $\sigma\ll \zeta_0$ so that $c\sigma + \zeta_{j}<\zeta_k$ whenever $j<k$ in the proof of Lemma  \ref{lem:weak_null_int} and Proposition \ref{prp:commutators}, allows us to replace the loss $s^{\gamma_i + \zeta_{k-i}}$ for $i=\overline{1,4}$ in \eqref{weak_est4}, \eqref{comm_high} and \eqref{comm_zeromode} with $s^{\zeta_k}$, hence having
	\begin{multline}\label{weak_est4'}
			\left\| Z^K P^\flat_{\alpha\beta}(\partial h, \partial h)\right\|_{L^2_{xy}(\hin_s)}
			\lesssim C_1\epsilon s^{-1}\sum_{K'} \enint(s,Z^{K'}\hff)^{1/2} \\+ C_1\ep s^{-1+C\ep}\sum_{K''} \enint(s,Z^{K''}\hff)^{1/2}+C_1^2\ep^2s^{-\frac32+2\delta_N} 
	+C_1^2\ep^2\delta_{k>N_1} s^{-1+\zeta_k}
		\end{multline}		 
	for multi-index $K$ of type $(N, k)$,
	\begin{equation}\label{comm_high'}
	\iint_{\hin_{[s_0, s]}}\hspace{-10pt} \big|[Z^K, H^{1, \mu\nu}\partial_\mu \partial_\nu]\hab\big| |\partial_t Z^K\hab| dxdydt \lesssim (C_0^2+C_1^2)C_2^2\ep^3 s^{1+2\zeta_k}
	\end{equation}
	for multi-indexes $K$ of type $(N+1, k)$ with $k\le N$, and
			\begin{multline}
			\iint_{\hin_{[s_0, s]}} \big|\big([Z^K, H^{1, \mu\nu}\partial_\mu\partial_\nu]\hab\big)^\flat \big| |\partial_t Z^K\habf| dx dt \\ 
			+ \iint_{\hin_{[s_0, s]}} \big|\big((H^{1, \mu\nu})^\natural \cdot\partial_\mu\partial_\nu Z^K\habn\big)^\flat \big| |\partial_t Z^K\habf| dx dt \lesssim  (C_0^2+C_1^2)C_2^2\ep^3 s^{2\zeta_k} \label{comm_zeromode'}
		\end{multline}
		for multi-indexes $K$ of type $(N, k)$.
		
		For multi-indexes $K$ of type $(N+1,k)$, we substitute \eqref{boundary_term_int}, \eqref{null_int}, \eqref{cub_int}, \eqref{est_F0_int}, \eqref{weak_est1},   \eqref{int_H_energy}, \eqref{comm_high'}, together with the energy bound \eqref{bootin1}, into \eqref{en_ineq_hab_high} and hence deduce the existence of a constant $\tilde{C}>0$ such that for all $s\in [s_0, S_0)$
		\[
		\enint(s, Z^K h^1)\le\tilde{C}\big( \enint(s_0, Z^Kh^1) + \ep^2 s^{2\sigma + C\epsilon} + (C_0^2 + C_1^2)C_2^2\ep^3 s^{1+2\zeta_k}\Big).
		\]
		The enhanced energy estimate \eqref{boot1_in_enh} is obtained by picking $C_1\gg 1$ sufficiently large so that $3\tilde{C}\enint(s_0, Z^Kh^1)\le C_1^2\ep^2$ and $3\tilde{C}\le C_1$, and $0<\ep_0\ll 1$ sufficiently small so that $3( C_0^2+1) C_2^2\ep_0\le C_1$ and $2\sigma+C\epsilon<2\delta_k$. 
		
		Analogously, if $K$ is of type $(N, k)$ with $k\leq N_1$, we derive from \eqref{boundary_term_int}, \eqref{null_lower_int}, \eqref{cub_lower_int}, \eqref{est_F0_int}, \eqref{weak_est3}, \eqref{comm_low}, \eqref{int_H_energy_lower}, together with the energy bound \eqref{bootin4} that there exists another constant $\tilde{C}>0$ such that for all $s\in [s_0, S_0)$
		\[
		\enint(s, Z^K h^1)\le\tilde{C}\big( \enint(s_0, Z^Kh^1) + \ep^2 s^{2\sigma + C\epsilon} + (C_0^2 + C_1^2)C_2^2\ep^3 s^{2\delta_k}\Big).
		\]
		Choosing accordingly $C_1$ and $\ep_0$ yields \eqref{boot3_in_enh}. 
		
		Finally, for any multi-index $K$ of type $(N, k)$, we have the following estimate for the energy of $Z^K\habf$, which is obtained by plugging \eqref{boundary_term_int},  \eqref{cub_int}, \eqref{null_zeromode},  \eqref{est_F0_int},  \eqref{int_H_energy_flat}, \eqref{weak_est4'}, \eqref{comm_zeromode'} and the energy bound \eqref{bootin2} into \eqref{en_ineq_habf_high}
		\[
		\enint(s, Z^K \hff)\le\tilde{C}\big( \enint(s_0, Z^K\hff) + \ep^2 s^{2\sigma + C\epsilon} + (C_0^2 + C_1^2)C_2^2\ep^3 s^{2\zeta_k}\Big),
		\]
		for a new constant $\tilde{C}$. Again, the enhanced energy bound \eqref{boot2_in_enh} follows by choosing $C_1, \ep_0$ appropriately. 
		
	\end{proof}

	\appendix
	\section{Energy Inequalities} \label{sec:energy_appendix}
	
	In this section we group together different energy inequalities that are useful in the paper. 
	We denote by $\Wf$ the solution of the following linear inhomogeneous wave equation
	$$g^{\mu \nu}\partial_\mu \partial_\nu \Wf =\Ff$$
	which can be also written
	\begin{equation}\label{linearized_wave}
		(-\partial^2_t + \Delta_x +\partial^2_y)\Wf + H^{\mu\nu}\partial_\mu\partial_\nu \Wf = \Ff, \qquad (t,x,y)\in \R^{1+3}\times \S^1
	\end{equation}
	where the tensor components $H^{\mu\nu}$ are assumed to be sufficiently small functions and $\Ff$ is some source term.
	
	In this section, a particular attention will be given to the energy flux on hyperboloids. These are spacelike hypersurfaces in Minkowski spacetime, but have a degeneracy caused by the fact that they are asymptotically null. This is something which could be destroyed by perturbations of the metric. We take advantage of the Schwarzschild component of $H$, introduced in \eqref{dec_H}, to show that the hyperboloids remain spacelike everywhere.

	\begin{proposition}[Exterior energy inequality] \label{prop:exterior_energy}
		Let $\Wf$ be a solution of equation \eqref{linearized_wave} decaying sufficiently fast as $|x|\rightarrow \infty$ and assume that there exist $\epsilon>0$ small such that tensor $H$ satisfies the following bounds
		$$|H(t,x,y)|\lesssim \frac{\ep}{(1+t+r)^\frac{3}{4}}, \quad \Big|H_{LL}+\chi\left(\frac{r}{t}\right)\chi(r)\frac{2M}{r} \Big| \lesssim \frac{\ep}{(1+t+r)^{1+\delta}},$$
		where $\chi$ is a cut-off function such that $\chi(s) = 0$ for $s\le 1/2$, $\chi(s)=1$ for $s\ge 3/4$ and $\delta$ is any fixed positive constant.
		Let $w(q)$ be a smooth function that only depends on the distance $q = r-t$ from the light cone and such that $w(q), w'(q)\ge 0$. For any $2\le t_1< t_2$, let $\tilde{\hcal}_{t_1 t_2}$ denote the portion of $\tilde{\hcal}$ in the time interval $[t_1, t_2]$ and $d\mu_{\hcal}$ be its surface element. We have the following inequality
		\begin{equation}\label{energy_ineq_ext}
			\begin{aligned}
				& \int_{\Sext_{t_2}} w(q) |\nabla_{txy}\Wf|^2 dxdy + \int_{\tilde{\hcal}_{t_1t_2}}  w(q)\Big[\frac{1}{2( 1+r^2)} +\chi\left(\frac{r}{t}\right)\chi(r)\frac{M}{2r} \Big] |\partial_t \Wf|^2 + w(q)|\underline{\tilde{\nabla}}\Wf|^2 dxdy \\
				& + \iint_{\dext_{[t_1, t_2]}} w'(q)\big(|L\Wf|^2  + |\slashed \nabla \Wf|^2\big) dtdxdy\lesssim \int_{\Sext_{t_1}} w(q) |\nabla_{txy}\Wf|^2 dxdy \\
				& + \iint_{\dext_{[t_1, t_2]}} w(q)|(\Ff + \partial_\mu H^{\mu\nu} \, \partial_\nu\Wf)\partial_t \Wf| + w(q)|\partial_t H^{\mu\nu} \, \partial_\mu\Wf \, \partial_\nu\Wf |\, dtdxdy \\
				& + \iint_{\dext_{[t_1, t_2]}} w'(q) | H^{\rho\sigma}\partial_\rho \Wf \partial_\sigma\Wf|\ dtdxdy  + \iint_{\dext_{[t_1, t_2]}}\hspace{-10pt} w'(q) |(-H^{0\nu} + \omega_\bmj H^{\bmj\nu})\partial_\nu \Wf \partial_t\Wf| \ dtdxdy
			\end{aligned}
		\end{equation}
		where $\underline{\tilde{\nabla}} = (\tilde{\pb}_1, \dots, \tilde{\pb}_4)$ is the tangent gradient to $\tilde{\hcal}$, i.e. $\tilde{\pb}_{\bmi} = \partial_{\bmi} + \frac{x^{\bmi}}{t-1}\partial_t$ for $\bmi=1,2,3$ and $\tilde{\pb}_4 = \partial_y$, and $\omega_\bmj = x_\bmj/r$.
		
	\end{proposition}
	\begin{proof}
		We start with the following computation
		\begin{equation}
			\begin{split}
				\partial_t \Wf\left(g^{\mu \nu} \partial_\mu \partial_\nu \Wf\right)=&\partial_\mu \left(g^{\mu\nu}\partial_\nu\Wf\partial_t\Wf\right)-\frac{1}{2}\partial_t\left(g^{\mu \nu}\partial_\mu \Wf\partial_\nu\Wf\right)\\
				&-(\partial_\mu H^{\mu\nu})\partial_\mu \Wf\partial_t \Wf +\frac{1}{2}(\partial_tH^{\mu \nu})\partial_\mu \Wf\partial_\nu\Wf.
			\end{split}
		\end{equation}
		Multiplying by $w(q)$ we obtain
		\begin{equation}\label{calcw}
			\begin{split}
				w(q)\partial_t \Wf\left(g^{\mu \nu} \partial_\mu \partial_\nu \Wf\right)=&\partial_\mu \left(w(q)g^{\mu\nu}\partial_\nu\Wf\partial_t\Wf\right)-\frac{1}{2}\partial_t\left(w(q)g^{\mu \nu}\partial_\mu \Wf\partial_\nu\Wf\right)\\
				&-w(q)(\partial_\mu H^{\mu\nu})\partial_\mu \Wf\partial_t \Wf +\frac{1}{2}w(q)(\partial_tH^{\mu \nu})\partial_\mu \Wf\partial_\nu\Wf\\
				&-(\partial_\mu w(q))\left(g^{\mu\nu}\partial_\nu\Wf\partial_t\Wf\right)+\frac12(\partial_t w(q))\left(g^{\mu \nu}\partial_\mu \Wf\partial_\nu\Wf\right)
			\end{split}
		\end{equation}
		
		We integrate \eqref{calcw} in the spacetime portion of the exterior region included between the two spacelike hypersurfaces $\Sext_{t_1}$ and $\Sext_{t_2}$, denoted by $\dext_{[t_1, t_2]}$.
		We treat the divergence term via Stokes' theorem, meaning
		$$\int_{\dext_{[t_1, t_2]}} (\partial_\mu X^\mu) dtdxdy = \int_{\Sigma^e_{t_2}} X^0 dxdy -\int_{\Sigma^e_{t_1}} X^0 dxdy
		+\int_{\tilde{\hcal}_{t_1 t_2}}\Big(X^0-\frac{x_\bmi}{t-1}X^\bmi\Big)dxdy.$$
		Applying this to \eqref{calcw} we obtain
		$$ 
		\begin{aligned}	
			& \int_{\Sigma^e_{t_2}} w(q)e^{curv}_{\Sigma} dxdy +\int_{\tilde{\hcal}_{t_1 t_2}} w(q)e^{curv}_{\tilde{\hcal}}dxdy+\iint_{\dext_{[t_1, t_2]}} w'(q) B\, dtdxdy
			\\ & = \int_{\Sigma^e_{t_1}}w(q) e^{curv}_{\Sigma} dx + \iint_{\dext_{[t_1, t_2]} }w(q) (C - \Ff)dxdt,
		\end{aligned}
		$$
		where the curved energy densities are defined by
		\begin{align}
			e^{curv}_{\Sigma} &= -g^{\mu 0}\partial_\mu \Wf\partial_t \Wf +\frac{1}{2}g^{\mu\nu}\partial_\mu \Wf\partial_\nu \Wf\\
			e^{curv}_{\tilde{\hcal}}&= -g^{\mu 0}\partial_\mu \Wf\partial_t \Wf +\frac{1}{2}g^{\mu\nu}\partial_\mu \Wf\partial_\nu \Wf+\frac{x_\bmi}{t-1}g^{\bmi\mu} \partial_\mu \Wf\partial_t \Wf \label{density_eH}
		\end{align}
		and the bulk terms are given by
		\begin{align}
			B&=-(\partial_\mu q)\left(g^{\mu\nu}\partial_\nu\Wf\partial_t\Wf\right)+\frac12(\partial_t q)\left(g^{\mu \nu}\partial_\mu \Wf\partial_\nu\Wf\right)\\
			C&=-(\partial_\mu H^{\mu\nu})\partial_\mu \Wf\partial_\nu \Wf +\frac{1}{2}(\partial_tH^{\mu \nu})\partial_\mu \Wf\partial_\nu\Wf. \label{def_C}
		\end{align}
		We have
		$$	e^{curv}_{\Sigma} =\frac{1}{2}\left((\partial_t \Wf)^2+|\nabla_x \Wf|^2+(\partial_y \Wf)^2\right)+ O\left(H |\nabla \Wf|^2\right),$$
		so with the hypothesis $|H|\leq \frac{1}{100}$ we easily obtain
		$$\frac{1}{4}\left((\partial_t \Wf)^2+|\nabla_x \Wf|^2+(\partial_y \Wf)^2\right)\leq e^{curv}_{\Sigma}\leq 4\left((\partial_t \Wf)^2+|\nabla \Wf_x|^2+(\partial_y \Wf)^2\right).$$
		We have to be a little more careful with $	e^{curv}_{\tilde{\hcal}}$.  We have
		\begin{align*}
			e^{curv}_{\tilde{\hcal}}=&\frac{1}{2}\left((\partial_t \Wf)^2+|\nabla_x \Wf|^2+(\partial_y \Wf)^2\right)+\frac{r}{t-1}\partial_r \Wf\partial_t \Wf
			-\frac{1}{4}H_{LL}(\partial_t \Wf)^2 \\
			&+O(H\cdot \partial \Wf\cdot \tilde{\pb}\Wf)+\left(1-\frac{r}{t-1}\right)O(H |\nabla \Wf|^2)+ O(H|\underline{\tilde{\nabla}}\Wf |^2)\\
			=&\frac{1}{2}\left(\frac{1}{(t-1)^2}(\partial_t \Wf)^2 +\sum_{\bmi=1}^3 \left(\partial_\bmi \Wf+\frac{x_\bmi}{t-1}\partial_t \Wf\right)^2+(\partial_y \Wf)^2\right)	-\frac{1}{4}H_{LL}(\partial_t \Wf)^2\\
			&+O(H\cdot\partial \Wf\cdot \tilde{\pb}\Wf)+\left(1-\frac{r}{t-1}\right)O(H |\nabla \Wf|^2)+ O(H|\underline{\tilde{\nabla}}\Wf |^2)\\
		\end{align*}
		We note that on $\hcal $ we have $(t-1)^2=1+r^2$.
		Using the decomposition \eqref{dec_H} of $H$ we write
			\[
			\left(\frac{1}{2(t-1)^2}-\frac{1}{4}H_{LL}\right)(\partial_t \Wf)^2  = 	\left(\frac{1}{2(1+r^2)}+\chi\left(\frac{r}{t}\right)\chi(r)\frac{M}{2r} -\frac{1}{4} H^1_{LL}\right) (\partial_t \Wf)^2
			\]
so that			
		\begin{align*}
			&\left(\frac{1}{2(t-1)^2}-\frac{1}{4}H_{LL}\right)(\partial_t \Wf)^2 
			+O(H \cdot \partial\Wf\cdot \pb\Wf)+ \left(1-\frac{r}{t-1}\right)O(H |\nabla \Wf|^2)\\
			&=\left(\frac{1}{2(1+r^2)}+\chi\left(\frac{r}{t}\right)\chi(r)\frac{M}{2r} +O(H^1_{LL})+O(\ep^{-\frac{1}{2}}|H|^2)\right) (\partial_t \Wf)^2 \\
			& + \left(1-\frac{r}{t-1}\right)O(H(\partial \Wf)^2)+O(\ep^\frac{1}{2} |\underline{\nabla}\Wf\nabla \Wf |^2).
		\end{align*}
Under the hypothesis 
			\[
			|H^1_{LL}|\lesssim \frac{\epsilon}{(1+t+r)^{1+\delta}}
			\]
			we obtain that for not too large values of $r$ (e.g. $r\ll 1/(2\ep)$), $H^1_{LL}$ is small in front of $\frac{1}{2(1+r^2)}$, while for
			$r\gtrsim 1/(2\ep)$ it is small compared to the then dominant term $\frac{M}{2r}$.
		
		Under the hypothesis on $H$ we obtain
		$$|h^1_{LL}+\ep^{-\frac{1}{2}}|H|^2| \lesssim \frac{\ep}{(1+r)^{-1-\delta}}.$$
		Consequently for not too large values of $r$ (e.g. $r\ll 1/(2\ep)$) we have that
		\[
		|h^1_{LL}+\ep^{-\frac{1}{2}}|H|^2| \le\frac{1}{100(1+r^2)};
		\]
		on the other hand, when $r\gtrsim 1/(2\ep)$ 
		the dominant term is $M/(2r)$ and for $\ep$ sufficiently small
		\[
		|h^1_{LL}+\ep^{-\frac{1}{2}}|H|^2| \le \frac{M}{100r}
		\]
		Consequently, we can bound
		\begin{multline*}
			\frac{1}{4}\left(\left(\frac{1}{2(1+r^2)} +\chi\left(\frac{r}{t}\right)\chi(r)\frac{M}{2r} \right) |\partial_t \Wf|^2 + |\underline{\tilde{\nabla}}\Wf|^2 \right) \\
			\leq e^{curv}_{\hcal} \leq 4\left(\left(\frac{1}{2(1+r^2)} +\chi\left(\frac{r}{t}\right)\chi(r)\frac{M}{2r} \right) |\partial_t \Wf|^2 + |\underline{\tilde{\nabla}}\Wf|^2 \right) .
		\end{multline*}
			Finally, a simple computation shows that
			\[
			B = |L\Wf|^2 + |\slashed \nabla \Wf|^2 + \frac12 H^{\mu\nu}\partial_\mu\Wf \cdot\partial_\nu\Wf + \Big(-H^{0\nu} + \frac{x^\bmi}{r}H^{\bmi\nu}\Big)\partial_\nu \Wf \partial_t\Wf
			\]

	\end{proof}
	
	\begin{proposition}[Energy inequality on hyperboloids] \label{prop:energy_ineq_hyperb}
		Let $\Wf$ be a solution of \eqref{linearized_wave} and $\enint(s, \Wf)$ be the energy functional defined in \eqref{interior_energy_functional}. We assume that $H$ satisfies the same hypothesis of proposition \ref{prop:exterior_energy}
		For any $2<s_1<s_2$, let $\tilde{\hcal}_{s_1s_2}$ denote the portion of $\tilde{\hcal}$ bounded by the hyperboloids $\hin_{s_i}$ with $i=1,2$. We have the following inequality
		\[
		\begin{aligned}
			\enint(s_2, \Wf) & \lesssim \enint(s_1, \Wf) + \int_{\tilde{\hcal}_{s_1s_2}}\Big[\frac{1}{2(1+r^2)}+\chi\left(\frac{r}{t}\right)\chi(r)\frac{M}{2r}\Big] |\partial_t \Wf|^2 + |\underline{\tilde{\nabla}}\Wf|^2 dxdy \\
			& + \iint_{\hin_{[s_1, s_2]}} |(\Ff + \partial_\mu H^{\mu\nu} \, \partial_\sigma \Wf)\partial_t \Wf| +  |\partial_t H^{\mu\nu} \, \partial_\mu\Wf \, \partial_\nu\Wf| \, dtdxdy.
		\end{aligned}
		\]
		The implicit constant in the above inequality is a universal constant.\\
		An analogue inequality holds true for solutions $\Wf$ to \eqref{linearized_wave} on $\R^{1+3}$.
	\end{proposition}
	\begin{proof}
		The proof is analogous to that of proposition \ref{prop:exterior_energy} except that we integrate \eqref{calcw} with $w\equiv 1$ in the portion of interior region bounded above by $\hin_{s_2}$, below by $\hin_{s_1}$ and laterally by $\tilde{\hcal}_{s_1s_2}$, which we denote by $\hin_{[s_1, s_2]}$.
		This yields
		$$
		\int_{\hin_{s_2}}  \mathbf{e}^{curv}_{\hin}dxdy= \int_{\hin_{s_1}} \mathbf{e}^{curv}_{\hin}dxdy+\int_{\hcal_{s_1s_2}} e^{curv}_{\hcal}dxdy +\iint_{\hin_{[s_1,s_2]}}(C-\Ff)dtdxdy,$$
		where $e^{curv}_{\hcal}$ and $C$ have been defined in \eqref{density_eH} and \eqref{def_C} respectively
		and
		$$\mathbf{e}^{curv}_{\hin} =-g^{\mu 0}\partial_\mu \Wf\partial_t \Wf +\frac{1}{2}g^{\mu\nu}\partial_\mu \Wf\partial_\nu \Wf+\frac{x_\bmi}{t}g^{\bmi\mu} \partial_\mu \Wf\partial_t \Wf .$$
		In the region $ \hin_{[s_1,s_2]}$ we have
		$$\frac{1}{4}\left(\frac{s^2}{t^2} |\partial_t \Wf|^2 + |\underline{\nabla}\Wf|^2 \right) \leq \mathbf{e}^{curv}_{\hcal} \leq 4\left(\frac{s^2}{t^2} |\partial_t \Wf|^2  + |\underline{\nabla}\Wf|^2 \right) .$$
		In fact, we have that $r\leq t-1$ and the mass term can be absorbed in the following way
		$$\chi\left(\frac{r}{t}\right) \chi(r)\frac{M}{r}\leq \frac{(t-r)(t+r)}{100(t+r)^2} =  \frac{s^2}{100t^2}.$$

		%

	\end{proof}
	
	\section{Sobolev and Hardy inequalities} \label{sec:weighted_Sobolev_Hardy}
	
	We start by listing some weighted inequalities that are used in section \ref{sec:exterior}. Their proofs can be found in Huneau-Stingo \cite{HS}. 
	
	\begin{lemma}[Weighted Sobolev inequalities]\label{lemma_wsi}
		Let $\beta\in\mathbb{R}$. For any sufficiently smooth function $u$ we have the following inequalities
		\begin{multline}\label{sobolev1_ext}
			\sup_{\Sext_t}\, (2+r-t)^{2\beta} r^2|u(t,x,y)|^2  \\
			\lesssim\iint_{\Sext_t} (2+r-t)^{1+2\beta}(\partial_r Z^{\le 2} u)^2
			+(2+r-t)^{2\beta-1}( Z^{\le 2} u)^2 \,	dxdy,
		\end{multline}
		\begin{equation}\label{sobolev2_ext}
			\sup_{\Sext_t}\, (2+r-t)^{2\beta} r^2|u(t,x,y)|^2 \lesssim  \iint_{\Sext_t} (2+r-t)^{2\beta}\Big((\partial_r Z^{\le 2}  u)^2
			+( Z^{\le 2} u)^2\Big)dxdy,
		\end{equation}
		\begin{equation}\label{sobolev3_ext}
			\begin{split}
				& \sup_{\Sext_t}\, (2+r-t)^{2\beta} r^2\|u(t,r)\|^2_{L^2(\m S^2 \times \m S^1)} \lesssim\iint_{\Sext_t} (2+r-t)^{1+2\beta}(\partial_r u)^2
				+(2+r-t)^{-1+2\beta} u^2 \,	dxdy.
			\end{split}
		\end{equation}
	\end{lemma}

	\begin{lemma}[Weighted Hardy inequality]\label{lem_hardy}
		Let $\beta>-1$. For any sufficiently regular function $u$ for which the left-hand side of the following inequality is finite we have
		\begin{equation}\label{hardy_ext}
			\iint_{\Sext_t} (2+r-t )^\beta u^2 dxdy \lesssim \iint_{\Sext_t}(2+r-t)^{\beta+2}(\partial  u)^2 dxdy.
		\end{equation}
		
	\end{lemma}

	\begin{corollary}
		Let $\beta >0$. For any sufficiently regular function $u$ we have the following inequalities
		\begin{gather}
			(2+r-t)^\beta r |u(t,x,y)|\lesssim \| (2+r-t)^{1/2 + \beta} \partial Z^{\le 2}u(t)\|_{L^2(\Sext_t)} \label{sob+hardy1_ext}\\
			(2+r-t)^\beta r \|u(t,r)\|_{L^2(\S^2\times \S^1)}\lesssim \| (2+r-t)^{1/2 + \beta} \partial u(t)\|_{L^2(\Sext_t)} \label{sob+hardy2_ext}
		\end{gather}
		
	\end{corollary}
	
	\begin{proof}
		Inequality \eqref{sob+hardy1_ext} (resp. \eqref{sob+hardy2_ext}) is a straight consequence of the combination of \eqref{sobolev1_ext} (resp. \eqref{sobolev3_ext}) and \eqref{hardy_ext}.
	\end{proof}

	Below are some Sobolev and Hardy inequalities that are useful in section \ref{sec:interior}. Lemma \ref{lem:standards_Sob_hyp} is standard while lemma \ref{lm:ksobolev} is a simple adaptation of a result in \cite{FWY21}. The result of lemma \ref{lem:Hardy_LR} can be also obtained with small modifications from the one in \cite{LR10}.
	
	\begin{lemma}\label{lem:standards_Sob_hyp}
		For any sufficiently smooth function $u$ we have the following Sobolev inequality
		\[
		\| u\|_{L^p(\hin_s)}\lesssim \|\underline{\nabla}u\|^{\frac{3}{2}-\frac{3}{p}}_{L^2(\hin_s)}\|u\|^{\frac{3}{p}-\frac{1}{2}}_{L^2(\hin_s)} + s^{-\big(\frac{3}{2}-\frac{3}{p}\big)}\| u\|_{L^2(\hin_s)}, \quad 2\le p\le 6
		\]
		as well as the trace inequality
		\[
		\|u\|_{L^4(S_{s,r})}\lesssim \|\underline{\nabla}u\|_{L^2(\hin_s)} + s^{-1}\|u\|_{L^2(\hin_s)}.
		\]
	\end{lemma}
	
	\begin{lemma}\label{lm:ksobolev}
		Let $B=\{\Omega_{0j}:j=1,2,3\}$.
		For any sufficiently smooth function $u=u(t,x)$ we have
		\[
		\sup_{\hin_s} |t^\frac{3}{2} u|\lesssim \|B^{\le 2}u\|_{L^2_x(\hin_s)}.
		\]
	\end{lemma}

	\begin{lemma}\label{lem:classical_Hardy}
		Let $s>0$, $r_s:=\max\{r\,|\, S_{r}\subset \hin_s\}$ and $t_s = \sqrt{s^2+r_s}$. For any sufficiently smooth function $u=u(t,x)$ we have that
		\begin{equation}\label{Hardy_classical}
			\| r^{-1}u\|_{L^2(\hin_s)}\lesssim \|\pb u\|_{L^2(\hin_s)} + \|\partial u(t_s)\|_{L^2(\Sigma_{t_s})}
		\end{equation}
		\begin{proof}
			It is a straightforward consequence of the classical Hardy inequality applied to 
			\[
			v(x) = \begin{cases}
				u(\sqrt{s^2+r^2}, x), \quad & \text{if } |x|<r_s \\
				u(t_s, x), \quad & \text{if } |x|>r_s.
			\end{cases}
			\]
		\end{proof}

	\end{lemma}
	
	\begin{lemma} \label{lem:Hardy_LR}
		Let $0\le\alpha\le 2$, $1+\mu>0$ and $\gamma>0$. For any function $u\in \q C^\infty_0([0,\infty))$, any arbitrary time $t> 0$ and $s>0$ there exists a constant $C$, depending on a lower bound for $\gamma$ and $1+\mu$, such that
		\begin{equation}\label{hardy_hyp_LR}
			\begin{aligned}
				& \int_{r(s,t)}^t\frac{u^2}{(1+t-r)^{2+\mu}}\, \frac{r^2dr}{(1+t+r)^\alpha}  + \int_{t}^\infty \frac{u^2}{(1+r-t)^{1-\gamma}}\frac{r^2 dr}{(1+t + r)^\alpha} \\
				& \le C \int_{r(s,t)}^t \frac{|\partial_r u|^2}{(1+t-r)^\mu} \frac{r^2 dr}{(1+t +r)^\alpha}\,  + C \int_t^\infty  |\partial_r u|^2 \frac{(1+r - t)^{1+\gamma}}{(1+t+r)^\alpha}\ r^2dr
			\end{aligned}
		\end{equation}
		where $r(s,t) = \sqrt{(t^2-s^2)^+}$.
	\end{lemma}

	\begin{corollary}\label{cor:Hardy}
		Under the same assumptions of Lemma \ref{lem:Hardy_LR}, we have that 
		\begin{multline*}
			\int_{s_0}^{t_s}\hspace{-3pt}\int_{\q C_t} \frac{|u|^2}{(1+t-r)^{2+\mu}} \frac{dxdt}{(1+t+r)^{\alpha}}
			\\
			\lesssim \int_{s_0}^s\hspace{-3pt} \int_{\hin_\tau} \frac{|\partial_r u|^2}{(1+t(\tau)-r)^\mu(1+t(\tau)+r)^\alpha} dxd\tau + \int_{s_0}^{t_s}\hspace{-3pt}\int_{\Sext_t}|\partial_r u|^2 \frac{(1+|r-t|)^{1+\gamma}}{(1+t+r)^\alpha}dxdt
		\end{multline*}
	\end{corollary}
	\begin{proof}
		The proof is a simple application of inequality \eqref{hardy_hyp_LR} and of a change of coordinates.
	\end{proof}

	\bibliographystyle{abbrv}
	\bibliography{Biblio_Anna}{}	
\end{document}